\documentclass[a4paper]{article}
\advance\textwidth by 3cm
\advance\oddsidemargin by -1.5cm
\advance\textheight by 3cm
\advance\topmargin by -3cm
\newif\iffull\fulltrue
\newif\ifdraft\draftfalse
\usepackage[numbers]{natbib}
\usepackage{booktabs}
\usepackage{tablefootnote}
\usepackage{bcprules}
\usepackage{multicol}
\usepackage{comment}
\usepackage{paralist}
\usepackage{stmaryrd}
\usepackage{algorithm2e}
\usepackage{hyperref}
\usepackage{subcaption}
\usepackage{amsthm,amsmath,amssymb}
\usepackage{url,graphicx}
\newtheorem{theorem}{Theorem}[section]

\newtheorem{lemma}[theorem]{Lemma}

\theoremstyle{definition}
\newtheorem{definition}{Definition}[section]
\newtheorem{remark}{Remark}[section]
\newtheorem{example}{Example}[section]

\ifdraft
\newcommand\nk[1]{\textcolor{red}{[#1 -nk]}}
\newcommand\ry[1]{\textcolor{blue}{[#1 -ry]}}
\newcommand\ryoshi[1]{\textcolor{teal}{[#1 -ryoshi]}}
\else
\newcommand\nk[1]{}
\newcommand\ry[1]{}
\newcommand\ryoshi[1]{}
\fi
\allowdisplaybreaks
\AtBeginDocument{}

\newcommand\arity{\mathit{arity}}
\newcommand\mergeop{\star}
\newcommand\Merge{\mergeop}
\newcommand\Sorted{\mathit{sorted}}
\newcommand\nsv{\mathit{nsv}}
\newcommand\sv[1]{#1^S}
\newcommand\CHC{\mathcal{S}}
\newcommand\CHCone{\mathcal{C}}
\newcommand\pmodel{\mathcal{M}}
\newcommand\Reva{\mathit{Reva}}
\newcommand\form{\varphi}
\newcommand\cform{\form}
\newcommand\oracle{g_{\mathcal{F}}}
\newcommand\COL{\mathbin{:}}

\newcommand\dom{\mathit{dom}}

\newcommand\residual[3]{\residualfun{#2}{#3}(#1)}
\newcommand\residualfun[2]{\langle #1\pred #2\rangle}
\newcommand\tupinf{\textsc{TupInf}}
\newcommand\chocolat{\textsc{CHoCoL}}
\newcommand\hoice{HoIce}
\newcommand\ringen{\textsc{RInGen}}
\newcommand\spacer{\textsc{Spacer}}
\newcommand\racer{\textsc{Racer}}
\newcommand\eld{\textsc{Eldarica}}
\newcommand\Pred[1]{\mathit{#1}}
\newcommand\tpnew{\tp_{\mathit{new}}}
\newcommand\tpempty{\tp_\emptyset}
\newcommand\enumM{\mathtt{enumModel}}

\newcommand\STPinf{\mathtt{STPinf}}
\newcommand\CSTPs{\mathit{CSTPs}}
\newcommand\CSTPnorm{\mathit{normalize}}
\newcommand\Plus{\mathit{Plus}}
\newcommand\None{\mathit{None}}
\newcommand\false{\mathit{false}}
\newcommand\true{\mathit{true}}
\newcommand\catalia{\textsc{Catalia}}
\newcommand\len[1]{|#1|}
\newcommand\T[1]{\texttt{#1}}
\newcommand\Nil{[\,]}

\newcommand\imp{\Rightarrow}
\newcommand\IFF{\Leftrightarrow}

\newcommand\measure[1]{\#{#1}}
\newcommand\tpequiv{\sim}
\newcommand\shapeof[1]{\mathit{shape}(#1)}
\newcommand\sizeof[1]{\mathit{size}(#1)}
\newcommand\rowsof[1]{\mathit{rows}({#1})}
\newcommand\colsof[1]{\mathit{dim}({#1})}

\newcommand\Vars{\mathcal{V}}
\newcommand\Alpha{\Sigma}

\newcommand\pat{p}
\newcommand\patalt{q}
\newcommand\rpat[1]{#1^R}
\newcommand\rev[1]{#1^R}

\newcommand\tp{t}
\newcommand\ctp{\gamma}
\newcommand\CSTPmodel{\Gamma}
\newcommand\cstp{\ctp}
\newcommand\ctpset{C}
\newcommand\TPinf{\mathtt{TPinf}}
\newcommand\CTPinf{\mathtt{CTPinf}}
\newcommand\CMTPinf{\mathtt{CMTPinf}}
\newcommand\CSTPinf{\CTPinf}
\newcommand\TP{T}
\newcommand\Dat{\Theta}
\newcommand\Subst[2]{[#2/#1]}
\newcommand\dat{M}
\newcommand\Mat[2]{\mathcal{M}_{m,n}}
\newcommand\red{\longrightarrow}
\newcommand\nred{\;\mathbin{\ \not\!\!\red}}
\newcommand\npred{\mathbin{\ \not\!\pred}}
\newcommand\remove[2]{#1\mathbin{\uparrow_{#2}}}
\newcommand\replace[3]{#1\set{#2\mapsto #3}}
\newcommand\seq[1]{\widetilde{#1}}
\newcommand\Tp{\mathcal{T}}
\newcommand\reds{\red^*}
\newcommand\preds{\pred^*}
\newcommand\pred{\leadsto}
\newcommand\rpred{\mbox{\rotatebox[origin=c]{180}{$\leadsto$}}}
\newcommand\rpreds{\rpred^*}
\newcommand\FV{\mathbf{Vars}}
\newcommand\lang{\mathcal{L}}
\newcommand\ldiff[2]{#1\backslash #2}
\newcommand\rdiff[2]{#2/#1}
\newcommand\smodels{\models_s}
\newcommand\leftinv[2]{#1\backslash#2}
\newcommand\set[1]{\{#1\}}

\title{Solvable Tuple Patterns and Their Applications\\ to Program Verification}
\author{Naoki Kobayashi\thanks{The University of Tokyo} \and Ryosuke Sato\thanks{Tokyo University of Agriculture and Technology} \and Ayumi Shinohara\thanks{Tohoku University} \and Ryo Yoshinaka\footnotemark[3]}

\begin{document}

\maketitle

\begin{abstract}
  Despite the recent progress of automated program verification techniques,
  fully automated verification of programs manipulating recursive data structures remains a challenge.
  We introduce solvable tuple patterns (STPs) and conjunctive STPs (CSTPs), novel formalisms
  for expressing and inferring invariants between list-like recursive data structures.
  A distinguishing feature of STPs is that they can be efficiently inferred from only a small number of positive samples; no negative
  samples are required. 
  After presenting properties and inference algorithms of STPs and CSTPs,
  we show how to incorporate the CSTP inference 
  into a CHC (Constrained Horn Clauses) solver supporting list-like data structures, which serves as a uniform backend for
  automated program verification tools. A CHC solver incorporating the (C)STP inference
  has won the ADT-LIN category of CHC-COMP 2025 by a significant margin.
\end{abstract}
 \section{Introduction}
\label{sec:intro}

\subsection{Background}
Although a lot of progress has recently been made in
automated program verification techniques,
fully automated verification of programs involving algebraic data types (ADTs)
such as lists and trees remains difficult.
In the case of programs manipulating only integers, key invariants on loops and recursive functions can often be expressed by linear integer arithmetic, and
good techniques (such as interpolants~\cite{DBLP:conf/popl/HenzingerJMM04,DBLP:journals/tcs/McMillan05} and counterexample-guided abstraction refinement~\cite{Clarke2003a})
for automatically finding them have been developed. In contrast,
in the case of programs manipulating ADTs, key invariants often involve
inductive predicates, which are difficult to find automatically.

To see the difficulty of automated verification of programs involving ADTs, let us
consider the following OCaml program:\footnote{This example has been taken from \url{https://github.com/chc-comp/chc-comp24-benchmarks/blob/main/ADT-LIA/chc-comp24-ADT-LIA-140.smt2}.}
\begin{quote}
\begin{verbatim}
let rec reva l1 l2 = match l1 with [] -> l2 | x::l1' -> reva l1' (x::l2)
let main l1 l2 = assert(reva (reva l1 l2) [] = reva l2 l1)
\end{verbatim}
\end{quote}
The function \T{reva} takes two lists \(l_1\) and \(l_2\) as arguments,
and returns a list obtained by appending the reverse of \(l_1\) to \(l_2\);
for instance, \(\T{reva}\; [1;2]\; [3;4] = [2;1;3;4]\). The function \T{main}
takes two lists \(l_1\) and \(l_2\) and asserts that
\(\T{reva}\; (\T{reva}\; l_1\; l_2)\; [\,] = \T{reva}\; l_2\; l_1\) holds.

Suppose we wish to prove that the assertion in the function \T{main}
does not fail for any \(l_1\) and \(l_2\). That amounts to proving the proposition:
\begin{align}
  \forall l_1,l_2, l_3, l_4, l_5.\;&  \Reva(l_1,l_2,l_3)\land \Reva(l_3, \Nil, l_4)\land \Reva(l_2,l_1,l_5)\imp l_4=l_5
  \label{reva:goal}
  \end{align}
where 
\(\Reva(l_1,l_2,l_3)\) intuitively means that \(\T{reva}\;l_1\;l_2\) may return \(l_3\);
it is defined as the least predicate that satisfies the following clauses:
\begin{align}
 \forall l_1,l_2.\; &\Reva(\Nil, l_2, l_2). \label{reva:base}\\
 \forall x, l'_1,l_2, l_3.\;& \Reva(l'_1, x::l_2, l_3) \imp \Reva(x::l'_1, l_2, l_3). \label{reva:induct}
 \end{align}
Equivalently, it suffices to prove that there exists a predicate \(\Reva\) that satisfies all the three clauses
\eqref{reva:goal}--\eqref{reva:induct}; in other words, it suffices to find an invariant \(\Reva\) among the
arguments and the return value of \(\T{reva}\) that is strong enough to satisfy \eqref{reva:goal}.

Automatically proving the above property is non-trivial. One can easily see that
directly trying to prove \eqref{reva:goal} by
induction on \(l_1, l_2\), or \(l_3\)
would get stuck.
Unfold/fold transformations~\cite{PettorossiP00,DBLP:journals/tplp/AngelisFPP18a,DBLP:journals/scp/AsadaSK17} are often used
for reasoning about relational properties (ones between multiple function calls, in the above case),
but they would not directly apply either, because the two recursive calls
\(\T{reva}\;l_1\;l_2\) and \(\T{reva}\;l_2\;l_1\)  induct over different variables.
A reasonable way to prove \eqref{reva:goal} is to find a lemma \(\Reva(l_1,l_2,l_3)\imp \rev{l_1}l_2=l_3\),\footnote{Here, we identify
lists with sequences. We write \(\rev{l_1}\) for the reverse of \(l_1\), and \(l_1l_2\) for the concatenation of two sequences (or lists)
\(l_1\) and \(l_2\).} but how to automatically find such a lemma remains a challenge. In fact, the satisfiability of
the clauses \eqref{reva:goal}--\eqref{reva:induct} is an instance of the CHC satisfiability problem~\cite{Bjorner15},
and there are various CHC solvers supporting ADTs~\cite{Eldarica,DBLP:journals/fmsd/KomuravelliGC16,DBLP:conf/aplas/Champion0S18,DBLP:conf/cav/UnnoTK21},
but to our knowledge, none of them (except the one proposed in this paper)
can automatically find the model of \(\Reva\) that satisfies the clauses \eqref{reva:goal}--\eqref{reva:induct}.

\subsection{Our Approach}
\subsubsection{Solvable Tuple Patterns for Automated Invariant Discovery}
To enable automated inference of invariant relations among list-like (functional) data structures
(like the relation \(\rev{l_1}l_2=l_3\) above),
we introduce the notion of \emph{solvable tuple patterns} (STPs) and conjunctive STPs (CSTPs), and apply a data-driven approach to infer them.
A \emph{tuple pattern}  is of the form \((\pat_1,\ldots,\pat_n)\), where \(p_i\) consists of constants (representing elements of
a sequence), variables, concatenations, and the reverse operation.
The pattern \((\pat_1,\ldots,\pat_n)\) represents the set of all tuples of sequences obtained by instantiating
the variables to sequences.
For example, the invariant on the relation \(\Reva\) above is
expressed as \((l_1,l_2, \rev{l_1}l_2)\), where \(l_1\) and \(l_2\) are variables representing sequences.
The tuple pattern \((l_1l_2, l_2l_3, l_1l_2l_3)\) expresses the set of tuples \((s_1,s_2,s_3)\) where \(s_1\) and \(s_2\) are respectively
a prefix and a suffix of \(s_3\), and \(s_1\) and \(s_2\) overlap with each other in \(s_3\) (as represented by \(l_2\));
thus \((ab, bcd, abcd)\)\footnote{Here, we denote elements by letters.} belongs to \((l_1l_2, l_2l_3, l_1l_2l_3)\), but \((ab, d, abcd)\)
does not. The tuple patterns have been inspired from Angluin's pattern language~\cite{ANGLUIN1980117}
(in fact, by using a special symbol \(\$\), a tuple pattern \((\pat_1,\ldots,\pat_n)\) can be expressed as a pattern
\(\pat_1\$\cdots\$\pat_n\)). Pattern languages have a good property that they can be learned from only positive samples
(in the sense of Gold~\cite{GOLD1967447}),
but dealing with them are computationally costly: even the membership problem is NP-complete~\cite{ANGLUIN198046}
and the inclusion is undecidable~\cite{nowotka_et_al:LIPIcs.CPM.2025.4}. 
We thus introduce a subclass of tuple patterns called \emph{solvable} tuple patterns.
As we discuss later, the solvable tuple patterns have several pleasant properties:\footnote{
The definition of solvability and the precise meaning of the properties (i)--(iv) will be
given later in Section~\ref{sec:stp}.}
(i) like pattern languages, they can be learned from only positive samples; (ii)
there is a (non-deterministic) polynomial-time algorithm for inferring a solvable tuple pattern, given positive samples;
 (iii) only a small number of positive samples are required for learning a solvable tuple pattern
(for example, to infer \((l_1,l_2, \rev{l_1}l_2)\), only two samples: \((ab, cd, bacd)\) and \((bc, da, cbda)\) are sufficient);
and (iv) the satisfiability of quantifier-free formulas consisting of word equations and membership constraints on solvable patterns
is decidable.
Both of the patterns \((l_1,l_2, \rev{l_1}l_2)\) and \((l_1l_2, l_2l_3, l_1l_2l_3)\) above
belong to the class of solvable patterns.

The learnability from only positive samples is attractive in the context of program verification.
Positive samples (such as input/output pairs of a function, and states at the beginning of a loop) can be
easily collected by random executions of a program.
In contrast, negative samples (states that would lead to errors) are also required by
many of the recent data-driven approaches to invariant discovery for
automated program verification~\cite{garg_2014,DBLP:journals/jar/ChampionCKS20,DBLP:conf/pldi/ZhuMJ18,zhu_2015},
but collecting negative samples is often harder, since it requires backward executions of a program.
Some of the approaches require even ``implication constraints''~\cite{garg_2014,DBLP:journals/jar/ChampionCKS20,DBLP:journals/pacmpl/EzudheenND0M18}.

To apply STPs to automated verification of list-manipulating programs, we consider CHCs on words, which are Horn clauses extended
with word equality/inequality constraints. The clauses on \(\Reva\) above may be considered CHCs on words,
by viewing lists as words. We show (i) given a set of CHCs on words consisting of definite clauses, there exists an algorithm
to compute the \emph{least} model describable by conjunctive STPs (i.e., a conjunction of a finite number of STPs),\footnote{Stated without using the terminology of CHCs, this means that there exists an algorithm that, given a list-manipulating program with loops and first-order recursive functions, computes the \emph{strongest} inductive invariant
expressible in CSTPs, .}
and (ii) it is decidable whether a given system of CHCs over words has a model describable by conjunctive STPs.
We can directly apply these results to the \(\Reva\) example, and automatically obtain the invariant:
\(\Reva(l_1,l_2,l_3) \equiv l_3=\rev{l_1}l_2\).

\subsubsection{Combination with Arithmetic Reasoning}
Whilst CSTPs have attractive features as mentioned above,
invariants expressed by CSTPs are not always strong enough for the purpose of program verification,
especially for programs using a combination of lists and integers.
For example, consider the following program.
\begin{multicols}{2}
\begin{verbatim}
let rec take n l =
  if n=0 then [] else 
    match l with [] -> []
      | x::l' -> x::(take (n-1) l')
\end{verbatim}
\begin{verbatim}
let rec drop n l =
  if n=0 then l else 
    match l with [] -> []
      | _::l' -> drop (n-1) l'
\end{verbatim}
\end{multicols}
\noindent
Given an integer \(n\) and a list \(l\) as arguments,
the function \(\T{take}\) returns the list consisting of the first (at most) \(n\) elements
of \(l\), and
\(\T{drop}\) returns the list obtained by removing the first (at most) \(n\) elements from \(l\).
Suppose we wish to prove that \((\T{take}\;n\;l)@(\T{drop}\;n\;l)=l\) holds for every (non-negative) integer \(n\) and list \(l\)
(where @ denotes the list append function).
By using solvable tuple patterns, we can infer the properties that \(\T{take}\;n\;l\) and \(\T{drop}\;n\;l\)
respectively return a prefix and a suffix of \(l\), but they are not strong enough to derive \((\T{take}\;n\;l)@(\T{drop}\;n\;l)=l\).

\iffull
To address the issue above, we combine STP inference with existing automated verification tools
(more specifically, CHC solvers over
integer arithmetic~\cite{Eldarica,DBLP:journals/fmsd/KomuravelliGC16,DBLP:conf/aplas/Champion0S18}).
For the example above, we construct the following integer program by abstracting
lists to their lengths.\footnote{As discussed in Section~\ref{sec:chc}, we actually apply this kind of transformation to CHCs,
rather than source programs.}
\begin{multicols}{2}
\begin{verbatim}
let rec take' n l =
  if n=0 then 0 else if l=0 then 0 else
    let l'=l-1 in 1+(take' (n-1) l')
\end{verbatim}
\begin{verbatim}
let rec drop' n l =
  if n=0 then l else if l=0 then 0 else 
    let l'=l-1 in drop' (n-1) l'
\end{verbatim}
\end{multicols}
\noindent
Here, the second argument \(l\) of \(\T{take'}\) and \(\T{drop'}\) is an integer,
representing the length of the second argument of the original functions
 \(\T{take}\) and \(\T{drop}\).
We then try to prove that the abstract version of the property: \((\T{take'}\;n\;l)+(\T{drop'}\;n\;l)=l\) holds
for all \emph{integers} \(n\) and \(l\). State-of-the-art CHC solvers can quickly solve that verification problem,
and can infer invariants like: 
\begin{align*}
&  \mathit{Take'}(n, l, r) \equiv (l<n\land r=l)\lor (l\ge n\land r=n)\\
&  \mathit{Drop'}(n, l, r) \equiv (l<n\land r=0)\lor (l\ge n\land r=l-n).
\end{align*}
Here, \(\mathit{Take'}\) (\(\mathit{Drop'}\), resp.) represents the relationship between the arguments \(n\) and \(l\)
of \(\T{take'}\) (\(\T{drop'}\), resp.) and the return value \(r\). By combining them with the output of solvable tuple pattern
inference, we obtain the invariant candidates:
\begin{align*}
&  \mathit{Take}(n, l, r) \equiv \exists s. rs=l\land ((\len{l}<n\land \len{r}=\len{l})\lor (\len{l}\ge n\land \len{r}=n))\\
&  \mathit{Drop}(n, l, r) \equiv \exists s. sr=l\land ((\len{l}<n\land \len{r}=0)\lor (\len{l}\ge n\land \len{r}=\len{l}-n)).
\end{align*}
These are indeed invariants strong enough to satisfy
\(\mathit{Take}(n,l,r_1)\land \mathit{Drop}(n,l,r_2)\imp r_1r_2=l\).
\else
To address the issue above, we can combine STP inference with existing automated verification tools
(more specifically, CHC solvers over
integer arithmetic~\cite{Eldarica,DBLP:journals/fmsd/KomuravelliGC16,DBLP:conf/aplas/Champion0S18}),
as discussed later in Section~\ref{sec:combination}.
For the example above, we can strengthen the STP-based invariant to:
\begin{align*}
&  \mathit{Take}(n, l, r) \equiv \exists s. rs=l\land (\len{r}=\len{l}<n\lor \len{l}\ge n=\len{r})\\
&  \mathit{Drop}(n, l, r) \equiv \exists s. sr=l\land ((\len{l}<n\land \len{r}=0)\lor \len{l}\ge n=\len{l}-\len{r}),
\end{align*}
where \(|l|\) denotes the length of \(l\), and
\(\mathit{Take}(n, l, r)\) and \(\mathit{Drop}(n, l, r)\) represent the relations between the arguments
\(n,l\) and the return value \(r\) of \(\mathit{take}\) and \(\mathit{drop}\) respectively. The parts \(\exists s. rs=l\)
and \(\exists s.sr=l\) are invariants obtained by the STP inference, and the rest has been obtained by
abstracting lists in the original program to their lengths, and analyzing the resulting integer programs.
The invariants are strong enough to satisfy \(\mathit{Take}(n,l,r_1)\land \mathit{Drop}(n,l,r_2)\imp r_1r_2=l\).
\fi

\iffull
The above method is somewhat reminiscent of methods (such as Nelson-Oppen combination procedure~\cite{10.1145/322186.322198})
for deciding combinations of theories; indeed, it combines a method for reasoning about lists (based on STP inference)
with one for verifying integer programs.
Various techniques to combine static analyses have also been
proposed~\cite{10.1145/200994.200998,DBLP:conf/pldi/GulwaniT06,10.1007/978-3-642-19805-2_31}.
To our knowledge, however,
the way we combine methods for solving problems in individual domains also seems novel.
\fi
\subsection{Contributions and the Organization of This Paper}
The contributions of this paper are summarized as follows.
\begin{itemize}
\item Development of the theory and algorithms for solvable tuple patterns (STPs) and conjunctive STPs (CSTPs).
To the best of our knowledge, the notions of STPs and CSTPs are novel and of independent interest even outside the context of automated program verification, for example in algorithmic learning theory.  
\item Applications of solvable tuple patterns to automated verification of list-manipulating programs.\footnote{This paper targets programs manipulating \emph{functional} (i.e., immutable) lists.
Handling programs manipulating mutable linked lists and other mutable data structures would require combining our method with other techniques~\cite{DBLP:journals/toplas/MatsushitaTK21}; see Remark~\ref{rem:mutable-lists}.}
  To this end, we introduce CHCs over words and show that it is decidable whether
  a given system of CHCs over words has a model
  describable by CSTPs. We also consider CHCs over multisets and sets, and
  show that it is decidable whether a given system of CHCs over multisets or sets has a model describable by the corresponding multiset or set version of CSTPs.
  We further propose a method to strengthen this approach by combining it with verification methods for integer-manipulating programs.
\item Implementation and experiments. Based on the results above,
  we implemented a new CHC solver called \chocolat{}, which supports
  list-like data structures. We evaluated \chocolat{} on the benchmark set of the CHC-COMP 2025 ADT-LIA category,
  and found that \chocolat{} significantly outperformed other CHC solvers.
  A combination of \chocolat{} with another solver called \catalia~\cite{KatsuraSAS25}
  won the CHC-COMP 2025 ADT-LIA category by a large margin. \end{itemize}

The rest of the paper is structured as follows.
In Sections~\ref{sec:tpinf}--\ref{sec:ext}, we introduce the notion of solvable tuple patterns (STPs) and conjunctive STPs (CSTPs),
present algorithms for inferring them from positive samples, and prove that the classes of STPs and CSTPs satisfy
desired properties. Section~\ref{sec:chc} shows how to apply CSTPs to CHC solving (and hence also
automated program verification). Section~\ref{sec:exp} reports an implementation and experimental results.
Section~\ref{sec:rel} discusses related work and Section~\ref{sec:conc} concludes the paper.
\iffull\else
Proofs and additional details are found in the longer version~\cite{STPlong}.
\fi
 
\section{Tuple Patterns and Inference Algorithm}
\label{sec:tpinf}
In this section,
we introduce the notions of \emph{tuple patterns} and \emph{conjunctive tuple patterns},
and provide inference algorithms.
In the next section, we will introduce
a subclass of tuple patterns called \emph{solvable tuple patterns}
as the sound and complete characterization of the class of tuple patterns inferable by
 the algorithm.

\subsection{Tuple Patterns and Conjunctive Tuple Patterns}
Let $\Alpha$ be a set of letters, and \(\Vars\) a (countably infinite) set of variables.
We assume that $\Alpha$ contains at least two distinct letters.
For the purpose of representing a relation among lists, \(\Alpha\) is the set of
possible list elements; for example, if we consider integer lists, then we let \(\Alpha\) be the set of integers.
Below we often write \(a,b,c,\ldots\) for elements of \(\Alpha\). \iffull We write \(\dom(f)\) for the domain of a map \(f\). \fi

The sets of \emph{tuple patterns} and \emph{conjunctive tuple patterns}, ranged over by \(\tp\) and \(\ctp\) respectively,
are 
\iffull
defined by:
\begin{align*}
  & \tp \mbox{ (tuple patterns) }::= (\pat_1,\ldots,\pat_k) \qquad
  \pat \in (\Alpha\cup \Vars)^* \\
  & \gamma \mbox{ (conjunctive tuple patterns) }::= \tp_1\land \cdots \land \tp_m. 
\end{align*}
\else
defined by
\(\tp ::= (\pat_1,\ldots,\pat_k)\) and \(\gamma ::= \tp_1\land \cdots \land \tp_m\) where 
\(\pat \in (\Alpha\cup \Vars)^*\).
\fi
A tuple pattern \(\tp=(\pat_1,\ldots,\pat_k)\) represents a set of \(k\)-tuples of
\(\Alpha\)-sequences,
i.e., a \(k\)-ary relation between sequences, obtained by replacing each variable with a sequence of letters in \(\Alpha\).
 For example, \((x, y, xy)\) represents the set of triples of the form
\((s_1,s_2,s_1s_2)\) where \(s_1,s_2\in \Alpha^*\).
Formally, we define \(\lang(\tp)\subseteq \Alpha^*\times \cdots \times \Alpha^*\) as:
\(\set{[s_1/x_1,\ldots,s_n/x_n]\tp\mid s_1,\ldots,s_n\in \Alpha^*}\)
where \(x_1,\ldots,x_n\) are the variables occurring in \(\tp\), and \([s_1/x_1,\ldots,s_n/x_n]\) denotes the simultaneous substitution
of \(s_1,\ldots,s_n\) for \(x_1,\ldots,x_n\).
A conjunctive tuple pattern \(\tp_1\land \cdots \land \tp_m\) (where \(m\ge 1\))
represents the intersection of the sets of tuples represented by
\(\tp_1,\ldots,\tp_m\), as defined by:
\(\lang(\tp_1\land\cdots\land \tp_m) = \lang(\tp_1)\cap \cdots \cap \lang(\tp_m)\).
For example, \((x, y, xz)\land (x,y,yz)\) represents the set of triples \((s_1,s_2,s_3)\) where both \(s_1\) and \(s_2\) are
prefixes of \(s_3\); here, note that variables are implicitly bound in each conjunct:
the two occurrences of \(z\) in \((x, y, xz)\) and \((x, y, yz)\) may be distinct.
We write \(\seq{\cdot}\) for a sequence. By abuse of notation,
  we write \(\seq{p}\) for a sequence \(p_1,\ldots,p_n\) and also for a tuple pattern
  \((p_1,\ldots,p_n)\), depending on the context; similarly for
  sequences \(\seq{x}\) and \(\seq{s}\) of variables and strings respectively.

We identify tuple patterns up to variable renaming. For example, we identify \((xy, x)\) with \((yx, y)\).
We write \(\FV(\pat)\) and \(\FV(\tp)\) for the sets of variables occurring in \(\pat\) and \(\tp\) respectively.
We write \(p_1\cdot p_2\) for the concatenation of two patterns sequences
\(p_1\) and \(p_2\), and often omit \(\cdot\),  writing \(p_1p_2\). We write \(|\pat|\)
for the length of a pattern \(\pat\).
For a tuple pattern \(\tp=(\pat_1,\ldots,\pat_k)\), we write
\(|\tp|\) and \(\measure{\tp}\) respectively for \(k\) and \(|\pat_1\cdots\pat_k|+k\).

We will later (in Section~\ref{sec:ext}) extend tuple patterns with the pattern \(\rev{x}\)
(which represents the reverse of \(x\)), to accommodate the \(\Reva\) example in
Section~\ref{sec:intro}.

\begin{remark}
\label{rem:pattern-language}
In the case of a singleton tuple pattern \((\pat)\), our definition of \(\lang\) coincides with the definition of a variation of 
pattern languages~\cite{ANGLUIN1980117,ANGLUIN198046} called
E-pattern languages (or erasing pattern languages)~\cite{10.1007/3-540-11980-9_19}.
In the definition of original pattern languages~\cite{ANGLUIN1980117,ANGLUIN198046} (called NE-pattern languages),
\(\lang(\pat)=\set{[s_1/x_1,\ldots,s_n/x_n]\pat\mid s_1,\ldots,s_n\in \Alpha^+}\); empty strings cannot be substituted for
variables.
In view of the application to program verification
discussed in Section~\ref{sec:chc}, we allow substitutions of empty strings, which correspond to empty lists.
\qed
\end{remark}

\subsection{Tuple Pattern Inference Problem}
\emph{Learning data}, denoted by \(\dat\),
is an \(m\times n\)-matrix consisting of elements of \(\Alpha^*\).
We write \(\dat[i][j]\) for the element in the \(i\)-th row and \(j\)-th column.
We also write \(\dat[i]\) for \((\dat[i][1],\ldots,\dat[i][n])\), and \(\dat[*][j]\) for \((\dat[1][j],\ldots,\dat[m][j])\).
We write \(\shapeof{\dat}\) for \((m,n)\), \(\rowsof{\dat}\) for \(m\),
\(\colsof{\dat}\) for \(n\),
and \(\sizeof{\dat}\) for \(\sum_{i\in\set{1,\ldots,m},j\in\set{1,\ldots,n}} (1+|\dat[i][j]|)\),
where \(|s|\) denotes the length of string \(s\in\Alpha^*\).
We sometimes regard learning data \(\dat\) as the \emph{set} \(\set{\dat[1],\ldots,\dat[m]}\),
and write \(\dat_1\subseteq\dat_2\) when \(\dat_1\) is a subset of \(\dat_2\) as sets.

Let \(\seq{x}=(x_1,\ldots,x_k)\) and \(\tp=(p_1,\ldots,p_n)\).
If \(\shapeof{\dat'}=(m,k)\) and \(\FV(\tp)\subseteq \set{\seq{x}}\),
we write \([\dat'/\seq{x}]\tp\) for the learning data \(\dat\) such that
\(\shapeof{\dat}=(m,n)\) and \(\dat[i][j] = [\dat'[i][1]/x_1, \ldots, \dat'[i][k]/x_k]p_j\).
For example, if
\(\dat' = \left(\begin{array}{cc} a & b \\ bb & a\end{array}\right)\)
  and \(\tp=(x, y, xy)\), then
  \([\dat'/(x,y)]\tp =
  \left(\begin{array}{ccc} a & b & ab\\ bb & a & bba\end{array}\right)\).
We use the metavariable \(\Theta\) for a substitution \([\dat/\seq{x}]\).    
We sometimes use the matrix notation:    
\[
\left(
\begin{array}{ccccc}
  x_1 & & \cdots & & x_k\\
  s_{1,1}&& \cdots & & s_{1,k}\\
  \cdots&& \cdots & & \cdots\\
  s_{m,1}&& \cdots & & s_{m,k}
  \end{array}
\right)
\]
for the substitution \([\dat/(x_1,\ldots,x_k)]\) such that \(\dat[i][j]=s_{i,j}\).

When \(\Theta \tp=\dat\) holds, we write \(\dat,\Theta\models \tp\).
We also write \(\dat,\Theta\smodels \tp\) if
\(\Theta\tp =\dat\) and there exists no variable \(x\) such that
\(\Theta(x)= \seq{\epsilon}\) (i.e., if
no variable is always mapped to \(\epsilon\)).
We write \(\dat\models \tp\) (\(\dat\smodels \tp\), resp.) if
there exists \(\Theta\) such that \(\dat,\Theta\models \tp\) (\(\dat,\Theta\smodels \tp\), resp.), and we call \(\Theta\) a \emph{witness substitution} for \(\dat\models\tp\) (\(\dat\smodels \tp\), resp.).
For a conjunctive tuple pattern \(\ctp = \tp_1\land \cdots \land \tp_m\), we write \(\dat\models \ctp\) (\(\dat\smodels \ctp\), resp.) if
\(\dat\models \tp_i\) (\(\dat\smodels \tp_i\), resp.) holds for every \(i\in\set{1,\ldots,m}\).

\begin{example}
  \label{ex:model}
  Let \(\dat\) be:
  \(
  \left(\begin{array}{ll}a & aaa\\ b & abb \end{array}\right)
  \) and \(\tp\) be \((x, axx)\).
  Then \(M\smodels \tp\), where 
  \(\Theta=\left(\begin{array}{c}x \\a \\ b\end{array}\right)\) is a witness substitution.
  We also have \(M, \Theta \models (x, axxy)\) for
  \(\Theta = \left( \begin{array}{cc} x &  y\\ a & \epsilon \\ b & \epsilon \end{array}\right)\), but \(M \not\smodels (x, xxy)\) as \(\Theta \) always maps \(y\) to \(\epsilon\).
  \qed
\end{example}

\iffull
\begin{example}
\label{ex:tupinf}
  Let \(M\) be:
  \(\left(\begin{array}{rrr}
    a &\ b &\ ab\\
    aa & \epsilon & aa
  \end{array}\right)
  \).
  Then, \(\dat\models (ax, y, axy)\).
  In fact, for \(\Theta =
  \left( \begin{array}{cc} x &  y\\ \epsilon & b \\ a & \epsilon \end{array}\right)\),
  we have \(\dat, \Theta\models (ax,y,axy)\).
  In this case, we also have \(\dat\smodels (ax, y, axy)\), with the same witness substitution. \qed
\end{example}
\fi
The \emph{tuple pattern inference} problem is the problem of, given learning data \(\dat\), finding
a tuple pattern \((\pat_1,\ldots,\pat_n)\) such that \(\dat\models (\pat_1,\ldots,\pat_n)\).
For instance, for \(\dat\) in Example~\ref{ex:model}, \((x, axx)\) is an answer to the inference problem.

\subsection{Inference Algorithm}
We now present a (non-deterministic) algorithm for tuple pattern inference.
We formalize it by using a rewriting relation \((\tp, \Dat)\red (\tp',\Dat')\).
Starting with the initial state \((\seq{x}, \Subst{\seq{x}}{\dat})\)
(where \(\seq{x}=(x_1,\ldots,x_n)\)),
the algorithm repeatedly applies the rewriting rules below until no further rule is applicable.
When the final pair is \((\tp',\Dat')\), then \(\tp'\) is the inferred tuple pattern.
In other words, given \(\dat\),
our algorithm for tuple pattern inference (non-deterministically) outputs
\(\tp'\) such that \((\seq{x}, \Subst{\seq{x}}{\dat})\reds (\tp',\Dat')\nred\).
As we will see
later (in Theorem~\ref{th:soundness}), \((\seq{x}, \Subst{\seq{x}}{\dat})\reds (\tp',\Dat')\) implies \(\dat,\Dat'\models \tp'\), because the rewriting relation \((\tp, \Dat)\red (\tp',\Dat')\) maintains the invariant \(\Dat\tp=\Dat'\tp'\).

\infrule[R-Prefix]{\dat[*][i]\ne\seq{\epsilon}\andalso \dat[*][j]=\dat[*][i]\cdot \seq{s}\andalso x_j'\mbox{ fresh}}
        {(\tp, \Subst{(x_1,\ldots,x_m)}{\dat})\red ([x_i x_j'/x_j]\tp, \Subst{(x_1,\ldots,x_{j-1},x_j', x_{j+1},\ldots, x_m)}{\replace{\dat}{j}{\seq{s}}})}
\infrule[R-CPrefix]{\dat[*][j]=a \cdot \seq{s}\andalso a\in\Alpha\andalso x_j'\mbox{ fresh}}
{(\tp, \Subst{(x_1,\ldots,x_m)}{\dat})\red ([ax_j'/x_j]\tp, \Subst{(x_1,\ldots,x_{j-1},x_j',x_{j+1},\ldots, x_m)}{\replace{\dat}{j}{\seq{s}}})}
\infrule[R-Suffix]{\dat[*][i]\ne\seq{\epsilon}\andalso \dat[*][j]=\seq{s}\cdot\dat[*][i] \andalso x_j'\mbox{ fresh}}
        {(\tp, \Subst{(x_1,\ldots,x_m)}{\dat})\red ([x_j' x_i/x_j]\tp, \Subst{(x_1,\ldots,x_{j-1},x_j', x_{j+1},\ldots, x_m)}{\replace{\dat}{j}{\seq{s}}})}
        
\infrule[R-CSuffix]{\dat[*][j]=\seq{s}\cdot a\andalso a\in\Alpha\andalso x_j'\mbox{ fresh}}
        {(\tp, \Subst{(x_1,\ldots,x_m)}{\dat})\red
          ([x_j'a/x_j]\tp, \Subst{(x_1,\ldots,x_{j-1},x_j',x_{j+1},\ldots, x_m)}{\replace{\dat}{j}{\seq{s}}})}

\infrule[R-Epsilon]{\dat[*][j]=\seq{\epsilon}}{(\tp, \Subst{(x_1,\ldots,x_m)}{\dat})\red ([\epsilon/x_j]\tp, \Subst{(x_1,\ldots,x_{j-1},x_{j+1},\ldots, x_m)}{\remove{\dat}{j}})}

\noindent
        Here, \(\remove{\dat}{j}\) denotes the matrix obtained from \(\dat\) by removing the \(j\)-th column,
        and \(\replace{\dat}{j}{\seq{s}}\) denotes
        the matrix obtained from \(\dat\) by replacing the \(j\)-th column with \(\seq{s}\).
        We write \(a\cdot \seq{s}\) for the sequence obtained from \(\seq{s}\) by appending \(a\) to the head of each element of \(\seq{s}\).
        
The rule \rn{R-Prefix} is for the case where the \(i\)-th column of
\(M\) is a prefix of the \(j\)-th column. In this case, we replace
the \(j\)-th column with its suffix obtained by removing the prefix \(M[*][i]\),
and instantiate \(x_j\) in  \(t\) accordingly to \(x_i\cdot x_j'\).
The rule \rn{R-CPrefix} is for the case where a constant \(a\) is a prefix of the \(j\)-th column, and the rule \rn{R-Epsilon} is for the case where every element of
the \(j\)-th column is an empty sequence. In the latter case, the \(j\)-th column
is removed.

\begin{example}
    \label{ex:algo}
  Recall the data in Example~\ref{ex:model}:
    \(
\dat=  \left(\begin{array}{ll}a & aaa\\ b & abb \end{array}\right)
  \).
  The pattern \((x,axx)\) is inferred (up to variable renaming) as follows.
  \begin{align*}
&  \Big((x_1, x_2),\left(\begin{array}{rr}
        x_1 & x_2\\
    a &\ aaa\\
    b & abb
      \end{array}\right)\Big) \red \Big((x_1,ax_2'), \left(\begin{array}{rr}
        x_1 & x_2'\\
    a &\ aa\\
    b & bb
      \end{array}\right)\Big)
\red \Big((x_1, ax_1x_2''), \left(\begin{array}{rr}
        x_1 & x_2''\\
    a &\ a\\
    b &\ b
    \end{array}\right)\Big)& \\
    &  \red \Big((x_1, ax_1x_1x_2'''), \left(\begin{array}{rr}
        x_1 & x_2''' \\
    a &\ \epsilon  \\
    b & \epsilon 
    \end{array}\right)\Big) \red \Big((x_1, ax_1x_1), \left(\begin{array}{cc}
        x_1\\
    a\\
    b
    \end{array}\right)\Big) \nred {}
  \end{align*}
  Actually, the patterns such as \((x_1,ax_2')\) and 
  \((x_1, ax_1x_1x_2''')\) obtained in intermediate steps are also valid solutions of the inference problem,
  but the final pattern describes a smaller relation than them.
  \qed
\end{example}

\iffull
\begin{example}
    \label{ex:algo2}
  Recall the data in Example~\ref{ex:tupinf}:
  \(\dat = \left(\begin{array}{rrr}
    a &\ b &\ ab\\
    aa & \epsilon & aa
  \end{array}\right)
  \).
  The pattern \((ax,y,axy)\) is inferred (up to variable renaming) as follows.
  \begin{align*}
    &  \Big((x_1, x_2, x_3), \left(\begin{array}{rrr}
      x_1 & x_2 & x_3\\
    a &\ b &\ ab\\
    aa & \epsilon & aa
  \end{array}\right)\Big) \red_{\mbox{(\rn{R-Prefix}, \(x_3=x_1x_3'\))}} \Big((x_1, x_2, x_1x_3'), \left(\begin{array}{rrr}
        x_1 & x_2 & x_3'\\
    a &\ b &\ b\\
    aa & \epsilon & \epsilon
      \end{array}\right)\Big) \\&  \red_{\mbox{ (\rn{R-Prefix}, \(x_3'=x_2x_3''\))}} \Big((x_1, x_2, x_1x_2x_3''), \left(\begin{array}{rrr}
        x_1 & x_2 & x_3''\\
    a &\ b &\ \epsilon\\
    aa & \epsilon & \epsilon
    \end{array}\right)\Big)& \\
    &  \red_{\mbox{ (\rn{R-Epsilon}, \(x_3''=\seq{\epsilon}\))}} \Big((x_1, x_2, x_1x_2), \left(\begin{array}{rr}
        x_1 & x_2 \\
    a &\ b \\
    aa & \epsilon 
    \end{array}\right)\Big)& \\
    &  \red_{\mbox{ (\rn{R-CPrefix}, \(x_1=ax_1'\))}} \Big((ax'_1, x_2, ax'_1x_2), \left(\begin{array}{rr}
        x_1' & x_2\\
    \epsilon &\ b \\
    a & \epsilon 
    \end{array}\right)\Big) \nred {}
  \end{align*}
  \qed
\end{example}
\fi
 The following example shows that the result of tuple pattern inference may not be unique.
 \begin{example}
   \label{ex:nondet}
       Consider the data:
    \(\dat = \left(\begin{array}{rrr}
    aa &\ a &\ aac\\
    b &\ bb &\ bbd.
  \end{array}\right)\).
    Here, both the first and second columns are prefixes of the third column.
    We obtain \((x,y, xz)\) or \((x,y,yz)\) depending on which column is chosen
    in \rn{R-Prefix}.
    \begin{align*}
&  \Big((x_1, x_2, x_3),
      \left(\begin{array}{rrr}
        x_1 & x_2 & x_3\\
    aa &\ a &\ aac\\
    b &\ bb &\ bbd.
    \end{array}\right)\Big) \red \Big((x_1, x_2, x_1x_3'), 
      \left(\begin{array}{rrr}
        x_1 & x_2 & x_3'\\
    aa &\ a &\ c\\
    b &\ bb &\ bd.
    \end{array}\right)\Big) \end{align*}
    \begin{align*}
&  \Big((x_1, x_2, x_3),
      \left(\begin{array}{rrr}
        x_1 & x_2 & x_3\\
    aa &\ a &\ aac\\
    b &\ bb &\ bbd.
    \end{array}\right)\Big) \red \Big((x_1, x_2, x_2x_3'), \left(\begin{array}{rrr}
        x_1 & x_2 & x_3'\\
    aa &\ a &\ ac\\
    b &\ bb &\ d.
    \end{array}\right)\Big). \end{align*}
    The first pattern obtained
    describes that the first element is a prefix of the third element,
    while the second pattern describes that the second element is a prefix of the third element.
    The conjunctive pattern \((x_1, x_2, x_1x_3')\land (x_1,x_2,x_2x_3')\) subsumes both patterns.
\qed
 \end{example}

\begin{definition}[Algorithms $\TPinf$ and $\CTPinf$]
\label{def:algo}
We write \(\TPinf\) for a non-deterministic algorithm that takes data \(\dat\) as input,
applies the rewriting rules to \((\seq{x},\Subst{\seq{x}}{\dat})\) until no further rule is applicable,
and outputs some tuple pattern \(\tp\) such that \((\seq{x},\Subst{\seq{x}}{\dat})
\reds (\tp, \Dat)\nred{}\).
We write \(\CTPinf\) for a deterministic algorithm that takes data \(\dat\) as input
and outputs the conjunctive tuple pattern
\(\bigwedge\set{\tp \mid (\seq{x}, \Subst{\seq{x}}{\dat})\reds (\tp,\Dat)\nred{}}\).
\end{definition}
Note that the set of \((\tp,\Dat)\) such that
\((\seq{x}, \Subst{\seq{x}}{\dat})\reds (\tp,\Dat)\nred{}\) is finite; thus we can indeed compute
\(\CTPinf(\dat)\) by an exhaustive search.
Note, however, that \(\CTPinf\) runs in time exponential in \(\sizeof{\dat}\) in the worst case.

We now state some key properties of the algorithm \(\TPinf\); other important properties are discussed in Section~\ref{sec:solvable}.
Proofs omitted below are found in \iffull Appendix~\ref{sec:proofs}. \else \cite{STPlong}. \fi
The theorem below states the soundness of the algorithm, which
follows immediately from the fact that \((\tp_1,\Dat_1)\red (\tp_2,\Dat_2)\)
implies \(\Dat_1\tp_1=\Dat_2\tp_2\). See \iffull Appendix~\ref{sec:proofs} \else \cite{STPlong} \fi for more details.
\begin{theorem}[soundness]
  \label{th:soundness}
  Suppose \(\seq{x}=x_1,\ldots,x_n\) are mutually distinct variables and \(\shapeof{\dat}=(m,n)\).
  If \(\TPinf(\dat)\) returns \(\tp\),
then \(\dat \models \tp\). \end{theorem}

Our non-deterministic algorithm can find a pattern in time polynomial in the size of given data.
\begin{theorem}
  \label{th:efficiency}
 Given \(\dat\) as input, \(\TPinf\) runs
in time polynomial in \(\sizeof{\dat}\).
\end{theorem}
\begin{proof}
  In each step of rewriting, an applicable rule can be found and applied in polynomial time,
  if there is any.
  Since each step of the rewriting strictly decreases \(\sizeof{\dat}\),
  the length of the rewriting sequence is linear in \(\sizeof{\dat}\). 
\end{proof}
 \section{Solvable Tuple Patterns}
\label{sec:stp}
\label{sec:solvable}

In this section, we introduce the notions of \emph{solvable tuple patterns} (STPs)
and \emph{conjunctive solvable tuple patters} (CSTPs),
which characterize the classes of
(conjunctive) tuple patterns that can be inferred by our algorithms.
The pattern \((xx)\) on a singleton tuple is obviously non-inferable by our algorithm, and it is indeed
deemed non-solvable in the characterization below. In contrast, the pattern \((xx, x)\) is solvable.

\begin{remark}
  Instead of restricting the class of tuple patterns, one may expect to obtain an extension of the algorithm
  that can infer arbitrary tuple patterns. Unfortunately, learning arbitrary tuple patterns from positive samples
  is computationally infeasible. As mentioned in Remark~\ref{rem:pattern-language}, singleton tuple patterns
  coincide with E-patterns~\cite{10.1007/3-540-11980-9_19}, and
  the full class of E-pattern languages
  is not learnable from positive data in general
  (more precisely, it is not learnable when the alphabet size is
  2, 3, or 4~\cite{REIDENBACH200691,REIDENBACH2008166}
  and the learnability is unknown when the alphabet size is larger).
Also, the inclusion problem is undecidable for both
  NE- and E-pattern languages~\cite{JIANG199553,nowotka_et_al:LIPIcs.CPM.2025.4}.
  In contrast, STPs have good algorithmic properties, as discussed in Section~\ref{sec:decision}.
\qed
\end{remark}

\subsection{Solvability}
We define the notion of solvability via a reduction relation \(\tp \pred\tp'\) on tuple patterns.
Just as the rewriting relation \((\tp,\Dat)\red(\tp',\Dat')\) in Section~\ref{sec:tpinf} detects
 prefix/suffix relationships in \emph{data} and simplifies the data accordingly, the reduction \(\tp\pred\tp'\)
detects prefix/suffix relationships in \emph{patterns} and simplifies the patterns.
The reduction relation \(\pred\) is defined by:
\iffull
\begin{minipage}[t]{0.49\textwidth}
\infrule[PR-Prefix]{p_j = p_i\cdot p'_j\andalso p_i\ne \epsilon}{(p_1,\ldots,p_n) \pred
  (p_1,\ldots,p_{j-1},p'_j,p_{j+1},\ldots,p_n)}
\infrule[PR-CPrefix]{p_j = a\cdot p'_j\andalso a\in \Alpha}
  {(p_1,\ldots,p_n) \pred
    (p_1,\ldots,p_{j-1},p'_j,p_{j+1},\ldots,p_n)}
\end{minipage}
\begin{minipage}[t]{0.49\textwidth}
\infrule[PR-Suffix]{p_j = p'_j\cdot p_i\andalso p_i\ne \epsilon}{(p_1,\ldots,p_n) \pred
  (p_1,\ldots,p_{j-1},p'_j,p_{j+1},\ldots,p_n)}
\infrule[PR-CSuffix]{p_j = p'_j\cdot a\andalso a\in \Alpha}
  {(p_1,\ldots,p_n) \pred
    (p_1,\ldots,p_{j-1},p'_j,p_{j+1},\ldots,p_n)}
\infrule[PR-Epsilon]{p_j=\epsilon}  
  {(p_1,\ldots,p_n) \pred
    (p_1,\ldots,p_{j-1},p_{j+1},\ldots,p_n)}
\end{minipage}\\
\else
\begin{multicols}{2}
\infrule[PR-Prefix]{p_j = p_i\cdot p'_j\andalso p_i\ne \epsilon}{(p_1,\ldots,p_n) \pred
  (p_1,\ldots,p_{j-1},p'_j,p_{j+1},\ldots,p_n)}
\infrule[PR-CPrefix]{p_j = a\cdot p'_j\andalso a\in \Alpha}
  {(p_1,\ldots,p_n) \pred
    (p_1,\ldots,p_{j-1},p'_j,p_{j+1},\ldots,p_n)}
\infrule[PR-Epsilon]{p_j=\epsilon}  
  {(p_1,\ldots,p_n) \pred
    (p_1,\ldots,p_{j-1},p_{j+1},\ldots,p_n)}
\end{multicols}
\noindent
Here, analogous rules for suffixes are omitted.
\fi
  Each rule of name \rn{PR-xx} corresponds to the rule \rn{R-xx};
  while the rule \rn{R-xx} manipulates data \(M\), the rule \rn{PR-xx} applies
  the corresponding operation to a tuple pattern, and simplifies the pattern.
We define the set \(\Tp_k\) of \(k\)-ary \emph{solvable tuple patterns} (STPs, for short) as:
\(\set{t \mid t\preds (x_1,\ldots,x_n), |t|=k, \mbox{ and $x_1,\ldots,x_n$ are distinct from each other}}\),
and write \(\Tp\) for \(\bigcup_{k} \Tp_k\).
In other words, a tuple pattern \(\tp\) is solvable if
\(\tp\) can be reduced to  a trivial pattern \((x_1,\ldots,x_n)\) consisting of distinct variables.
  A conjunctive tuple pattern \(\ctp= \tp_1\land\cdots \land \tp_m\) is \emph{solvable}
  if \(\tp_1,\ldots,\tp_m\in\Tp_k\) for some \(k\).
\iffull
  We call the pattern modified or removed in
  the rewriting \((p_1,\ldots,p_n)\pred (p'_1,\ldots,p'_m)\)
  (i.e., \(p_j\) in the above rules) a \emph{principal} pattern. We also call the pattern
  \(p_i\) in \rn{PR-Prefix} or \rn{PR-Suffix} an \emph{auxiliary} pattern.
  \fi

\begin{example}
    \label{ex:tp}
\((x_1x_2, x_2x_1)\not\in \Tp_2\), but   \((x_1x_2, x_2x_1, x_1)\in \Tp_3\).
In fact, \((x_1x_2, x_2x_1)\not\in \Tp_2\) is obvious, as no rule is applicable to \((x_1x_2,x_2x_1)\).
The latter follows by:
\begin{align*}
  (x_1x_2, x_2x_1, x_1)\pred (x_2,x_2x_1,x_1) \pred (x_2,x_1,x_1)\pred (x_2,x_1,\epsilon)\pred (x_2,x_1),
\end{align*}
where \rn{PR-Prefix} is applied in the first three steps, and \rn{PR-Epsilon} is applied in the last
step.
\end{example}
\iffull
\begin{remark}
  In the definition of \(\Tp_k\), we can actually drop the
  requirement that the variables
  \(x_1,\ldots,x_n\) are distinct from each other.
  Indeed, if \(x_1=x_2\), for example, then we can erase \(x_1\) by:
  \[(x_1,x_1,x_3,\ldots,x_n) \pred (\epsilon,x_1,x_3,\ldots,x_n)\pred (x_1,x_3,\ldots,x_n).\]
  Thus, \(\Tp_k\) does not change even if we do not require that
  \(x_1,\ldots,x_n\) are distinct from each other.
  \qed
\end{remark}
\fi
\begin{remark}
  \label{rem:solvable}
  The name  ``\emph{solvable} tuple patterns'' comes from the following fact.
  Suppose \(\tp=(p_1,\ldots,p_n)\in \Tp_n\). Then, the general solution for a system of equations
  \(p_1=s_1,\ldots,p_n=s_n\) can be expressed by using (i) \(s_1,\ldots,s_n, \epsilon\), (ii) constants \(a\in\Alpha \),
   (iii) the (partial) operation
  \(\ldiff{s}{s'}\), which is defined as \(s_0\) such that \(ss_0=s'\) if \(s\) is a prefix of \(s'\),
  and (iv) the (partial) operation
  \(\rdiff{s}{s'}\), which is defined as \(s_0\) such that \(s_0s=s'\) if \(s\) is a suffix of \(s'\).
  For example, \(\tp=(x_1x_2, x_2x_1, x_1)\) is solvable in the sense that
  the system of equations \(x_1x_2=s_1, x_2x_1=s_2, x_1=s_3\) has a general solution \(x_1=s_3, x_2=\ldiff{s_3}{s_1}\)
  provided  \(\ldiff{(\ldiff{s_3}{s_1})}s_2=s_3\) (or, equivalently \((\ldiff{s_3}{s_1})\cdot s_3=s_2\)).

As in the case of ordinary pattern languages, the class of languages described by solvable tuple patterns
(where a tuple pattern \((\pat_1,\ldots,\pat_n)\) is identified with the language described by
\(\pat_1\$\cdots\$\pat_n\))
is incomparable with the class of context-free languages, but is subsumed by the class of indexed languages~\cite{10.1145/321479.321488}.
Note that the language \(\set{w\$ww\mid w\in\set{a,b}^*}\), described by the solvable
tuple pattern \((x,xx)\),
is not context-free.  
  \qed
\end{remark}

\subsection{Properties of the Inference Algorithm}
\label{sec:stp-properties}
As stated in the following theorems, the algorithms \(\TPinf\) and \(\CTPinf\) defined in
 Definition~\ref{def:algo}
 indeed infer STPs and CSTPs, respectively. Moreover, modulo a minor condition,
 they are \emph{complete} for these classes.
\begin{theorem}[inferred patterns are solvable]
  \label{th:soundness2}
If \(\TPinf(\dat)\) returns \(\tp\), then \(\tp\in \Tp\).
\end{theorem}
\begin{theorem}[completeness of $\TPinf$]
  \label{th:completeness}
  Let \(\seq{x}=x_1,\ldots,x_n\) be mutually distinct variables.
  \begin{enumerate}[(I)]
\item  If \(\dat \smodels \tp\) and \(\tp\in \Tp_n\),
\((\seq{x},\Subst{\seq{x}}{\dat})\reds (\tp,\Dat)\)
  for some \(\Dat\).\footnote{Recall that tuple patterns are identified up to variable renaming.}
\item   If \(\dat \models \tp\) and \(\tp\in \Tp_n\),
  then there exists \(\tp'\in\TP\) such that \(\lang(\tp')\subseteq \lang(\tp)\) and \(\tp'\) is
  a possible output of \(\TPinf(\dat)\), i.e.,  \((\seq{x},\Subst{\seq{x}}{\dat})\reds (\tp',\Dat)\nred{}\).
\end{enumerate}
\end{theorem}
The following theorem, obtained as a corollary of Theorems~\ref{th:soundness} and
\ref{th:completeness},
says that $\CTPinf(\dat)$ returns the \emph{least} CSTP \(\ctp\) such that \(\dat\models \ctp\).
\begin{theorem}[completeness of $\CTPinf$]
  \label{th:completeness-ctpinf}
  Let \(\dat\) be learning data.
  Then \(\dat\models \CTPinf(\dat)\), and
  \(\lang(\CTPinf(\dat))\subseteq \lang(\ctp)\) for every CSTP \(\ctp\) such that \(\dat\models \ctp\).
\end{theorem}
\begin{remark}
Theorem~\ref{th:completeness}~(I) would not hold if the assumption
  \(\dat \smodels t\) were weakened to    \(\dat \models t\).
  In fact, let \(t=(x,x)\) and \(\dat = (\epsilon\ \epsilon)\). Then \(\dat\models t\) and \(t\in \Tp_2\),
  but we can only obtain \((\epsilon,\epsilon), (\epsilon,x), (x,\epsilon), (x,y)\) by rewriting.
  \iffull
  In fact, the only possible reduction sequences are:
  \begin{align*}
    & \Big((x_1,x_2), \left(\begin{array}{rr} x_1 & x_2 \\ \epsilon & \epsilon\end{array}\right)\Big) \red \Big((x_1,\epsilon), \left(\begin{array}{r} x_1 \\ \epsilon \end{array}\right)\Big)
    \red \Big((\epsilon,\epsilon), \big(\,\big)\Big)\\
    & \Big((x_1,x_2), \left(\begin{array}{rr} x_1 & x_2 \\ \epsilon & \epsilon\end{array}\right)\Big) \red \Big((\epsilon, x_2), \left(\begin{array}{r} x_2 \\ \epsilon \end{array}\right)\Big)
    \red \Big((\epsilon,\epsilon), \big(\,\big)\Big).
  \end{align*}
  \fi
  The rule \rn{R-Prefix} is inapplicable
  because of the side condition \(\dat[*][i]\ne\seq{\epsilon}\). The pattern \((x,x)\) could be obtained if we removed the condition, but then we would
  obtain unboundedly many patterns, such as \((xxx, xx)\).
If we add one more sample, like \((a,a)\) to \(\dat\) above, then we can infer \((x,x)\) as expected.
  \qed
\end{remark}

As observed in Example~\ref{ex:nondet}, the output of \(\TPinf\) is not unique due to non-determinism (whereas \(\CTPinf\) is), but
 the following theorem states that \(\TPinf\) always outputs a \emph{minimal} solvable pattern.

\begin{theorem}[minimality]
  \label{th:minimality}
  Suppose \(\TPinf(\dat)\) returns \(\tp_1\).
If \(\dat \models \tp_0\) with \(\tp_0\in \Tp\), then 
\(\lang(\tp_1)\supseteq \lang(\tp_0)\) implies 
\(\lang(\tp_1)= \lang(\tp_0)\).
\end{theorem}
The minimality property above is important for learning from positive samples only.
Suppose that we are trying to learn the language described by
a tuple pattern \(\tp_0\). In learning from positive samples,
we are given data \(M\) consisting of only a (finite) subset of \(\lang(\tp_0)\).
Thus, our non-deterministic algorithm may output a ``wrong'' pattern
\(\tp_1\) such that \(M\models \tp_1\) but \(\lang(\tp_1)\ne \lang(\tp_0)\).
Thanks to the theorem above, in such a case, there exists data
\(\seq{s}\in \lang(\tp_0)\setminus \lang(\tp_1)\).
Thus, by adding the new positive sample
\(\seq{s}\) to \(M\) and running our algorithm again,
the wrong pattern \(\tp_1\) will not be encountered again.
Without the minimality property, it could be the case
\(\lang(\tp_1)\subsetneq \lang(\tp_0)\), and then the pattern \(\tp_1\) would not be refutable
 with only positive samples.

For each STP \(\tp\),
there exist polynomial-size data that uniquely characterize \(\tp\).
This property ensures convergence of the learning process (cf. Theorem~\ref{th:learnability}).
\begin{theorem}[Characteristic Data]
  \label{th:data-size}
  Let \(\tp=(p_1,\ldots,p_n)\) be an STP such that \(|p_1\cdots p_n|=m\).
  Then, there exists \(\dat\) such that (i) \(\sizeof{\dat} = O((m+n)\log n)\),
  (ii) for any \(\dat'\) such that \(\dat\subseteq \dat'\subseteq \lang(\tp)\),
  there exists \(\Dat\) such that \((\seq{x},\Subst{\seq{x}}{\dat'})\reds (\tp, \Dat)\nred{}\), and
  (iii) for any \(\dat'\) such that \(\dat\subseteq \dat'\subseteq \lang(\tp)\),
  \((\seq{x},\Subst{\seq{x}}{\dat'})\reds (\tp', \Dat')\nred{}\)
  implies \(\lang(\tp)=\lang(\tp')\).
  Furthermore, given \(\tp\), \(\dat\) can be constructed in polynomial time. \end{theorem}

\subsection{Learnability}
\label{sec:learnability}
We have so far considered an algorithm for inferring \(\tp\) (\(\ctp\), resp.) such that \(\dat\models \tp\) (\(\dat\models\ctp\), resp.),
given \(\dat\). In the whole learning process, we need to repeatedly invoke the algorithm for gradually increasing learning data 
\(\dat_0\subset \dat_1\subset \dat_2\subset \cdots\), until a true STP (or CSTP) is found.
In Gold's learning model~\cite{GOLD1967447}, a class \(\mathcal{C}\) of languages is \emph{identifiable in the limit} if
there exists an algorithm \(f\) such that, for any language \(L\in\mathcal{C}\) and
any infinite sequence \(s_0,s_1,s_2,\ldots\) s.t. \(L = \set{s_i \mid i\in \omega}\),
the sequence \(f(\set{s_0}),f(\set{s_0,s_1}),f(\set{s_0,s_1,s_2}),\ldots\) eventually converges to (a representation of) \(L\).
Both the class of STPs and that of CSTPs are learnable in this sense, and our algorithms
\(\TPinf\) and \(\CTPinf\) serve as \(f\) above, as stated below.
\begin{theorem}
\label{th:learnability}  
  Suppose \(\tp\) is an STP, \(\lang(\tp) = \set{\seq{s}_i\mid i\in\omega}\), and
   \(\dat_i = \set{\seq{s}_j\mid 0\le j\le i}\). Then, there exists \(k\) such that
  \(\lang(\TPinf(\dat_i))=\lang(\tp)\) for all \(i\ge k\).
  Similarly, if
  \(\ctp\) is a CSTP, \(\lang(\ctp) = \set{\seq{s}_i\mid i\in\omega}\), and
  \(\dat_i = \set{\seq{s}_j\mid 0\le j\le i}\), then there exists \(k\) such that
  \(\lang(\CTPinf(\dat_i))=\lang(\ctp)\) for all \(i\ge k\).
\end{theorem}

In applications to program verification discussed later, we use CSTPs to represent and
infer \emph{inductive} invariants. 
To that end, we need to consider a different learning model where the goal is to find the strongest inductive invariant among those
describable by a CSTP,
given an oracle to check the inductiveness of a CSTP. This can be regarded as a kind of active
learning framework where the oracle is a teacher. The following theorem states that our algorithm \(\CTPinf\)
can also be used in this learning framework.
\begin{theorem}
  \label{th:cstp-lfp}
  Let \(\mathcal{F}\COL 2^{\Sigma^* \times \cdots \times \Sigma^*}\to 2^{\Sigma^* \times \cdots \times \Sigma^*} \) be a monotonic function,
  and suppose that there exists an algorithm \(g_{\mathcal{F}}\) that, given a CSTP \(\ctp\),
 returns ``None'' if \(\lang(\ctp)\supseteq \mathcal{F}(\lang(\ctp))\) holds,
  and returns \(Some(\seq{s})\) for some \(\seq{s}\in \mathcal{F}(\lang(\ctp))\setminus \lang(\ctp)\) if the inclusion does not hold.
  Then the procedure below eventually terminates and
  returns the least CSTP \(\ctp\) such that \(\lang(\ctp)\supseteq \mathcal{F}(\lang(\ctp))\).
\begin{align*}
&  \dat \gets \emptyset;  \ctp \gets (a,\epsilon,\ldots,\epsilon)\land (\epsilon,\epsilon,\ldots,\epsilon) ;\\
&  \textbf{\upshape while}\ \mathit{true}\ \textbf{\upshape do }  \textbf{\upshape if }
         {\oracle(\ctp)= Some(\seq{s})}\ \textbf{\upshape then }
         {      (\dat \gets \dat \cup \{\seq{s}\};  \ctp \gets \CTPinf(\dat))}\ \textbf{\upshape else }
         {\textbf{\upshape return}\ \ctp  }
\end{align*}
\end{theorem}
The algorithm above first sets \(\ctp\) to \((a,\epsilon,\ldots,\epsilon)\land (\epsilon,\epsilon,\ldots,\epsilon)\) so that \(\lang(\ctp)=\emptyset\).
It then repeatedly adds to \(\dat\) the element returned by the oracle \(\oracle\) and invokes \(\CTPinf\).
A crucial property to guarantee the theorem is that
\(\CTPinf(\dat)\) always returns the \emph{least} CSTP \(\ctp\) such that \(\dat\subseteq \lang(\ctp)\).
\iffull
See Appendix for more details.
\else
See \cite{STPlong} for more details.
\fi
\begin{example}
  \label{ex:cstp-lfp}
  Consider the function \texttt{append} defined by:
  \begin{verbatim}
  let rec append l1 l2 = match l1 with [] -> l2 | x::l1' -> x::(append l1' l2).
\end{verbatim}
  A ternary relation \(I\) is an inductive invariant describing
  the input-output relation of \texttt{append} if and only if 
  \(I\) satisfies \(I \supseteq \mathcal{F}(I)\) for the following function \(\mathcal{F}\):
  \begin{align*}
    \mathcal{F}(R) =
    \set{(l_1,l_2,l_3) \mid
      (l_1=\epsilon\land l_2=l_3) \lor
      (l_1=xl'_1\land (l'_1,l_2,l'_3)\in R\land l_3=xl'_3)}.
  \end{align*}
  Here, we identify lists with words. 
  We wish to find the strongest CSTP \(\ctp\) such that \(\lang(\ctp)\) is such an \(I\), that is,
  the strongest \(\ctp\) satisfying \(\lang(\ctp)\supseteq \mathcal{F}(\lang(\ctp))\).
  Using the procedure above, we can derive \((y,z, yz)\) as the strongest \(\ctp\).
  See also Example~\ref{ex:solving-reva}, where the same procedure is used  
  to solve the \(\Reva\) example in Section~\ref{sec:intro}.
  \qed
\end{example}

Theorem~\ref{th:cstp-lfp} above would not hold if CSTPs were replaced with STPs, because
there is no guarantee that there exists a \emph{least} \(\tp\) such that \(\lang(\tp)\supseteq \mathcal{F}(\lang(\tp))\).
For example, suppose that \(\mathcal{F}\) is a constant function that always returns
\(\set{(a, ab, ab), (cd,c,cde)}\). Then, \(\lang(\tp_1)\supseteq \mathcal{F}(\lang(\tp_1))\) and
\(\lang(\tp_2)\supseteq \mathcal{F}(\lang(\tp_2))\) hold for \(\tp_1=(x,y,xz)\) and \(\tp_2=(x,y,yz)\),
but there is no STP \(\tp\) such that \(\lang(\tp)\subseteq \lang(\tp_1)\cap \lang(\tp_2)\) and
\(\lang(\tp)\supseteq \mathcal{F}(\lang(\tp))\). We can, however enumerate all the \emph{minimal} \(\tp\)'s
such that \(\lang(\tp)\supseteq \mathcal{F}(\lang(\tp))\).
\begin{theorem}
  \label{th:stp-lfp}
  Let \(\mathcal{F}\COL 2^{\Sigma^* \times \cdots \times \Sigma^*}\to 2^{\Sigma^* \times \cdots \times \Sigma^*} \) be a monotonic function,
  and suppose that there exists an algorithm \(g_{\mathcal{F}}\) that, given an STP \(\tp\),
 returns ``None'' if \(\lang(\tp)\supseteq \mathcal{F}(\lang(\tp))\) holds,
 and returns \(Some(\seq{s})\) for some \(\seq{s}\in \mathcal{F}(\lang(\tp))\setminus \lang(\tp)\) if the inclusion does not hold.
 Suppose also that \(\mathcal{F}(\emptyset)\ne \emptyset\), and an element
 \(\seq{s}\) such that \(\seq{s}\in\mathcal{F}(\emptyset)\) is computable.
  Then there exists an algorithm that enumerates all the STPs \(\tp\) such that 
  \(\lang(\tp)\supseteq \mathcal{F}(\lang(\tp))\) and, for every \(\tp'\),
 \(\lang(\tp')\subsetneq \lang(\tp)\) implies  \(\lang(\tp')\not\supseteq \mathcal{F}(\lang(\tp'))\).
\end{theorem}
\subsection{Decision Problems}
\label{sec:decision}
This section investigates decision problems on STPs.
We first study the membership and inclusion problems for STPs, which is useful for \(\CTPinf\).
Note that the membership and inclusion problems are
respectively NP-complete and undecidable for both NE-pattern languages and
E-pattern languages~\cite{nowotka_et_al:LIPIcs.CPM.2025.4}. 
\begin{theorem}
  \label{th:decision-problems}
  The following decision problems can be solved in polynomial time.
  \begin{enumerate}
  \item Given a tuple pattern \(\tp\), decide whether \(\tp\) is solvable.
  \item Given an STP \(\tp\) and a tuple \(\seq{s}\in\Alpha^*\times\cdots \times \Alpha^*\),
    decide whether \(\seq{s}\in \lang(\tp)\).
  \item Given an STP \(\tp_2\) and a tuple pattern \(\tp_1\), decide whether \(\lang(\tp_1)\subseteq \lang(\tp_2)\).
    \iffull
  \item Given two STPs \(\tp_1\) and \(\tp_2\), decide whether \(\lang(\tp_1)=\lang(\tp_2)\).
    \fi
  \end{enumerate}
\end{theorem}
\iffull
To solve (1), it suffices to reduce \(\tp\) by \(\pred\) and check whether it ends up with a trivial pattern \((x_1,\ldots,x_k)\).
Since the reduction strategy does not matter (as shown in Lemma~\ref{lem:pred-wconf} in Appendix,
\(\pred\) is weakly confluent up to permutations), and each reduction strictly decreases \(\measure{\tp}\), it can be checked in
time polynomial in \(\measure{\tp}\). For (3), it suffices to observe that if 
\(\tp_2\pred \tp_2'\) and \(\lang(\tp_1)\subseteq \lang(\tp_2)\), then we can
find a corresponding reduction \(\tp_1\preds \tp_1'\) such that  \(\lang(\tp_1')\subseteq \lang(\tp_2')\)
in polynomial time (cf. Lemma~\ref{lem:minimality-sim} in Appendix~\ref{sec:proofs}).
If there exists no such corresponding reduction, we can conclude \(\lang(\tp'_1)\not\subseteq \lang(\tp_2)\);
otherwise \(\tp_2\) is eventually reduced to a trivial pattern  \((x_1,\ldots,x_k)\), at which point we can conclude 
\(\lang(\tp_1)\subseteq \lang(\tp_2)\). (2) and (4) are immediate consequences of (3).
More details are found in Appendix~\ref{sec:proofs}.
\fi

We next prove the decidability of the theory of quantifier-free formulas containing formulas of the form \(\seq{w}\in \lang(\tp)\).
This decidability is important for constructing the oracle \(\oracle\) in Theorem~\ref{th:cstp-lfp},
and plays an important role in the applications to program verification discussed in Section~\ref{sec:chc}.
\begin{definition}
  \label{df:STP-formulas}
  The set of (quantifier-free) STP formulas, ranged over by \(\form\), is given by:
  \begin{align*}
   & \form ::= w_1=w_2 \mid w_1\ne w_2 \mid \seq{w}\in \lang(\tp) \mid \seq{w}\not\in \lang(\tp) \mid
    \form_1\land\form_2 \mid \form_1\lor\form_2\qquad w \in (\Alpha\cup \Vars)^*.
  \end{align*}
  Here, \(\tp\) ranges over the set of STPs.
\end{definition}
For example, \((x, y, z)\in \lang(\tp)\land w=yz\land (xy, y, w)\not\in \lang(\tp)\)
(which is the negation of ``\((x, y, z)\in \lang(\tp)\land w=yz\) implies
\((xy, y, w)\in \lang(\tp)\)''),
where \(\tp = (u,v,uv)\) is an STP formula. The semantics of STP formulas is defined in an obvious manner,
and is therefore deferred to \iffull the Appendix. \else the long version~\cite{STPlong}. \fi
The only point to note is that the variables occurring in \(\tp\) are implicitly bound
within \(\lang(\tp)\); for example,  \((x, y, z)\in \lang(u,v,uv)\) is equivalent to
\((x, y, z)\in \lang(y,x,yx)\) (although we usually avoid variable overloading to prevent confusion).
\iffull
\begin{remark}
  The definition of quantifier-free STP formulas above can be actually simplified to:
  \begin{align*}
   & \form ::= \seq{w}\in \lang(\tp) \mid \seq{w}\not\in \lang(\tp) \mid
    \form_1\land\form_2 \mid \form_1\lor\form_2, \end{align*}
  because \(w_1=w_2\) and \(w_1\ne w_2\) can be expressed by \((w_1,w_2)\in \lang(x,x)\) and \((w_1,w_2)\not\in \lang(x,x)\) respectively. \qed
\end{remark}
\fi
According to the result on word equations~\cite{Plandowski,10.1145/337244.337255},
the satisfiability of STP formulas without STP membership constraints
(i.e., Boolean combinations of formulas of the form \(w_1=w_2\)) is in PSPACE. 
Thus, we have:
\begin{theorem}
  \label{th:satisfiability}
  Given an STP formula \(\form\), one can effectively construct an equi-satisfiable STP formula \(\form'\)
such that \(\form'\) 
contains no subformulas of the form \((w_1,\ldots,w_k)\in \lang(\tp)\) or \((w_1,\ldots,w_k)\not\in \lang(\tp)\), and the size of \(\form'\) is polynomial in that of \(\form\).
  Therefore, the satisfiability of STP formulas is in PSPACE.
\end{theorem}
Encoding of a formula \(\seq{w}\not\in \lang(\tp)\) into word equations relies on
the solvability of \(\tp\). In fact, if we allowed arbitrary tuple patterns, then the satisfiability 
would become undecidable, because
\(\exists x\in\Alpha^*.(x \in \lang(p_1)\land x\not\in \lang(p_2))\) if and only if 
\(\lang(p_1){\not\subseteq} \lang(p_2)\) and the latter is undecidable~\cite{JIANG199553,nowotka_et_al:LIPIcs.CPM.2025.4}.
 \section{Extensions and Variations of Solvable Tuple Patterns}
\label{sec:ext}

This section discusses extensions and variations of STPs. Section~\ref{sec:ext-rev} extends STPs
with the reverse pattern \(x^R\), which is needed to handle the \(\Reva\) example in Section~\ref{sec:intro}.
Section~\ref{sec:set} discusses variations of STPs to reason about sets and 
multisets.

 \subsection{Reverse}
\label{sec:ext-rev}
Let us extend the set of pattern expressions (ranged over by \(\pat\)) to \((\Alpha\cup\Vars\cup\Vars^R)^*\),
where \(\Vars^R = \set{\rpat{x} \mid x\in \Vars}\).
Here, \(\rpat{x}\) represents the set of the reverse of strings represented by \(x\).
We extend \(\rpat{(\cdot)}\) to the operation on patterns, by \(\rpat{a}=a\), \(\rpat{(\rpat{x})}= x\),
and \(\rpat{(\pat_1\pat_2)}=\rpat{\pat_2}\cdot\rpat{\pat_1}\), so that
  the set of extended pattern expressions is closed under substitutions.
For example, \([ay/x]\rpat{x} = \rpat{(ay)}=\rpat{y}a\).

We extend  the tuple pattern inference algorithm with the following rules:
\typicallabel{R-RPrefix}
\infrule[R-RPrefix]{\dat[*][i]\ne\seq{\epsilon}\andalso \dat[*][j]=\rev{\dat[*][i]}\cdot\seq{s} \andalso x_j'\mbox{ fresh}\andalso \seq{x}=(x_1,\ldots,x_m)}
        {(\tp, \Subst{\seq{x}}{\dat}) \red ([\rpat{x_i}\cdot x_j'/x_j]\tp,
          \Subst{(x_1,\ldots,x_{j-1},x_j', x_{j+1},\ldots, x_m)}{\replace{\dat}{j}{\seq{s}}})}
\infrule[R-RSuffix]{\dat[*][i]\ne\seq{\epsilon}\andalso \dat[*][j]=\seq{s}\cdot\rev{\dat[*][i]} \andalso x_j'\mbox{ fresh}\andalso \seq{x}=(x_1,\ldots,x_m)}
        {(\tp, \Subst{\seq{x}}{\dat})\red ([x_j'\cdot \rpat{x_i}/x_j]\tp,
          \Subst{(x_1,\ldots,x_{j-1},x_j', x_{j+1},\ldots, x_m)}{\replace{\dat}{j}{\seq{s}}})}
Here, \(\dat[*][i]^R\) denotes the pointwise extension of the reverse operation.
Accordingly, we extend solvable tuple patterns by adding the following rules:
\iffull
\begin{multicols}{2}
\infrule[PR-RPrefix]{p_j = \rpat{p_i}\cdot p'_j\andalso p_i\ne \epsilon}{(p_1,\ldots,p_n) \pred
  (p_1,\ldots,p_{j-1},p'_j,p_{j+1},\ldots,p_n)}
\infrule[PR-RSuffix]{p_j = p'_j\cdot \rpat{p_i}\andalso p_i\ne \epsilon}{(p_1,\ldots,p_n) \pred
  (p_1,\ldots,p_{j-1},p'_j,p_{j+1},\ldots,p_n)}
\end{multicols}
\else
\begin{multicols}{2}
\infrule{p_j = \rpat{p_i}\cdot p'_j\andalso p_i\ne \epsilon}{(p_1,\ldots,p_n) \pred
  (p_1,\ldots,p_{j-1},p'_j,p_{j+1},\ldots,p_n)}
\infrule{p_j = p'_j\cdot \rpat{p_i}\andalso p_i\ne \epsilon}{(p_1,\ldots,p_n) \pred
  (p_1,\ldots,p_{j-1},p'_j,p_{j+1},\ldots,p_n)}
\end{multicols}
\fi

The theorems in Sections~\ref{sec:tpinf} and \ref{sec:stp} continue to hold for this extension.
For the decidability and complexity of the satisfiability problem (Theorem~\ref{th:satisfiability}), we rely on known results
on the satisfiability of equations over free monoids with involution~\cite{DIEKERT2005105}.

\begin{example}
  \label{ex:rev}
  Let \(\dat\) be 
      \( \left(\begin{array}{rrr}
    ab &\ cd &\ bacd \\
    bc &  da & cbda
  \end{array}\right).
  \)
  The pattern \((x, y, x^Ry)\) is inferred as follows.
  \iffull
  \begin{align*}
    &  \Big((x_1, x_2, x_3), \left(\begin{array}{rrr}
      x_1 & x_2 & x_3\\
    ab &\ cd &\ bacd \\
    bc &  da & cbda
  \end{array}\right)\Big)\\
    &  \red \Big((x_1, x_2, x_1^Rx_3'), 
    \left(\begin{array}{rrr}
    x_1 &\ x_2 &\ x_3' \\
    ab &\ cd &\ cd \\
    bc &  da & da
    \end{array}\right)\Big) & \mbox{ (\rn{R-RPrefix})}\\
    &  \red \Big((x_1, x_2, x_1^Rx_2x_3''), 
    \left(\begin{array}{rrr}
    x_1 &\ x_2 &\ x_3'' \\
    ab &\ cd &\ \epsilon \\
    bc &  da & \epsilon
    \end{array}\right)\Big) & \mbox{ (\rn{R-Prefix})}\\
    &  \red \Big((x_1, x_2, x_1^Rx_2), 
    \left(\begin{array}{rrr}
    x_1 &\ x_2 \\
    ab &\ cd \\
    bc &  da 
    \end{array}\right)\Big) & \mbox{ (\rn{R-Epsilon})}.
  \end{align*}
  \else
  \begin{align*}
    &  ((x_1, x_2, x_3),  \left(\begin{array}{rrr}
      x_1 & x_2 & x_3\\
    ab &\ cd &\ bacd \\
    bc &  da & cbda
  \end{array}\right))
    \red ((x_1, x_2, x_1^Rx_3'), 
    \left(\begin{array}{rrr}
    x_1 &\ x_2 &\ x_3' \\
    ab &\ cd &\ cd \\
    bc &  da & da
    \end{array}\right)) \\& \red ((x_1, x_2, x_1^Rx_2x_3''), 
    \left(\begin{array}{rrr}
    x_1 &\ x_2 &\ x_3'' \\
    ab &\ cd &\ \epsilon \\
    bc &  da & \epsilon
    \end{array}\right)) \red ((x_1, x_2, x_1^Rx_2), 
    \left(\begin{array}{rrr}
    x_1 &\ x_2 \\
    ab &\ cd \\
    bc &  da 
    \end{array}\right)). \end{align*}
  \fi
\end{example}

\begin{remark}
\label{rem:sort}
One may also be tempted to extend patterns with \(x^S\), which represents the sequence obtained by sorting \(x\).
This is problematic, however, because \((\cdot)^S\) does not commute with concatenation: \((y\cdot z)^S\ne y^S\cdot z^S\).
One way to reason about sorted sequences is to introduce a new class of patterns called \emph{sort patterns}, given by:
\(p_S ::= x \mid  u\) and \(u ::= x^S \mid u_1 \mergeop u_2\).
Here, \(u_1\mergeop u_2\) denotes the set of sorted sequences obtained by merging sorted sequences
represented by \(u_1\) and \(u_2\).
Sort patterns can naturally be extended to tuple sort patterns. For example,
\((x,x^S)\) denotes the set of pairs, whose second element is obtained by sorting the first element,
and \((x, y^S, x^S\mergeop y^S)\) denotes the set of triples \((s_1,s_2,s_3)\),
where \(s_3\) is obtained by sorting \(s_1\) and merging it with \(s_2\).
We can then design an inference algorithm for (solvable) tuple sort patterns in a manner similar to ordinary tuple patterns.
The details will be described in a separate paper. \qed
\end{remark}

\subsection{Beyond Sequence Patterns}
\label{sec:set}

  We have so far considered tuple patterns consisting of words, which
  form a \emph{free} monoid.  Almost the same technique can be applied
  to represent and infer relations among other similar algebraic
  structures such as multisets and sets.  Note that multisets form a
  free \emph{commutative} monoid, while sets are obtained by further imposing
  idempotency on the binary operation.

To deal with multisets, we just need to interpret the empty pattern and the pattern concatenation 
as the empty set and the multiset union, respectively, and identify patterns up to permutations
(for example, \(xyz = yzx\)).
The inferences and reduction rules for tuple patterns can be adjusted accordingly, as follows.
\iffull
\infrule[R-Subset]{\dat[*][i]\ne\seq{\epsilon}\andalso
  \forall k.\dat[k][j]\supseteq \dat[k][i]\andalso \seq{s}=\dat[*][j]\setminus \dat[*][i]
  \andalso x_j'\mbox{ fresh}}
        {(\tp, \Subst{(x_1,\ldots,x_m)}{\dat})\red ([x_ix_j'/x_j]\tp,
          \Subst{(x_1,\ldots,x_{j-1},x_j', x_{j+1},\ldots, x_m)}{\replace{\dat}{j}{\seq{s}}})}
\infrule[R-CSubset]
                {\forall k.a\in \dat[k][j]\andalso \seq{s}=\dat[*][j]\setminus \set{a}\andalso a\in\Alpha}
                {(\tp, \Subst{(x_1,\ldots,x_m)}{\dat})\red ([ax_j'/x_j]\tp, \Subst{(x_1,\ldots,x_{j-1},x_j',x_{j+1},\ldots, x_m)}{\replace{\dat}{j}{\seq{s}}})}

\infrule[R-Empty]{\dat[*][j]=\seq{\emptyset}}
        {(\tp, \Subst{(x_1,\ldots,x_m)}{\dat})\red
          ([\epsilon/x_j]\tp, \Subst{(x_1,\ldots,x_{j-1},x_{j+1},\ldots, x_m)}{\remove{\dat}{j}})}

 \infrule[PR-Subset]{p_j = p_ip'_j\andalso p_i\ne \epsilon}{(p_1,\ldots,p_n) \pred
  (p_1,\ldots,p_{j-1},p'_j,p_{j+1},\ldots,p_n)}
                
\infrule[PR-CSubset]{p_j = a p'_j\andalso a\in \Alpha}
  {(p_1,\ldots,p_n) \pred
    (p_1,\ldots,p_{j-1},p'_j,p_{j+1},\ldots,p_n)}
\infrule[PR-Empty]{p_j=\epsilon}  
  {(p_1,\ldots,p_n) \pred
    (p_1,\ldots,p_{j-1},p_{j+1},\ldots,p_n)}

\else
For example, \rn{R-Prefix} and \rn{PR-Prefix} are replaced by:
\infrule[R-Subset]{\dat[*][i]\ne\seq{\epsilon}\andalso
  \forall k.\dat[k][j]\supseteq \dat[k][i]\andalso \seq{s}=\dat[*][j]\setminus \dat[*][i]
  \andalso x_j'\mbox{ fresh}}
        {(\tp, \Subst{(x_1,\ldots,x_m)}{\dat})\red ([x_ix_j'/x_j]\tp,
          \Subst{(x_1,\ldots,x_{j-1},x_j', x_{j+1},\ldots, x_m)}{\replace{\dat}{j}{\seq{s}}})}
\infrule[PR-Subset]{p_j = p_ip'_j\andalso p_i\ne \epsilon}{(p_1,\ldots,p_n) \pred
  (p_1,\ldots,p_{j-1},p'_j,p_{j+1},\ldots,p_n)}
Here, \(\dat[*][j]\setminus \dat[*][i]\) denotes the pointwise extension of the multiset difference.
The full set of rules is given in the longer version.
\fi
Most of the results for STPs and CSTPs carry over to their multiset versions (called SMTPs and CSMTPs),
including the soundness and completeness of the inference algorithm, Theorem~\ref{th:cstp-lfp},
and Theorem~\ref{th:satisfiability} (where word equations are replaced with multiset equations,
and the complexity of satisfiability is changed to NP).
A notable exception is that the minimality property (Theorem~\ref{th:minimality}) fails, but it is not an obstacle
to the application to program verification discussed in the next section.
Similarly, we can deal with sets by interpreting concatenation as disjoint union.

\begin{example}
  Let \(\dat\) be:
  \(\left(\begin{array}{rrr}
    \set{a,a} &\ \set{a,b} &\ \set{a,a,b}\\
    \set{a,b} &\ \set{b,c} & \set{a,b,c}
  \end{array}\right)
  \).
  The pattern \((xz,yz,xyz)\) (which means that the third set is a superset of the first and second sets,
  where \(z\) is the intersection of the first and second sets) is inferred as follows.
    \begin{align*}
      & \Big((x_1, x_2, x_3), \left(\begin{array}{rrr}
        x_1 & x_2 & x_3\\
    \set{a,a} &\ \set{a,b} &\ \set{a,a,b}\\
    \set{a,b} &\ \set{b,c} & \set{a,b,c}
  \end{array}\right)\Big)\red \Big((x_1, x_2, x_1x_3'), \left(\begin{array}{rrr}
        x_1 & x_2 & x_3'\\
    \set{a,a} &\ \set{a,b} &\ \set{b}\\
    \set{a,b} &\ \set{b,c} & \set{c}
      \end{array}\right)\Big) \\
      &  \red \Big((x_1, x_3'x'_2, x_1x_3'), \left(\begin{array}{rrr}
        x_1 & x_2' & x_3'\\
    \set{a,a} &\ \set{a} &\ \set{b}\\
    \set{a,b} &\ \set{b} & \set{c}
    \end{array}\right)\Big) \red \Big((x_2'x'_1, x_3'x'_2, x_2'x'_1x_3'), \left(\begin{array}{rrr}
        x'_1 & x_2' & x_3'\\
    \set{a} &\ \set{a} &\ \set{b}\\
    \set{a} &\ \set{b} & \set{c}
    \end{array}\right)\Big). \end{align*}
\end{example}
\iffull
\fi

  \begin{remark}
    \label{rem:treepat}
    Similar techniques can also be applied to binary trees, which are obtained by dropping associativity of the binary operation and thus form a \emph{free magma}.
Using the tree version
  of STPs, we can, for example, represent and infer tree relations such as
  \((x, N(x,y))\), where \(N\) is the binary tree constructor; this means that
  the first element is the left subtree of the second. The formalization
  and applications of this
  tree version of STPs to automated program verification are left for future work. \qed
\end{remark}

  \begin{remark}
    \label{rem:freeness}
      While the theory of STPs can be applied to other similar
      \emph{free} algebraic structures as discussed above, freeness
       is important for STPs.
      Many of the theorems in Section~\ref{sec:stp} rely on
       properties such as the following:
      if \(\lang(\pat') \subseteq \lang(a\pat_1)\),
      then \(\pat' = a\pat'_1\) with \(\lang(\pat'_1)\subseteq \lang(\pat_1)\)
      for some \(\pat'_1\);
      and if \(\lang(\pat'_1,\pat'_2)\subseteq \lang(\pat_2\pat_3, \pat_2)\),
      then \(\pat'_1=\pat'_2\pat'_3\) with
      \(\lang(\pat'_3)\subseteq \lang(\pat_3)\) for some \(\pat'_3\).\footnote{
        See the longer version~\cite{STPlong} for a table summarizing the properties used in the proof of each main theorem.}
      Such properties do not necessarily hold for arbitrary monoids.
      For example, consider a monoid generated by \(a, b\), and \(c\)
      with identity element \(\epsilon\), subject to the equation
      \(ab=ca\) and \(ac=ba\). In this monoid, \(\lang(ya)\subseteq \lang(ax)\) holds,
      but there is no pattern \(\pat\) such that \(ya=a\pat\).

      One may wonder whether it is possible to develop a general, unifying theory of STPs for a certain class of
      free algebraic structures, and then instantiate the theory to each structure.
      That does not seem straightforward, however. The proof of minimality (Theorem~\ref{th:minimality})
      relies on word-specific properties, such as the fact that \(\pat_1\pat_2=\pat'_1\pat'_2\) implies \(\pat_1\) and \(\pat'_1\) are
      in a prefix relation. Such properties do not hold for free commutative monoids, and that is why minimality fails
      for multisets. Moreover, the forms of characteristic data differ between words and multisets.
      Thus, developing such a unifying theory is outside the scope of this paper and left for future work. \qed
\end{remark}
 
\section{Applications to Automated Verification of List-Manipulating Programs}
\label{sec:chc}

This section discusses how to apply the (C)STP inference to
automated verification of programs manipulating list-like data
structures.  We focus on CHC-based automated program
verification~\cite{DBLP:conf/sas/BjornerMR13}, where program
verification problems are reduced to the satisfiability problem for
constrained Horn clauses (CHCs, a.k.a. constraint logic programs), and
passed to an automated CHC solver.
The CHC satisfiability problem is undecidable in general, but a number of CHC solvers that automatically
solve it have been developed, though they are inevitably incomplete. These solvers typically support
integer and real arithmetic, arrays, and ADTs. An annual competition on CHC solvers called CHC-COMP is held. The CHC-based approach provides a
uniform, language-agnostic approach to automated verification.  It has
recently been studied extensively and applied to various languages,
ranging from low-level languages like C and
LLVM~\cite{DBLP:conf/tacas/GurfinkelKN15,DBLP:conf/esop/TomanSSI020},
to high-level languages like Java~\cite{DBLP:conf/cav/KahsaiRSS16},
Rust~\cite{DBLP:journals/toplas/MatsushitaTK21}, and
OCaml~\cite{Unno09PPDP,DBLP:journals/jar/ChampionCKS20}.
As noted in Section~\ref{sec:intro}, however, CHC solvers over ADTs have so far been unsatisfactory.

Below we first introduce a class of CHCs with word constraints,
and show that it is decidable whether a system of CHCs (i.e., a finite set of CHCs) in that class
has a model describable using CSTPs in Section~\ref{sec:chc-words}.
This result is directly applicable to the \(\Reva{}\) example in Section~\ref{sec:intro}.
Sections~\ref{sec:combination} and \ref{sec:app-multiset} discuss
further enhancement of the STP-based CHC solving.

\subsection{Constrained Horn Clauses over Words}
\label{sec:chc-words}

A \emph{constrained Horn clause} (CHC) is a Horn clause extended with
constraints. It is of the form: \(P_1(\seq{y}_1)\land \cdots \land
P_n(\seq{y}_n)\land \cform \imp P(\seq{x})\), or \(P_1(\seq{y}_1)\land
\cdots \land P_n(\seq{y}_n)\land \cform \imp \false\), where \(n\) may be
\(0\).  We call a clause of the former form a \emph{definite clause},
and one of the latter form a \emph{goal clause}.  Here, \(P\) and
\(P_i\) are predicate variables, and \(\cform\) ranges over the set of
constraint formulas over certain domains; the first-order variables
\(\seq{y}_i\) and \(\seq{x}\) are assumed to be implicitly universally
quantified in each clause.  A clause that can be normalized to the
above form is also called a CHC. For example, \(P(x,1)\imp x<1\) is
also a CHC, as it is equivalent to \(P(x,y)\land (y=1\land x\ge 1)\imp
\false\).  We often just write \(P(t_1,\ldots,t_n)\) for the clause
\(x_1=t_1\land \cdots \land x_n=t_n\imp P(x_1,\ldots,x_n)\), by
eliminating redundant variables and omitting ``\(\true\imp\)''.

The \emph{CHC satisfiability problem} (or, the CHC problem) asks, given a set of CHCs,
whether there exists a model (i.e., an interpretation of predicate variables) that makes all the clauses valid.
For example, consider the following system of CHCs over integer arithmetic.
\begin{align*}
  &  \Plus(0, y, y).\qquad
\Plus(x-1, y, z)\land x\ne 0 \imp \Plus(x,y, z+1). \qquad
\Plus(x, y, z) \imp z\ge y.
\end{align*}
It is satisfiable. Indeed, \(\set{\Plus\mapsto \set{(x,y,z)\mid x\ge 0 \land z=x+y}}\) satisfies all the clauses.
Henceforth, we often use a formula to express a model, like  \(\Plus(x,y,z) \equiv x\ge 0 \land z=x+y\).
A model is not unique in general; for the example above, \(\Plus(x,y,z) \equiv z\ge y\) is also a model.

In this section, instead of CHCs over integer arithmetic,
we consider CHCs on words over \(\Alpha\). Thus, 
each variable ranges over the set \(\Alpha^*\) of words, and
constraints \(\cform\) are drawn from the set of quantifier-free STP formulas
as defined in Definition~\ref{df:STP-formulas}.

\begin{example}
  \label{ex:reva-as-chc}
  The \(\Reva\) example in Section~\ref{sec:intro} is expressed as the following system of CHCs
  \(\CHC_{\Reva}\).
  \begin{align*}
 &\Reva(\epsilon, l_2, l_2). \qquad
 \Reva(l'_1, x l_2, l_3) \land |x|=1 \imp \Reva(x l'_1, l_2, l_3). \\
&  \Reva(l_1,l_2,l_3)\land \Reva(l_3, \epsilon, l_4)\land \Reva(l_2,l_1,l_5)\imp l_4=l_5.
  \end{align*}
  Here, the constraint \(|x|=1\) can be expressed by \(\bigvee_{a\in\Alpha} x=a\), assuming
  that the alphabet \(\Alpha\) is finite.
  
  Let \(\tp = (x,y,x^R,y)\). Then \(\set{\Reva\mapsto \lang(x, y, x^Ry)}\) is a model of the CHCs above.
  \qed
\end{example}

We say that
a system \(\CHC\) of CHCs on words has a \emph{CSTP-model} if
there exist \(\ctp_1,\ldots,\ctp_n\) such that \(\set{P_1\mapsto \lang(\ctp_1),\ldots,P_n\mapsto \lang(\ctp_n)}\)
is a model of \(\CHC\).
The following is the main result of this section.
\begin{theorem}
  \label{th:decidability-chc}
  Given a system \(\CHC\) of CHCs on words, it is decidable whether \(\CHC\) has a CSTP-model.
\end{theorem}
\begin{proof}[Proof Sketch.]
  Let \(\CHC = \CHC_D\cup\CHC_G\), where \(\CHC_D\) and \(\CHC_G\) respectively consist of definite and
  goal clauses. 
  By Theorem~\ref{th:cstp-lfp}, the least CSTP-model \(\mathcal{M}_0 = \set{P_1\mapsto \lang(\ctp_1),
    \ldots,P_n\mapsto \lang(\ctp_n)}\) of \(\CHC_D\) exists and it is computable (note that the algorithm for
  \(\oracle\) can be constructed by using Theorem~\ref{th:satisfiability}).
  Then, \(\CHC\) has a CSTP-model if and only if \(\mathcal{M}_0\) is a model of \(\CHC_G\).
  The latter is decidable by Theorem~\ref{th:satisfiability}.
\end{proof}

An analogous result holds also when restricted to STP-models. We say that
a system \(\CHC\) of CHCs on words has an \emph{STP-model} if
there exist STPs \(\tp_1,\ldots,\tp_n\) such that \(\set{P_1\mapsto \lang(\tp_1),\ldots,P_n\mapsto \lang(\tp_n)}\)
is a model of \(\CHC\).  Here, we add a special STP \(\tp_\emptyset\) that denotes the empty set (i.e., \(\lang(\tp_\emptyset)=\emptyset\)).
\begin{theorem}
  \label{th:decidability-chc-stp}
  Given a system \(\CHC\) of CHCs on words, it is decidable whether \(\CHC\) has an STP-model.
\end{theorem}
To prove the theorem, it suffices to use Theorem~\ref{th:stp-lfp} instead of Theorem~\ref{th:cstp-lfp}.
\iffull
\begin{proof}[Proof Sketch.]
  Let \(\CHC = \CHC_D\cup\CHC_G\), where \(\CHC_D\) and \(\CHC_G\) respectively consist of definite and
  goal clauses. 
  By Theorems~\ref{th:stp-lfp} and \ref{th:satisfiability},
  we can enumerate all the minimal STP-models \(\mathcal{M}_0 = \set{P_1\mapsto \lang(\tp_1),
    \ldots,P_n\mapsto \lang(\tp_n)}\) of \(\CHC_D\). 
  Then, it suffices to check whether one of the minimal STP-models is also a model of \(\CHC_G\).
\end{proof}
\fi

The unsatisfiability of CHCs (where models are not restricted to CSTPs) is also semi-decidable:
\begin{theorem}
  \label{th:refutational-completeness}
  There exists a procedure which, given a system \(\CHC\) of CHCs on words,
  eventually outputs ``Unsat'' whenever \(\CHC\) is not satisfiable.
\end{theorem}
\begin{proof}
  This follows from the fact that \(\CHC\) is unsatisfiable if and only if there is a resolution proof.
  Thanks to the decidability of word equations, we just need to enumerate a candidate resolution proof,
  and check whether it is a valid proof.
\end{proof}
Note, however, that  the satisfiability of CHCs is undecidable in general
if models are not restricted to CSTPs.
To see this, notice that a natural number \(n\) can be encoded as the word \(a^n\).

Below, we demonstrate how we can automatically obtain a CSTP-model for the \(\Reva\)-example.
For the sake of simplicity, we omit constant patterns.
\begin{example}
  \label{ex:solving-reva}
  Recall the CHCs   \(\CHC_{\Reva}\) in Example~\ref{ex:reva-as-chc}.
  We first compute the least model of the definite clauses (i.e., the first two clauses) \(\CHC_D\),
  based on the algorithm in Theorem~\ref{th:cstp-lfp}.
  First, we set \(\dat_0=\emptyset\) and \(\ctp_0=(a,\epsilon,\epsilon)\land (\epsilon,\epsilon,\epsilon)\),
  and check whether \(\ctp_0\) satisfies \(\CHC_D\).
Using the algorithm in Theorem~\ref{th:satisfiability}, we obtain an element
  \(\seq{s}_0=(\epsilon, ab, ab)\) that should belong to the least model.

  Then, we set \(\dat_1=\set{(\epsilon,ab,ab)}\), and get
  \(\ctp_1 = \CTPinf(\set{(\epsilon,ab,ab)}) = (\epsilon,x,x)\).
  This does not satisfy the second clause. Indeed, the formula obtained by 
   negating the second clause:
  \((l'_1,xl_2,l_3)\in \lang(\ctp_1)\land |x|=1\land (xl'_1,l_2,l_3)\not\in \lang(\ctp_1)\)
  has a model \((x,l'_1,l_2,l_3) = (a, \epsilon,b,ab)\), from which we obtain
  a new element \(\seq{s}_1=(a, b, ab)\).
By setting \(\dat_2=\set{(\epsilon,ab,ab), (a,b,ab)}\), we obtain
  \(\ctp_2 = (x,y,xy)\land (x,y,x^Ry)\).
  Since \((l'_1,xl_2,l_3)\in \lang(\ctp_2)\land |x|=1\land (xl'_1,l_2,l_3)\not\in \lang(\ctp_2)\)
  has a model \((x,l'_1,l_2,l_3) = (a, b,\epsilon,ba)\), we obtain the new element:
  \(\seq{s}_2 = (ab, \epsilon, ba)\). 

  By setting \(\dat_3=\set{(\epsilon,a,a),(a,b,ab),(ab,\epsilon,ba)}\),
  we obtain \(\cstp_3 = (x,y,x^Ry)\).
  Now, \(\cstp_3\) is a model of \(\CHC_D\). Since \(\cstp_3\) is also a model of
  the goal clause, we can conclude that   \(\CHC_{\Reva}\) is satisfiable,
  and hence also that the \texttt{reva} program in Section~\ref{sec:intro}
  does not suffer from assertion failures. \qed
\end{example}

\iffull
In Appendix~\ref{sec:pcstp}, we extend CSTP-models to
\emph{piecewise} CSTP-models, and apply them to verification of functional queues.
\else
In the longer version~\cite{STPlong}, we extend CSTP-models to
\emph{piecewise} CSTP-models, and apply them to verification of functional queues.
\fi

\iffull
\begin{remark}
  \label{rem:CSTPvsSTP}
  In Example~\ref{ex:solving-reva} above, if we replace \(\CTPinf\) used inside the algorithm of Theorem~\ref{th:cstp-lfp}
  with \(\TPinf\),
  then the inference may fail.
  Consider the state \(\dat_2=\set{(\epsilon,ab,ab), (a,b,ab)}\) above.
  If \(\TPinf\) is invoked instead of \(\CTPinf\), we may obtain \(\tp_2 = (x,y,xy)\) as a minimal (but not the least) \(\tp\)
  such that \(\dat_2\subseteq \lang(\tp)\).
  Then, as a model of \((l'_1,xl_2,l_3)\in \lang(\tp_2)\land |x|=1\land (xl'_1,l_2,l_3)\not\in \lang(\tp_2)\),
  we may obtain \((x,l'_1,l_2,l_3)=(c, ab, \epsilon, abc)\), which yields \(\seq{s}_2'=(cab, \epsilon,abc)\).
  
  By setting \(\dat_3'=\set{(\epsilon,a,a),(a,b,ab),(cab,\epsilon,abc)}\), we obtain \(\TPinf(\dat_3')=(x,y,zy)\).
  Note that \(\lang(x,y,x^Ry)\subsetneq \lang(x,y,zy)\).
  Thus, we have failed to infer the least model \(\set{\Reva\mapsto (x,y,x^Ry)}\) of \(\CHC_D\),
  and therefore, we have also failed to prove that \(\CHC_{\Reva}\) is satisfiable.

  We can still prove the satisfiability of \(\CHC_{\Reva}\) by using only STPs, based on Theorem~\ref{th:decidability-chc-stp}
  (which relies on Theorem~\ref{th:stp-lfp} instead of Theorem~\ref{th:cstp-lfp}).
  \qed
\end{remark}
\fi

\begin{remark}
  \label{rem:mutable-lists}
    The results above can be directly applied to automated verification of
  programs manipulating \emph{functional} lists. To handle programs 
  manipulating \emph{mutable} linked lists, we can use
  the technique of RustHorn~\cite{DBLP:journals/toplas/MatsushitaTK21} to
  convert {mutable} linked lists into functional lists, and then apply our
  technique. Other mutable data structures can be mapped to
  (functional) trees using the technique of
  RustHorn~\cite{DBLP:journals/toplas/MatsushitaTK21}.
We can then either (i) apply the tree version of STPs mentioned in Remark~\ref{rem:treepat}, or (ii) abstract trees to lists and then apply our results on CHCs over words. We expect, for example, that programs over binary search trees can be verified using the latter approach, by abstracting a binary search tree to the list of its elements. Skip lists~\cite{DBLP:journals/cacm/Pugh90} could also be handled in a similar manner: despite the name “lists,” they would first be mapped to trees reflecting their physical structure, and then abstracted to a tuple of lists \((\ell_1,\ldots,\ell_k)\),
where \(\ell_i\) is the list of elements at level \(i\).
The structural properties of a skip list---namely, 
that each \(\ell_i\) is a sorted sequence and that \(\ell_i\) is a subsequence of
\(\ell_j\) for any \(j<i\)---can be described using the pattern of sorted
sequences mentioned in Remark~\ref{rem:sort}.
At present, however, these applications remain speculative. A concrete formalization and implementation are left for future work.
\qed
\end{remark}

\subsection{Combination with CHC Solving for Integer Arithmetic}
\label{sec:combination}

The CHCs obtained from list-manipulating programs typically contain
integer constraints in addition to word (or list) constraints.
The CSTPs alone are not always sufficient to express appropriate invariants
for such CHCs.
Let us recall the second example (involving \texttt{take} and \texttt{drop})
in Section~\ref{sec:intro}. It can be expressed as the satisfiability
of the following CHCs:
\begin{align*}
&  \Pred{Take}(0, l, \epsilon).
  \qquad n\ne 0 \imp \Pred{Take}(n, \epsilon, \epsilon). \qquad n\ne 0\land \Pred{Take}(n-1, l', r)\land |x|=1 \imp \Pred{Take}(n, xl', xr).\\
    & \Pred{Drop}(0, l, l).
  \qquad n\ne 0\land \Pred{Drop}(n, \epsilon, \epsilon). \qquad n\ne 0\land \Pred{Drop}(n-1, l', r) \imp \Pred{Drop}(n, xl', r).\\
  &\Pred{Take}(n,l,r_1)\land \Pred{Drop}(n,l,r_2)\land \Pred{Append}(r_1,r_2,l')
  \land l\ne l' \imp \false.
\end{align*}
Here, we have omitted the clauses for \(\Pred{Append}\).
As before we have encoded lists into words, but the resulting CHCs also contain
integer constraints.
Ignoring the integer constraints,
the least CSTP model of the definite clauses obtained by \(\CTPinf\) is:
\(\Pred{Take}(n,l_1,l_2) \equiv \exists l_3.l_1=l_2l_3\),
\(\Pred{Drop}(n,l_1,l_2) \equiv \exists l_3.l_1=l_3l_2\), and
\(\Pred{Append}(l_1,l_2,l_3)\equiv l_3=l_1l_2\).
It is not strong enough to satisfy the goal clause (i.e., the last clause above).

To address the issue,
we use a CHC solver over integers to strengthen the model
of definite clauses. We abstract lists to their lengths and obtain the following
abstract version of the CHCs.
\begin{align*}
  &  \Pred{Take'}(0, l, 0).
  \qquad n\ne 0 \imp \Pred{Take'}(n, 0, 0). \qquad n\ne 0\land \Pred{Take'}(n-1, l', r) \imp \Pred{Take'}(n, 1+l', 1+r).\\
    & \Pred{Drop'}(0, l, l).
  \qquad n\ne 0\land \Pred{Drop'}(n, 0, 0). \qquad n\ne 0\land \Pred{Drop'}(n-1, l', r) \imp \Pred{Drop'}(n, 1+l', r).\\
  &\Pred{Take'}(n,l,r_1)\land \Pred{Drop'}(n,l,r_2)\land \Pred{Append'}(r_1,r_2,l')
  \land l\ne l' \imp \false.
\end{align*}
In general, given a system of CHCs over lists \(\CHC = \CHC_D\cup \CHC_G\), 
  (i) the list constructors \([\,]\) and \(\mathbin{::}\)  are replaced
 with their length abstractions \(0\) and \(\lambda (x,y)=1+y\), 
(ii) the list equalities (which implicitly occur in the clause bodies) are replaced
with integer equalities; and
(iii) the list inequalities \(l\ne l'\) in the body of  \emph{goal} clauses \(\CHC_G\)
are replaced
with the integer inequalities, and those in
the bodies of \emph{definite} clauses \(\CHC_D\) are
  replaced with \(\true\).
  
  The above abstraction ensures that a model of the abstract version \(\CHC'_D\) of
  definite clauses yields that of the original definite clauses \(\CHC_D\):
  if \(P'(\seq{x},\seq{l}')\equiv \varphi\) 
  satisfies \(\CHC'_D\), then
  \(P(\seq{x}, \seq{l})\equiv [\len{\seq{l}}/\seq{l}']\varphi\)
  satisfies \(\CHC_D\)
 (where we assume
  \(\seq{x}\) are integer arguments and \(\seq{l}\) are list arguments);
  this follows from the soundness of the catamorphism-based abstraction
  of CHCs~\cite{KatsuraSAS25}.
\iffull
  
  In contrast, there is no such guarantee for the goal clauses: because of
  the replacement of \(=_{\mathit{list}}\) and \(\neq_{\mathit{list}}\) with
  \(=_{\mathit{int}}\) and \(\neq_{\mathit{int}}\), the abstract version of goal
  clauses is neither stronger nor weaker than the original goal clauses.
  (Notice that \(\len{l_1}=\len{l_2}\) is an over-approximation of \(l_1=l_2\),
  while \(\len{l_1}\ne\len{l_2}\) is an under-approximation of \(l_1\ne l_2\).)
  The abstract goal clauses, however, tend to help a CHC solver to find an appropriate
  invariant. In the case above, if \(l\ne l'\) were replaced with \(\true\),
  we would obtain 
  \(\Pred{Take'}(n,l,r_1)\land \Pred{Drop'}(n,l,r_2)\land \Pred{Append'}(r_1,r_2,l')
  \imp\false\), which is so strong that the abstract CHCs would become unsatisfiable.
  \else
  (In contrast, there is no such guarantee for the goal clauses, because of
  the replacement of \(=_{\mathit{list}}\) and \(\neq_{\mathit{list}}\) with
  \(=_{\mathit{int}}\) and \(\neq_{\mathit{int}}\).)
  \fi

  Suppose a CHC solver yields the following model for
   the abstract CHCs above:
  \begin{align*}
&  \mathit{Take'}(n, l, r) \equiv r=l<n\lor l\ge n=r \qquad \mathit{Drop'}(n, l, r) \equiv (l<n\land r=0)\lor l\ge n=l-r.
\end{align*}
Then we have the following model for the definite clauses of the original CHCs.
  \begin{align*}
    &  \mathit{Take}(n, l, r) \equiv \len{r}=\len{l}<n\lor \len{l}\ge n=\len{r} \qquad \mathit{Drop}(n, l, r) \equiv (\len{l}<n\land \len{r}=0)\lor \len{l}\ge n=\len{l}-\len{r}.
  \end{align*}
  By conjoining them with the candidate model obtained by using STP inference,
  we obtain:
  \begin{align*}
    &  \mathit{Take}(n, l, r) \equiv \exists s.l=rs\land (\len{r}=\len{l}<n\lor \len{l}\ge n=\len{r})\\
    &  \mathit{Drop}(n, l, r) \equiv \exists s.l=sr\land ((\len{l}<n\land \len{r}=0)\lor \len{l}\ge n=\len{l}-\len{r}).
  \end{align*}
  This satisfies the goal clause, and thus we have proved that the original system
  of CHCs is satisfiable.

  \iffull
  In general, in parallel to running the least CSTP model, we apply the length
  abstraction and use a CHC solver over integer arithmetic to solve the abstract
  CHCs and obtain a model \(P'(\seq{x},\seq{l}')\equiv \varphi\). (Here, we assume
  that \(P'(\seq{x},\seq{l'})\) is the length abstraction of the original predicate
  \(P(\seq{x},\seq{l})\), where \(\seq{x}\) and \(\seq{l}\) are integer and list
  arguments respectively.)
  Then when checking whether the least CSTP model satisfies the goal clause,
  we strengthen the candidate model \(P\) by
  conjoining \([\len{\seq{l}}/\seq{l}']\varphi\).
\fi
  \iffull
    \begin{remark}
      Instead of just using the information obtained from the length abstraction,
      it is possible to propagate information gathered from
      STP inference and length abstraction to each other, as in
      Nelson-Oppen procedure~\cite{10.1145/322186.322198} for theory combinations.
      For example, suppose that \(P(l_1,l_2,l_3)\) is found to imply
      \(l_3=l_1l_2\) by STP inference. Then, each occurrence of \(P(l_1,l_2,l_3)\)
      in the body of a clause can be replaced by \(P(l_1,l_2,l_3)\land l_3=l_1l_2\),
      before applying the length abstraction. Then, a clause of the form
      \(P(l_1,l_2,l_3)\land \cdots \imp P(l_4,l_5,l_6)\) is abstracted to
      \(P'(l_1,l_2,l_3)\land l_3=l_1+l_2\land \cdots \imp P'(l_4,l_5,l_6)\). \qed
    \end{remark}
    \begin{remark}
      In theory, the
      above method for utilizing CHCs over integer arithmetic is not restricted
      to the length abstraction: any sound abstraction of lists to integers~\cite{KatsuraSAS25}
      would
      work. Our restriction to the length abstraction comes from the following
      practical motivation. As in the example above, the final candidate model
      involves both sequence constraints and constraints on integers obtained by
      abstracting lists. The state-of-the-art SMT solvers like CVC5 and Z3 work
      fairly well for a combination of sequences
      and length functions, but do not
      work well for a combination of sequences and arbitrary recursive functions
      on sequences, to our knowledge. \qed
    \end{remark}
    \fi
    \subsection{Beyond Sequence and Integer Constraints}
    \label{sec:app-multiset}
    \label{sec:chc-multisets}
    The (conjunctive) solvable set/multiset tuple patterns discussed in Section~\ref{sec:set} are also applicable to CHC solving,
    as stated in the theorem below. Here, CHCs over multisets refer to those whose constraint
    formulas are drawn from those on multisets~\cite{multiset}.
    The theorem follows from the set/multiset versions of Theorems~\ref{th:cstp-lfp}, \ref{th:stp-lfp}, and \ref{th:satisfiability}.
    \iffull\else See the longer version~\cite{STPlong} for details. \fi
\begin{theorem}
  \label{th:decidability-chc-set}
  Given a system \(\CHC\) of CHCs over multisets, it is decidable whether \(\CHC\) has a CSMTP-model.
  Similarly, given a system \(\CHC\) of CHCs over sets, it is decidable whether \(\CHC\) has a CSSTP-model.
\end{theorem}

\begin{theorem}
  \label{th:decidability-chc-set-stp}
  Given a system \(\CHC\) of CHCs over multisets, it is decidable whether \(\CHC\) has an SMTP-model.
  Similarly, given a system \(\CHC\) of CHCs over sets, it is decidable whether \(\CHC\) has an SSTP-model.
\end{theorem}

\begin{example}
  \label{ex:isort}
  Let us consider the following insertion sort program:
\begin{verbatim}
  let rec insert x l =
    match l with [] -> [x] | y::l' -> if x<y then x::l else y::(insert x l')
  let rec isort l = match l with [] -> [] | x::l' -> insert x (isort l')
  let main l = assert(eq_as_multiset (isort l) l).
\end{verbatim}
The function \texttt{main} takes a list \(l\) as input, sorts it, and asserts that the result is equivalent to \(l\) as multisets;
here, for the sake of simplicity, we assume that \verb|eq_as_multiset| is a primitive.
By abstracting lists to multisets and overapproximating the formula \(x<y\) by \(\true\),
we obtain the following CHCs over multisets, whose satisfiability
implies the lack of assertion failures.
\begin{align*}
  & |x|=1\imp \Pred{Insert}(x, \epsilon, x).     \qquad l=yl' \land |x|=|y|=1\imp \Pred{Insert}(x,l,xl).  \\  &
  l=yl' \land |x|=|y|=1\land \Pred{Insert}(x,l',r)\imp \Pred{Insert}(x,l,yr).\\
  & \Pred{Sort}(\epsilon, \epsilon). \qquad |x|=1\land l=xl'\land \Pred{Sort}(l', r')\land \Pred{Insert}(x,r',r)\imp \Pred{Sort}(l,r). \qquad \Pred{Sort}(l,r)
    \imp l=r. \end{align*}
By running the multiset version of the procedure of Theorem~\ref{th:cstp-lfp},
we obtain the following least CSMTP-model for definite clauses:
\[ \Pred{Insert}(x,l,r) \equiv (x,l,r)\in \lang(x,l,xl)\qquad
\Pred{Sort}(l,r) \equiv (l,r)\in \lang(l,l).\]
Since it satisfies the goal clause \(  \Pred{Sort}(l,r)    \imp l=r\), we can conclude that the system of CHCs above is
satisfiable, and hence also that the program does not suffer from assertion failures.
\iffull
In the same manner, we can also verify that the merge and quick sort programs preserve the multiset of elements,
 fully automatically (without any invariant annotations for auxiliary functions such as merge and partition). \fi
\qed
\end{example}

 \section{Implementation and Experiments}
\label{sec:exp}

We have implemented a tool called \tupinf{} for (C)STP  inference
based on the algorithm in Section~\ref{sec:stp},
and a new CHC solver called \chocolat{} based on the method
in Section~\ref{sec:chc}, using \tupinf{} as a backend.
Section~\ref{sec:imp} gives an overview of those tools,
and Section~\ref{sec:expresult} reports experimental results.
\subsection{Implementation}
\label{sec:imp}
\subsubsection{\tupinf{}: a Solvable Tuple Pattern Inference Tool}

\tupinf{} has been implemented in OCaml. \tupinf{} takes learning data in the csv format, and
outputs STPs or CSTPs with the extensions described in Sections~\ref{sec:ext-rev} (reverse),
and \ref{sec:set} (set and multiset patterns).
The alphabet \(\Alpha\) used in the learning data consists of digits and uppercase and lowercase English letters.
\tupinf{} provides options to (i) enable or disable constant patterns (i.e., \rn{R-CPrefix} and
\rn{R-CSuffix}), (ii) choose between STPs and CSTPs, and (iii) select between sequence, multiset, and set patterns.
The option for constant patterns  was disabled in the experiments reported below. 
\iffull
\else \fi

\subsubsection{\chocolat{}: a CHC Solver for List-like Data Structures}
\chocolat{} (\textbf{C}onstrained \textbf{Ho}rn \textbf{C}lauses
\textbf{o}ver \textbf{L}ists)
is also implemented in OCaml and uses \tupinf{} for CSTP inference, Z3~\cite{DBLP:conf/tacas/MouraB08} and CVC5~\cite{DBLP:conf/tacas/BarbosaBBKLMMMN22} for SMT solving,
and \hoice{}~\cite{DBLP:journals/jar/ChampionCKS20} for CHC solving over integer arithmetic, as backend solvers.
Currently, \chocolat{} supports only list-like data structures as ADTs, i.e., ADTs that have
two constructors: one with no arguments and the other with two arguments whose types are some element type and the ADT itself.
\chocolat{} also supports natural numbers represented as lists with unit elements.
\chocolat{} has five internal modes: list-stp-mode, set-mode, multiset-mode, list-len-mode, and list-cstp-mode.
Given a system of CHCs as input, \chocolat{} tries five modes (in the order listed above)
 one by one until it finds that the CHCs are satisfiable.
 In the list-stp-mode and list-cstp-mode, the tool
 infers tuple patterns for list-like data structures as described in Section~\ref{sec:chc-words},
 but the inferred patterns are restricted to STPs in the former for the sake of efficiency. (Experimental results suggest that
 CSTP inference itself can be performed efficiently in practice, but it sometimes outputs a large CSTP,
 which can impose significant overhead on the underlying SMT solver when checking its validity.)
In the list-len-mode, it additionally infers invariants on the lengths of lists as described in Section~\ref{sec:combination} by using the backend CHC solver for integer arithmetic.
\begin{comment}
Besides the support of multiple predicate variables,
the procedure in the list-mode 
deviates from Procedure~\ref{fig:algo} in Section~\ref{sec:chc} in the following points:
(1) We skip line 5 to deduce unsatisfiability; the unsatisfiability is checked by
a separate procedure as described below.
(2) We do not accumulate the tuple patterns found so far (cf. lines 8--9);
so \(T\) is a singleton set \(\set{\tpnew}\) on line 10, and 
if \(T\models G\) does not hold on line 10, 
we just proceed to the next step for inferring invariants on list lengths, to strengthen
the invariants represented by \(\tpnew\).
(3) \(\STPinf{}\) on line 7 is actually deterministic.
(4) To collect more samples on line 13, we use only the counterexample-guided method, which does not satisfy the second assumption in Theorem~\ref{th:solve-relcomp}
(recall Remark~\ref{rem:counterexample}).
Due to these differences, the implementation does not
satisfy the relative completeness as stated in
Theorem~\ref{th:solve-relcomp}.
That does not, however, seem to be a significant issue in practice, according
to the experimental results in Section~\ref{sec:expresult}.

In the set-mode (resp. multiset-mode), \chocolat{} infers set (resp. multiset) tuple patterns as described in Section~\ref{sec:app-multiset}
and checks whether inferred invariants imply the goal clause.
\end{comment}

The procedure explained so far can only prove the satisfiability of the CHCs.
\iffull
Thus, \chocolat{} also runs in parallel
a procedure based on symbolic execution to prove the unsatisfiability.
Given a CHC whose goal clause is \(\bigwedge_i P_i(\seq{y}_i) \land C \imp \false\),
\chocolat{} searches for an assignment to variables \(\seq{y}_i\) that makes the body (\(\bigwedge_i P_i(\seq{y}_i) \land C\)) of the goal clause valid
by recursively unfolding the predicates in the goal in a breadth-first manner.
The details of this procedure are omitted,
since it is rather straightforward and contains few novel ideas.
It would be possible to utilize the result of tuple pattern inference also to
speed-up the refutation process, but that is left for future work.
\else
Thus, \chocolat{} also runs in parallel
a resolution-based procedure to prove the unsatisfiability.
The details are omitted,
since it contains few novel ideas.
It would be possible to utilize the result of tuple pattern inference also to
speed-up the refutation process, which is left for future work.
\fi
\subsection{Experiments}
\label{sec:expresult}
We have conducted experiments on a benchmark set selected from CHC-COMP 2025 ADT-LIN category. The benchmark set consists of 445 instances of the CHC satisfiability problem,
which are all the instances of the ADT-LIN category where the used ADTs are
only list-like data structures.
The experiments were conducted on a machine with AMD Ryzen 9 5900X CPU and 32GB RAM.
The cpu and wallclock time limits were 3 minutes, and the memory limit was 8GB. The number of CPU cores was limited to 4.

We compared \chocolat{} with the state-of-the-art CHC solvers that support ADTs:
\ringen{}~\cite{DBLP:conf/pldi/KostyukovMF21},
\spacer{}/\racer{}~\cite{DBLP:journals/fmsd/KomuravelliGC16,DBLP:journals/pacmpl/KSG22},
\eld{}~\cite{Eldarica},
\hoice{}~\cite{DBLP:journals/jar/ChampionCKS20,DBLP:conf/aplas/Champion0S18}, and
\catalia~\cite{KatsuraSAS25}.
The results are shown in Table~\ref{tab:exp} and Figure~\ref{fig:exp}.
To confirm that our method serves as a complement to previous methods,
we have also prepared a combination of \chocolat{} and \catalia{}, shown as
\chocolat+\catalia{} in the table and figure; it runs \chocolat{} and \catalia{}
in parallel.\footnote{A combination of \chocolat{} (without seq-cstp-mode) and \catalia{}, called ``ChocoCatalia''
  won the ADT-LIN category of CHC-COMP 2025;
the configuration of \chocolat{}+\catalia{} reported in this submission slightly differs from
 ``ChocoCatalia''.}

\begin{table}[t]
  \centering
  \caption{Numbers of solved instances by \chocolat{} and the existing CHC solvers. 
  The numbers in parentheses are the numbers of uniquely solved instances.}
  \label{tab:exp}
  \begin{tabular}{lccc}
    \toprule
    \textbf{Solver} & \textbf{Solved (SAT)} & \textbf{Solved (UNSAT)} & \textbf{Solved (all)} \\
    \chocolat{}          & 180 ( 97) &  80 (  0) & 260 ( 97) \\
    \ringen{}            &  64 ( 23) &  37 (  4) & 101 ( 27) \\
    \spacer{}            &  20 (  2) &  90 (  0) & 110 (  2) \\
    \eld{}               &  20 (  0) &  84 (  0) & 104 (  0) \\
    \hoice{}             &  25 (  6) &  45 (  0) &  70 (  6) \\
    \catalia{}           &  87 ( 11) &  87 (  0) & 174 ( 11) \\
    \midrule
    \chocolat+\catalia{} & 194       &  84       & 278       \\
    \bottomrule
  \end{tabular}
\end{table}

\begin{figure}[t]
  \centering
  \begin{subfigure}[b]{0.48\textwidth}
    \centering
    \includegraphics[width=\textwidth]{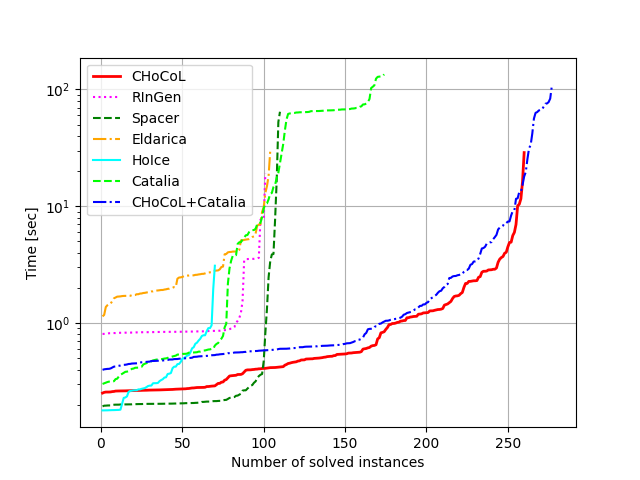}
    \caption{All instances}
  \end{subfigure}
  \hfill
  \begin{subfigure}[b]{0.48\textwidth}
    \centering
    \includegraphics[width=\textwidth]{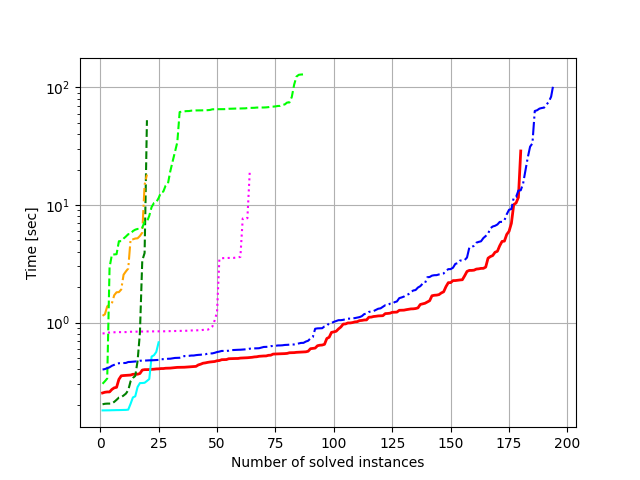}
    \caption{SAT instances}
  \end{subfigure}
  \caption{Numbers of solved instances within the time.}
  \label{fig:exp}
\end{figure}

Table~\ref{tab:exp} shows the numbers of solved instances by \chocolat{} and the other CHC solvers.
For each solver, the column \textbf{Solved (SAT)} (resp. \textbf{Solved (UNSAT)}) shows the number of satisfiable (resp. unsatisfiable) instances solved,
and the column \textbf{Solved (all)} shows the total number of instances solved.
The numbers in parentheses indicate instances uniquely solved by the solver,
i.e., those that no other solver (except \chocolat{}+\catalia{}) could solve within the time limit.

The experimental results demonstrate the effectiveness of \chocolat{}.
\chocolat{} achieved the highest total number of
solved instances (260) among all compared solvers, showing its overall superiority
in handling CHCs with list-like data structures.
A key strength of \chocolat{} is its performance on satisfiable instances,
where it solved 180 instances compared to 87 by \catalia{} and 64 by \ringen{}.
More importantly, \chocolat{} uniquely solved 97 satisfiable instances,\footnote{
A increase of the time limit would not significantly change the situation,
because the benchmark set contains a number of instances that require
reasoning about list equalities, as in the examples in Section~\ref{sec:intro},
which cannot be solved by \catalia{} and \ringen{}, which abstract
lists to integers or elements of a finite domain.}
demonstrating the effectiveness of the
STP inference approach for CHC solving.

The result of \chocolat{}+\catalia{}, though, shows that \chocolat{} also
has a limitation, and that it is better to use it as a complement to other solvers.
Also, in contrast with the strong performance of \chocolat{} for SAT instances,
\chocolat{} showed slightly lower performance than other solvers
(except \ringen{}) for UNSAT instances.
This is unsurprising because our current refutation procedure is
quite naive.
We plan to apply STP inference
also to improve the refutation procedure.

Figure~\ref{fig:exp} gives cactus plots, showing how many instances were solved
within a given time.
The left subfigure (a) shows the numbers of solved instances for all instances,
and the right subfigure (b) shows the numbers of solved instances for SAT instances.
It is remarkable that, for the successful instances,  \chocolat{} solved many of them
almost instantly, within one second.

\begin{table}[t]
  \centering
  \caption{Numbers of solved instances by each mode of \chocolat{}.}
  \label{tab:mode-num}
  \begin{tabular}{lcc}
    \toprule
    \textbf{Mode} & \textbf{Solved SAT instances} & \textbf{Uniquely solved SAT instances} \\
    \midrule
    list-stp-mode & 89 & 50 \\
    set/multiset-mode & 9 & 5 \\
    list-len-mode & 78 & 42 \\
    refutation & 4 & 0 \\
    \midrule
    Total & 180 & 97 \\
    \bottomrule
  \end{tabular}
\end{table}

Table~\ref{tab:mode-num} shows the numbers of solved SAT instances by each mode of \chocolat{}.
The row ``refutation'' shows the number of instances found to be satisfiable only by finite unfolding.
Remarkably, the list-stp-mode alone solved more SAT instances than the other solvers.
As shown in the row ``list-len-mode'', the combination of STPs and CHC solving over integer arithmetic
was also effective. There were no uniquely solved instances for list-cstp-mode; this suggests that
although CSTPs are required in theory (cf. Theorem~\ref{th:decidability-chc-set}), STPs often suffice in practice.

We have also manually inspected some of the SAT instances that could not be solved by \chocolat{}.
They typically involve functions such as sort, filter, map, and fold, whose input/output relations cannot be
expressed by STPs. We plan to extend our framework to support those list functions, along the line suggested in
Remark~\ref{rem:sort}.
 \section{Related Work}
\label{sec:rel}

\iffull
\subsection{Learning Languages from Positive Samples}
\else
\subsubsection*{Learning Languages from Positive Samples}
\fi
The problem of learning a language from only positive samples~\cite{GOLD1967447} has
been actively studied, and various language classes have been found to be learnable.
As already mentioned,
most closely related to our notion of tuple patterns 
is Angluin's pattern languages~\cite{ANGLUIN1980117,ANGLUIN198046}
and its variations~\cite{10.1007/3-540-11980-9_19,SHINOHARA1994175}.
Indeed, ``erasing'' pattern languages~\cite{10.1007/3-540-11980-9_19,SHINOHARA1994175}
(in which empty strings may be substituted for variables, as opposed to Angluin's original pattern languages)
can be considered singleton tuple patterns, and conversely, if we introduce a special symbol \$ (which must not
occur in strings substituted for variables; such a restriction may be expressed by patterns with regular constraints~\cite{10.1007/3-540-11980-9_19}), then a tuple pattern \((\pat_1,\ldots,\pat_n)\) can be expressed as
the pattern \(\pat_1\$\cdots\$\pat_n\).
No polynomial time learning algorithm is known for the full class of pattern languages (and in fact,
the existence of such an algorithm is highly unlikely given the result of ~\cite{ANGLUIN198046}).
Various subclasses of pattern languages are known to be polynomial time learnable from positive data~\cite{ANGLUIN198046,SHINOHARA1982,REIDENBACH200691,10.1007/3-540-11980-9_19,Mitchell98}.
Among others,
Angluin~\cite{ANGLUIN198046} and Shinohara~\cite{SHINOHARA1994175} gave
polynomial time algorithms respectively
for patterns consisting of a single variable and for ``regular'' (erasing) patterns where each variable may occur only once. Mitchell~\cite{Mitchell98} extended the latter
result to quasi-regular patterns, where variables must occur the same number of times.
Those subclasses are not large enough to express tuple patterns like
\((xy, xyx)\) (which contains two variables, each of which occurs a different number of times).

Other classes of languages learnable only from positive samples (i.e., ``identifiable in the limit''~\cite{GOLD1967447})
include: reversible languages~\cite{10.1145/322326.322334},
strictly locally testable languages~\cite{LocallyTestableLanguage},
``very simple grammars''~\cite{YOKOMORI2003179,10.1007/11872436_5},
substitutable context-free languages and their variations~\cite{JMLR:v8:clark07a,10.1007/978-3-540-88009-7_21,YOSHINAKA20111821},
and length-bounded elementary formal systems~\cite{SHINOHARA1994175}.
Except the length-bounded elementary formal systems~\cite{SHINOHARA1994175}, those classes are not expressive enough
to subsume solvable tuple patterns. In fact, most of them are subclasses of regular or context-free languages.
For length-bounded elementary formal systems~\cite{SHINOHARA1994175}, there exists no efficient learning algorithm.

Gulwani et al.~\cite{DBLP:journals/cacm/GulwaniHS12,DBLP:conf/popl/Gulwani11} invented
a method for inferring string functions from translation examples
and applied it to spreadsheets. Although their method is
heuristic and does not satisfy good learning-theoretic properties like ours
to our knowledge,
it would be interesting to integrate their technique with ours to infer invariants
for functional relations.

\iffull
\subsection{Data-Driven Approaches to Program Verification}
\else
\subsubsection*{Data-Driven Approaches to Program Verification}
\fi
Various data-driven approaches to program verification have recently been proposed,
especially for automated discovery of program invariants.
Some of them require only positive samples~\cite{10.1007/978-3-642-37036-6_31,DBLP:conf/aplas/IkedaSK23,Zhu16PLDI}, some require negative samples as well~\cite{zhu_2015,DBLP:conf/pldi/ZhuMJ18,DBLP:conf/iclr/RyanWYGJ20},
\iffull
and some require
even implication examples of the form \(d_1\land \cdots \land d_k\imp d\),
which means ``if \(d_1,\ldots,d_k\) are positive samples, then \(d\) is also positive''~\cite{garg_2014,DBLP:journals/jar/ChampionCKS20,DBLP:journals/pacmpl/EzudheenND0M18,KobayashiNeuGuS}.
\else
even implication examples~\cite{garg_2014,DBLP:journals/jar/ChampionCKS20,DBLP:journals/pacmpl/EzudheenND0M18,KobayashiNeuGuS}.
\fi
The used techniques vary from algebraic ones~\cite{10.1007/978-3-642-37036-6_31,DBLP:conf/aplas/IkedaSK23}, SVM~\cite{zhu_2015,DBLP:conf/pldi/ZhuMJ18}, and neural networks~\cite{DBLP:conf/iclr/RyanWYGJ20,KobayashiNeuGuS}.
Most of those techniques have been used mainly for the discovery of invariants on integers, and it is unclear how they can
be adopted to discover invariants on list-like data structures.

Among them, the techniques of Sharma et al.~\cite{10.1007/978-3-642-37036-6_31}
and Ikeda et al.~\cite{DBLP:conf/aplas/IkedaSK23} are technically closest to
ours (although they are for discovering numerical equality constraints), in that
those methods can also quickly infer equality constraints only from a small number of examples.
In fact, our STP inference algorithm, which repeatedly simplifies learning data in matrix form,
has been inspired by those techniques. Ikeda et al.~\cite{DBLP:conf/aplas/IkedaSK23} applied
that technique to CHC solving. The application to CHC solving is less trivial in
our context, however. In the case of Ikeda et al.~\cite{DBLP:conf/aplas/IkedaSK23}, which finds
equality constraints on integers, it is sufficient to insert those invariants as a part of preprocessing,
and then invoke an off-the-shelf CHC solver. For example, if we can find an invariant \(P(x,y)\imp
x+2y=3\), then we can provide that information to a CHC solver, just by replacing every occurrence
of \(P(z,w)\) in CHC with \(P(z,w)\land z+2w=3\). In contrast, since there are no CHC solvers
that can directly reason about sequence equality constraints, we had to develop a more sophisticated
method for applying the result of CSTP inference to CHC solving, as discussed in Section~\ref{sec:chc}.
\iffull Actually, Theorem~\ref{th:decidability-chc} can also be adopted for the method of
Ikeda et al.~\cite{DBLP:conf/aplas/IkedaSK23} as follows: given a system \(\CHC\) of CHCs over linear integer arithmetic,
it is decidable whether \(\CHC\) has a model that can be described solely by conjunctions of linear integer equalities of the form
\(\bigwedge_i\vec{c}_i\vec{x}_i+d_i=0\).
\fi

Zhu et al.~\cite{Zhu16PLDI} proposed a data-driven approach to learning invariants of (immutable) data structures.
Their technique learns relations between two nodes of a data structure (e.g., ``for every cons-cell \(N_1\) in a list,
the value of every cell reachable from \(N_1\) is no less than that of \(N_1\)''), and  can therefore reason about
properties such as the sortedness of a list.
\iffull
To our knowledge, however, their technique cannot be applied to
reasoning about the \emph{logical equivalence} of lists;
in fact, their abstraction for lists cannot distinguish [1;1;1;2;2] from [1;1;2;2;2].
Therefore, their technique cannot handle either of our motivating examples in Section~\ref{sec:intro}.
\else
To our knowledge, however, their technique cannot be applied to
reasoning about the \emph{logical equivalence} of lists, and hence cannot handle either of our motivating examples in Section~\ref{sec:intro}.  
\fi
Related techniques have been developed for verifying programs that manipulate
\emph{mutable} data structures~\cite{10.1145/3022187,10.1007/978-3-319-08867-9_3}, but these techniques also
cannot reason about the logical equality of lists.

\iffull
\subsection{CHC Solving for Algebraic Data Types}
\else
\subsubsection*{CHC Solving for Algebraic Data Types.}
\fi
Solving CHCs over algebraic data types (or, automated verification of
programs manipulating ADTs) often require finding and reasoning about inductive properties.
To this end, Unno et al.~\cite{DBLP:conf/cav/UnnoTS17} introduced an automated method
for applying induction in CHC solving.
De Angelis et al.~\cite{DBLP:journals/tplp/AngelisFPP18a} later proposed a method to achieve similar effects without explicit induction, by using unfold/fold transformations.
Mordvinov and Fedyukovich~\cite{DBLP:conf/lpar/MordvinovF17} also proposed a variation of
unfold/fold transformations.
Losekoot et al.~\cite{DBLP:conf/fscd/LosekootGJ23,DBLP:conf/sas/LosekootGJ24}
applied automata-based techniques to reason about relational properties of ADTs, which also have
a similar effect to unfold/fold transformation-based techniques.
These methods tend to be fragile,  requiring a lot of heuristics to automatically decide
how and when to apply induction, unfold/fold transformations, etc.

Several methods for CHC solving have also been proposed, which abstract ADTs
to elements of a finite domain or integers.
Kostyukov~\cite{DBLP:conf/pldi/KostyukovMF21} proposed a method for abstracting
ADTs using automata, while the Eldarica CHC solver~\cite{Eldarica} abstracts ADTs by their size.
Govind V. K. et al.~\cite{DBLP:journals/pacmpl/KSG22} proposed a method that uses
recursively defined functions (RDTs) over ADTs for abstraction, but those RDTs must be explicitly specified.
Katsura et al.~\cite{KatsuraSAS25} proposed a method for
automatically finding a catamorphism-based abstraction for ADTs.
While these methods are effective for certain CHCs, they are not effective
for proving the equality of ADTs; although the abstraction by using an abstraction map \(\alpha\)
may be used for proving \(\alpha(l_1)=\alpha(l_2)\), it does not imply \(l_1=l_2\) in general.
Thus, those methods cannot be used to prove neither of the examples in Section~\ref{sec:intro}.
As confirmed in our experiments, our method using CSTPs is complementary 
to those methods; in fact, the combination of our method with \catalia{} achieved the best
performance in the experiments in Section~\ref{sec:exp}, and also in the competition CHC-COMP 2025.

\iffull
\subsection{Combination of Static Analyses}
\else
\subsubsection*{Combination of Static Analyses.}
\fi
In Section~\ref{sec:combination}, we have shown how to strengthen
our STP inference-based CHC solving procedure by combining it with
an off-the-shelf CHC solver for integer arithmetic.
There have been studies on how to combine static analyses~\cite{10.1145/200994.200998,DBLP:conf/pldi/GulwaniT06,10.1007/978-3-642-19805-2_31}
in the context of abstract interpretation.
The particular way of combining CHC solving methods, however,
appears to be novel. As discussed in Section~\ref{sec:combination},
the CHCs obtained by the length abstraction is neither sound nor complete
with respect to the original CHCs (the satisfiability of the abstract CHCs
neither implies nor is implied by the satisfiability of the original ones),
but  we can still obtain useful information from a solution of the abstract CHCs,
for solving the original CHCs. Our combination method 
treats a CHC solver for integer arithmetic as a black box, allowing any CHC
solvers to be used, including those based on PDR~\cite{DBLP:journals/fmsd/KomuravelliGC16}, predicate abstraction~\cite{Eldarica}, and data-driven methods~\cite{DBLP:conf/aplas/Champion0S18,DBLP:conf/pldi/ZhuMJ18}.
We have discussed only a specific combination in Section~\ref{sec:combination};
it would be an interesting
direction for future work
to further generalize and formalize the idea of combining CHC solvers over different domains.
 \section{Conclusions}
\label{sec:conc}

We have introduced the new notions of solvable tuple patterns (STPs) and conjunctive STPs (CSTPs),
and developed their theories. \iffull In particular we have shown that both STPs and CSTPs are learnable from only positive samples,
and that the satisfiability of quantifier-free formulas is decidable. \fi
We have applied those theories to
prove that it is decidable whether a system of CHCs over words has a CSTP-describable model.
 We have implemented a new CHC solver called \chocolat{} based on 
 the proposed approach, and confirmed the effectiveness of the approach
 through experiments.

 Future work includes: (i) generalizing STPs to trees and other
 algebraic structures (cf. Remark~\ref{rem:treepat}) to handle
 tree-like data structures and support reasoning about other list
 operations (cf. Remark~\ref{rem:sort}); (ii) extending STPs to infer
 Boolean combinations of STPs, such as \((x, y, xz)\lor (x,y,yz)\),
 meaning that either the first or second element is a prefix of
 the third; (iii) applying other classes of languages learnable from
 positive samples to CHC solving; and (iv) extending and generalizing
 the method for combining CHC solvers over different domains.

\subsection*{Acknowledgment}
  We thank the anonymous reviewers for their useful feedback.
This work was supported by
JSPS KAKENHI Grant Number JP20H05703 and JP26H02486.

\subsection*{Data Availability Statement}
The artifact of this paper is available on Zenodo~\cite{Artifact}. It includes the instructions, software, benchmark set, and scripts to reproduce the experiments in the paper.

%\bibliographystyle{ACM-Reference-Format}
%\bibliography{abbrv,koba,ry}

\begin{thebibliography}{76}

%%% the Web address.

\ifx \showCODEN    \undefined \def \showCODEN     #1{\unskip}     \fi
\ifx \showDOI      \undefined \def \showDOI       #1{#1}\fi
\ifx \showISBNx    \undefined \def \showISBNx     #1{\unskip}     \fi
\ifx \showISBNxiii \undefined \def \showISBNxiii  #1{\unskip}     \fi
\ifx \showISSN     \undefined \def \showISSN      #1{\unskip}     \fi
\ifx \showLCCN     \undefined \def \showLCCN      #1{\unskip}     \fi
\ifx \shownote     \undefined \def \shownote      #1{#1}          \fi
\ifx \showarticletitle \undefined \def \showarticletitle #1{#1}   \fi
\ifx \showURL      \undefined \def \showURL       {\relax}        \fi
% invisible to TeX
\providecommand\bibfield[2]{#2}
\providecommand\bibinfo[2]{#2}
\providecommand\natexlab[1]{#1}
\providecommand\showeprint[2][]{arXiv:#2}

\bibitem[Aho(1968)]%
        {10.1145/321479.321488}
\bibfield{author}{\bibinfo{person}{Alfred~V. Aho}.}
  \bibinfo{year}{1968}\natexlab{}.
\newblock \showarticletitle{Indexed Grammars—An Extension of Context-Free
  Grammars}.
\newblock \bibinfo{journal}{\emph{J. ACM}} \bibinfo{volume}{15},
  \bibinfo{number}{4} (\bibinfo{date}{Oct.} \bibinfo{year}{1968}),
  \bibinfo{pages}{647–671}.
\newblock
\showISSN{0004-5411}
\urldef\tempurl%
\url{https://doi.org/10.1145/321479.321488}
\showDOI{\tempurl}


\bibitem[Angluin(1980a)]%
        {ANGLUIN198046}
\bibfield{author}{\bibinfo{person}{Dana Angluin}.}
  \bibinfo{year}{1980}\natexlab{a}.
\newblock \showarticletitle{Finding patterns common to a set of strings}.
\newblock \bibinfo{journal}{\emph{J. Comput. System Sci.}}
  \bibinfo{volume}{21}, \bibinfo{number}{1} (\bibinfo{year}{1980}),
  \bibinfo{pages}{46--62}.
\newblock
\showISSN{0022-0000}
\urldef\tempurl%
\url{https://doi.org/10.1016/0022-0000(80)90041-0}
\showDOI{\tempurl}


\bibitem[Angluin(1980b)]%
        {ANGLUIN1980117}
\bibfield{author}{\bibinfo{person}{Dana Angluin}.}
  \bibinfo{year}{1980}\natexlab{b}.
\newblock \showarticletitle{Inductive inference of formal languages from
  positive data}.
\newblock \bibinfo{journal}{\emph{Information and Control}}
  \bibinfo{volume}{45}, \bibinfo{number}{2} (\bibinfo{year}{1980}),
  \bibinfo{pages}{117--135}.
\newblock
\showISSN{0019-9958}
\urldef\tempurl%
\url{https://doi.org/10.1016/S0019-9958(80)90285-5}
\showDOI{\tempurl}


\bibitem[Angluin(1982)]%
        {10.1145/322326.322334}
\bibfield{author}{\bibinfo{person}{Dana Angluin}.}
  \bibinfo{year}{1982}\natexlab{}.
\newblock \showarticletitle{Inference of Reversible Languages}.
\newblock \bibinfo{journal}{\emph{J. ACM}} \bibinfo{volume}{29},
  \bibinfo{number}{3} (\bibinfo{date}{July} \bibinfo{year}{1982}),
  \bibinfo{pages}{741–765}.
\newblock
\showISSN{0004-5411}
\urldef\tempurl%
\url{https://doi.org/10.1145/322326.322334}
\showDOI{\tempurl}


\bibitem[Asada et~al\mbox{.}(2017)]%
        {DBLP:journals/scp/AsadaSK17}
\bibfield{author}{\bibinfo{person}{Kazuyuki Asada}, \bibinfo{person}{Ryosuke
  Sato}, {and} \bibinfo{person}{Naoki Kobayashi}.}
  \bibinfo{year}{2017}\natexlab{}.
\newblock \showarticletitle{Verifying relational properties of functional
  programs by first-order refinement}.
\newblock \bibinfo{journal}{\emph{Sci. Comput. Program.}}
  \bibinfo{volume}{137} (\bibinfo{year}{2017}), \bibinfo{pages}{2--62}.
\newblock
\urldef\tempurl%
\url{https://doi.org/10.1016/j.scico.2016.02.007}
\showDOI{\tempurl}


\bibitem[Barbosa et~al\mbox{.}(2022)]%
        {DBLP:conf/tacas/BarbosaBBKLMMMN22}
\bibfield{author}{\bibinfo{person}{Haniel Barbosa}, \bibinfo{person}{Clark~W.
  Barrett}, \bibinfo{person}{Martin Brain}, \bibinfo{person}{Gereon Kremer},
  \bibinfo{person}{Hanna Lachnitt}, \bibinfo{person}{Makai Mann},
  \bibinfo{person}{Abdalrhman Mohamed}, \bibinfo{person}{Mudathir Mohamed},
  \bibinfo{person}{Aina Niemetz}, \bibinfo{person}{Andres N{\"{o}}tzli},
  \bibinfo{person}{Alex Ozdemir}, \bibinfo{person}{Mathias Preiner},
  \bibinfo{person}{Andrew Reynolds}, \bibinfo{person}{Ying Sheng},
  \bibinfo{person}{Cesare Tinelli}, {and} \bibinfo{person}{Yoni Zohar}.}
  \bibinfo{year}{2022}\natexlab{}.
\newblock \showarticletitle{cvc5: {A} Versatile and Industrial-Strength {SMT}
  Solver}. In \bibinfo{booktitle}{\emph{Tools and Algorithms for the
  Construction and Analysis of Systems - 28th International Conference, {TACAS}
  2022, Held as Part of the European Joint Conferences on Theory and Practice
  of Software, {ETAPS} 2022, Munich, Germany, April 2-7, 2022, Proceedings,
  Part {I}}} \emph{(\bibinfo{series}{Lecture Notes in Computer Science},
  Vol.~\bibinfo{volume}{13243})}, \bibfield{editor}{\bibinfo{person}{Dana
  Fisman} {and} \bibinfo{person}{Grigore Rosu}} (Eds.).
  \bibinfo{publisher}{Springer}, \bibinfo{pages}{415--442}.
\newblock
\urldef\tempurl%
\url{https://doi.org/10.1007/978-3-030-99524-9\_24}
\showDOI{\tempurl}


\bibitem[Bj{\o}rner et~al\mbox{.}(2015)]%
        {Bjorner15}
\bibfield{author}{\bibinfo{person}{Nikolaj Bj{\o}rner}, \bibinfo{person}{Arie
  Gurfinkel}, \bibinfo{person}{Kenneth~L. McMillan}, {and}
  \bibinfo{person}{Andrey Rybalchenko}.} \bibinfo{year}{2015}\natexlab{}.
\newblock \showarticletitle{Horn Clause Solvers for Program Verification}. In
  \bibinfo{booktitle}{\emph{Fields of Logic and Computation {II} - Essays
  Dedicated to Yuri Gurevich on the Occasion of His 75th Birthday}}
  \emph{(\bibinfo{series}{LNCS}, Vol.~\bibinfo{volume}{9300})}.
  \bibinfo{publisher}{Springer}, \bibinfo{pages}{24--51}.
\newblock
\urldef\tempurl%
\url{https://doi.org/10.1007/978-3-319-23534-9_2}
\showDOI{\tempurl}


\bibitem[Bj{\o}rner et~al\mbox{.}(2013)]%
        {DBLP:conf/sas/BjornerMR13}
\bibfield{author}{\bibinfo{person}{Nikolaj~S. Bj{\o}rner},
  \bibinfo{person}{Kenneth~L. McMillan}, {and} \bibinfo{person}{Andrey
  Rybalchenko}.} \bibinfo{year}{2013}\natexlab{}.
\newblock \showarticletitle{On Solving Universally Quantified {Horn} Clauses}.
  In \bibinfo{booktitle}{\emph{Static Analysis - 20th International Symposium,
  {SAS} 2013, Seattle, WA, USA, June 20-22, 2013. Proceedings}}
  \emph{(\bibinfo{series}{Lecture Notes in Computer Science},
  Vol.~\bibinfo{volume}{7935})}, \bibfield{editor}{\bibinfo{person}{Francesco
  Logozzo} {and} \bibinfo{person}{Manuel F{\"{a}}hndrich}} (Eds.).
  \bibinfo{publisher}{Springer}, \bibinfo{pages}{105--125}.
\newblock
\urldef\tempurl%
\url{https://doi.org/10.1007/978-3-642-38856-9\_8}
\showDOI{\tempurl}


\bibitem[Champion et~al\mbox{.}(2020)]%
        {DBLP:journals/jar/ChampionCKS20}
\bibfield{author}{\bibinfo{person}{Adrien Champion}, \bibinfo{person}{Tomoya
  Chiba}, \bibinfo{person}{Naoki Kobayashi}, {and} \bibinfo{person}{Ryosuke
  Sato}.} \bibinfo{year}{2020}\natexlab{}.
\newblock \showarticletitle{{ICE}-Based Refinement Type Discovery for
  Higher-Order Functional Programs}.
\newblock \bibinfo{journal}{\emph{J. Autom. Reason.}} \bibinfo{volume}{64},
  \bibinfo{number}{7} (\bibinfo{year}{2020}), \bibinfo{pages}{1393--1418}.
\newblock
\urldef\tempurl%
\url{https://doi.org/10.1007/s10817-020-09571-y}
\showURL{%
\tempurl}


\bibitem[Champion et~al\mbox{.}(2018)]%
        {DBLP:conf/aplas/Champion0S18}
\bibfield{author}{\bibinfo{person}{Adrien Champion}, \bibinfo{person}{Naoki
  Kobayashi}, {and} \bibinfo{person}{Ryosuke Sato}.}
  \bibinfo{year}{2018}\natexlab{}.
\newblock \showarticletitle{{HoIce}: An {ICE}-Based Non-linear {Horn} Clause
  Solver}. In \bibinfo{booktitle}{\emph{Programming Languages and Systems -
  16th Asian Symposium, {APLAS} 2018, Wellington, New Zealand, December 2-6,
  2018, Proceedings}} \emph{(\bibinfo{series}{Lecture Notes in Computer
  Science}, Vol.~\bibinfo{volume}{11275})},
  \bibfield{editor}{\bibinfo{person}{Sukyoung Ryu}} (Ed.).
  \bibinfo{publisher}{Springer}, \bibinfo{pages}{146--156}.
\newblock
\urldef\tempurl%
\url{https://doi.org/10.1007/978-3-030-02768-1\_8}
\showDOI{\tempurl}


\bibitem[Clark and Eyraud(2007)]%
        {JMLR:v8:clark07a}
\bibfield{author}{\bibinfo{person}{Alexander Clark} {and}
  \bibinfo{person}{R{{\'e}}mi Eyraud}.} \bibinfo{year}{2007}\natexlab{}.
\newblock \showarticletitle{Polynomial Identification in the Limit of
  Substitutable Context-free Languages}.
\newblock \bibinfo{journal}{\emph{Journal of Machine Learning Research}}
  \bibinfo{volume}{8}, \bibinfo{number}{60} (\bibinfo{year}{2007}),
  \bibinfo{pages}{1725--1745}.
\newblock
\urldef\tempurl%
\url{http://jmlr.org/papers/v8/clark07a.html}
\showURL{%
\tempurl}


\bibitem[Clarke et~al\mbox{.}(2003)]%
        {Clarke2003a}
\bibfield{author}{\bibinfo{person}{Edmund Clarke}, \bibinfo{person}{Orna
  Grumberg}, \bibinfo{person}{Somesh Jha}, \bibinfo{person}{Yuan Lu}, {and}
  \bibinfo{person}{Helmut Veith}.} \bibinfo{year}{2003}\natexlab{}.
\newblock \showarticletitle{Counterexample-guided abstraction refinement for
  symbolic model checking}.
\newblock \bibinfo{journal}{\emph{JACM}} \bibinfo{volume}{50},
  \bibinfo{number}{5} (\bibinfo{year}{2003}), \bibinfo{pages}{752--794}.
\newblock


\bibitem[Codish et~al\mbox{.}(1995)]%
        {10.1145/200994.200998}
\bibfield{author}{\bibinfo{person}{Michael Codish}, \bibinfo{person}{Anne
  Mulkers}, \bibinfo{person}{Maurice Bruynooghe},
  \bibinfo{person}{Maria~Garc\'{\i}a de~la Banda}, {and}
  \bibinfo{person}{Manuel Hermenegildo}.} \bibinfo{year}{1995}\natexlab{}.
\newblock \showarticletitle{Improving abstract interpretations by combining
  domains}.
\newblock \bibinfo{journal}{\emph{ACM Trans. Program. Lang. Syst.}}
  \bibinfo{volume}{17}, \bibinfo{number}{1} (\bibinfo{date}{Jan.}
  \bibinfo{year}{1995}), \bibinfo{pages}{28–44}.
\newblock
\showISSN{0164-0925}
\urldef\tempurl%
\url{https://doi.org/10.1145/200994.200998}
\showDOI{\tempurl}


\bibitem[Cousot et~al\mbox{.}(2011)]%
        {10.1007/978-3-642-19805-2_31}
\bibfield{author}{\bibinfo{person}{Patrick Cousot}, \bibinfo{person}{Radhia
  Cousot}, {and} \bibinfo{person}{Laurent Mauborgne}.}
  \bibinfo{year}{2011}\natexlab{}.
\newblock \showarticletitle{The Reduced Product of Abstract Domains and the
  Combination of Decision Procedures}. In \bibinfo{booktitle}{\emph{Foundations
  of Software Science and Computational Structures}},
  \bibfield{editor}{\bibinfo{person}{Martin Hofmann}} (Ed.).
  \bibinfo{publisher}{Springer Berlin Heidelberg}, \bibinfo{address}{Berlin,
  Heidelberg}, \bibinfo{pages}{456--472}.
\newblock
\showISBNx{978-3-642-19805-2}


\bibitem[{De Angelis} et~al\mbox{.}(2018)]%
        {DBLP:journals/tplp/AngelisFPP18a}
\bibfield{author}{\bibinfo{person}{Emanuele {De Angelis}},
  \bibinfo{person}{Fabio Fioravanti}, \bibinfo{person}{Alberto Pettorossi},
  {and} \bibinfo{person}{Maurizio Proietti}.} \bibinfo{year}{2018}\natexlab{}.
\newblock \showarticletitle{Solving {Horn} Clauses on Inductive Data Types
  Without Induction}.
\newblock \bibinfo{journal}{\emph{Theory Pract. Log. Program.}}
  \bibinfo{volume}{18}, \bibinfo{number}{3-4} (\bibinfo{year}{2018}),
  \bibinfo{pages}{452--469}.
\newblock
\urldef\tempurl%
\url{https://doi.org/10.1017/S1471068418000157}
\showDOI{\tempurl}


\bibitem[de~Moura and Bj{\o}rner(2008)]%
        {DBLP:conf/tacas/MouraB08}
\bibfield{author}{\bibinfo{person}{Leonardo~Mendon{\c{c}}a de Moura} {and}
  \bibinfo{person}{Nikolaj~S. Bj{\o}rner}.} \bibinfo{year}{2008}\natexlab{}.
\newblock \showarticletitle{{Z3:} An Efficient {SMT} Solver}. In
  \bibinfo{booktitle}{\emph{Tools and Algorithms for the Construction and
  Analysis of Systems, 14th International Conference, {TACAS} 2008, Held as
  Part of the Joint European Conferences on Theory and Practice of Software,
  {ETAPS} 2008, Budapest, Hungary, March 29-April 6, 2008. Proceedings}}
  \emph{(\bibinfo{series}{Lecture Notes in Computer Science},
  Vol.~\bibinfo{volume}{4963})}, \bibfield{editor}{\bibinfo{person}{C.~R.
  Ramakrishnan} {and} \bibinfo{person}{Jakob Rehof}} (Eds.).
  \bibinfo{publisher}{Springer}, \bibinfo{pages}{337--340}.
\newblock
\urldef\tempurl%
\url{https://doi.org/10.1007/978-3-540-78800-3\_24}
\showDOI{\tempurl}


\bibitem[Diekert et~al\mbox{.}(2005)]%
        {DIEKERT2005105}
\bibfield{author}{\bibinfo{person}{Volker Diekert}, \bibinfo{person}{Claudio
  Gutierrez}, {and} \bibinfo{person}{Christian Hagenah}.}
  \bibinfo{year}{2005}\natexlab{}.
\newblock \showarticletitle{The existential theory of equations with rational
  constraints in free groups is PSPACE-complete}.
\newblock \bibinfo{journal}{\emph{Information and Computation}}
  \bibinfo{volume}{202}, \bibinfo{number}{2} (\bibinfo{year}{2005}),
  \bibinfo{pages}{105--140}.
\newblock
\showISSN{0890-5401}
\urldef\tempurl%
\url{https://doi.org/10.1016/j.ic.2005.04.002}
\showDOI{\tempurl}


\bibitem[Diekert et~al\mbox{.}(2016)]%
        {DBLP:journals/iandc/DiekertJP16}
\bibfield{author}{\bibinfo{person}{Volker Diekert}, \bibinfo{person}{Artur
  Jez}, {and} \bibinfo{person}{Wojciech Plandowski}.}
  \bibinfo{year}{2016}\natexlab{}.
\newblock \showarticletitle{Finding all solutions of equations in free groups
  and monoids with involution}.
\newblock \bibinfo{journal}{\emph{Inf. Comput.}}  \bibinfo{volume}{251}
  (\bibinfo{year}{2016}), \bibinfo{pages}{263--286}.
\newblock
\urldef\tempurl%
\url{https://doi.org/10.1016/J.IC.2016.09.009}
\showDOI{\tempurl}


\bibitem[Diekert and Muscholl(2006)]%
        {Diekert06}
\bibfield{author}{\bibinfo{person}{Volker Diekert} {and} \bibinfo{person}{Anca
  Muscholl}.} \bibinfo{year}{2006}\natexlab{}.
\newblock \showarticletitle{SOLVABILITY OF EQUATIONS IN GRAPH GROUPS IS
  DECIDABLE}.
\newblock \bibinfo{journal}{\emph{International Journal of Algebra and
  Computation}} \bibinfo{volume}{16}, \bibinfo{number}{06}
  (\bibinfo{year}{2006}), \bibinfo{pages}{1047--1069}.
\newblock
\urldef\tempurl%
\url{https://doi.org/10.1142/S0218196706003372}
\showDOI{\tempurl}
\showeprint{https://doi.org/10.1142/S0218196706003372}


\bibitem[Ezudheen et~al\mbox{.}(2018)]%
        {DBLP:journals/pacmpl/EzudheenND0M18}
\bibfield{author}{\bibinfo{person}{P. Ezudheen}, \bibinfo{person}{Daniel
  Neider}, \bibinfo{person}{Deepak D'Souza}, \bibinfo{person}{Pranav Garg},
  {and} \bibinfo{person}{P. Madhusudan}.} \bibinfo{year}{2018}\natexlab{}.
\newblock \showarticletitle{Horn-{ICE} learning for synthesizing invariants and
  contracts}.
\newblock \bibinfo{journal}{\emph{Proc. {ACM} Program. Lang.}}
  \bibinfo{volume}{2}, \bibinfo{number}{{OOPSLA}} (\bibinfo{year}{2018}),
  \bibinfo{pages}{131:1--131:25}.
\newblock
\urldef\tempurl%
\url{https://doi.org/10.1145/3276501}
\showURL{%
\tempurl}


\bibitem[Garg et~al\mbox{.}(2014)]%
        {garg_2014}
\bibfield{author}{\bibinfo{person}{Pranav Garg}, \bibinfo{person}{Christof
  L{\"{o}}ding}, \bibinfo{person}{P. Madhusudan}, {and} \bibinfo{person}{Daniel
  Neider}.} \bibinfo{year}{2014}\natexlab{}.
\newblock \showarticletitle{{ICE:} {A} Robust Framework for Learning
  Invariants}. In \bibinfo{booktitle}{\emph{Proceedings of CAV 2014}}
  \emph{(\bibinfo{series}{LNCS}, Vol.~\bibinfo{volume}{8559})}.
  \bibinfo{publisher}{Springer}, \bibinfo{pages}{69--87}.
\newblock


\bibitem[Gold(1967)]%
        {GOLD1967447}
\bibfield{author}{\bibinfo{person}{E~Mark Gold}.}
  \bibinfo{year}{1967}\natexlab{}.
\newblock \showarticletitle{Language identification in the limit}.
\newblock \bibinfo{journal}{\emph{Information and Control}}
  \bibinfo{volume}{10}, \bibinfo{number}{5} (\bibinfo{year}{1967}),
  \bibinfo{pages}{447--474}.
\newblock
\showISSN{0019-9958}
\urldef\tempurl%
\url{https://doi.org/10.1016/S0019-9958(67)91165-5}
\showDOI{\tempurl}


\bibitem[Gulwani(2011)]%
        {DBLP:conf/popl/Gulwani11}
\bibfield{author}{\bibinfo{person}{Sumit Gulwani}.}
  \bibinfo{year}{2011}\natexlab{}.
\newblock \showarticletitle{Automating string processing in spreadsheets using
  input-output examples}. In \bibinfo{booktitle}{\emph{Proceedings of the 38th
  {ACM} {SIGPLAN-SIGACT} Symposium on Principles of Programming Languages,
  {POPL} 2011, Austin, TX, USA, January 26-28, 2011}},
  \bibfield{editor}{\bibinfo{person}{Thomas Ball} {and} \bibinfo{person}{Mooly
  Sagiv}} (Eds.). \bibinfo{publisher}{{ACM}}, \bibinfo{pages}{317--330}.
\newblock
\urldef\tempurl%
\url{https://doi.org/10.1145/1926385.1926423}
\showDOI{\tempurl}


\bibitem[Gulwani et~al\mbox{.}(2012)]%
        {DBLP:journals/cacm/GulwaniHS12}
\bibfield{author}{\bibinfo{person}{Sumit Gulwani}, \bibinfo{person}{William~R.
  Harris}, {and} \bibinfo{person}{Rishabh Singh}.}
  \bibinfo{year}{2012}\natexlab{}.
\newblock \showarticletitle{Spreadsheet data manipulation using examples}.
\newblock \bibinfo{journal}{\emph{Commun. {ACM}}} \bibinfo{volume}{55},
  \bibinfo{number}{8} (\bibinfo{year}{2012}), \bibinfo{pages}{97--105}.
\newblock
\urldef\tempurl%
\url{https://doi.org/10.1145/2240236.2240260}
\showDOI{\tempurl}


\bibitem[Gulwani and Tiwari(2006)]%
        {DBLP:conf/pldi/GulwaniT06}
\bibfield{author}{\bibinfo{person}{Sumit Gulwani} {and} \bibinfo{person}{Ashish
  Tiwari}.} \bibinfo{year}{2006}\natexlab{}.
\newblock \showarticletitle{Combining abstract interpreters}. In
  \bibinfo{booktitle}{\emph{Proceedings of the {ACM} {SIGPLAN} 2006 Conference
  on Programming Language Design and Implementation, Ottawa, Ontario, Canada,
  June 11-14, 2006}}, \bibfield{editor}{\bibinfo{person}{Michael~I.
  Schwartzbach} {and} \bibinfo{person}{Thomas Ball}} (Eds.).
  \bibinfo{publisher}{{ACM}}, \bibinfo{pages}{376--386}.
\newblock
\urldef\tempurl%
\url{https://doi.org/10.1145/1133981.1134026}
\showDOI{\tempurl}


\bibitem[Gurfinkel et~al\mbox{.}(2015)]%
        {DBLP:conf/tacas/GurfinkelKN15}
\bibfield{author}{\bibinfo{person}{Arie Gurfinkel}, \bibinfo{person}{Temesghen
  Kahsai}, {and} \bibinfo{person}{Jorge~A. Navas}.}
  \bibinfo{year}{2015}\natexlab{}.
\newblock \showarticletitle{{SeaHorn}: {A} Framework for Verifying {C} Programs
  (Competition Contribution)}. In \bibinfo{booktitle}{\emph{Tools and
  Algorithms for the Construction and Analysis of Systems - 21st International
  Conference, {TACAS} 2015, Held as Part of the European Joint Conferences on
  Theory and Practice of Software, {ETAPS} 2015, London, UK, April 11-18, 2015.
  Proceedings}} \emph{(\bibinfo{series}{Lecture Notes in Computer Science},
  Vol.~\bibinfo{volume}{9035})}, \bibfield{editor}{\bibinfo{person}{Christel
  Baier} {and} \bibinfo{person}{Cesare Tinelli}} (Eds.).
  \bibinfo{publisher}{Springer}, \bibinfo{pages}{447--450}.
\newblock
\urldef\tempurl%
\url{https://doi.org/10.1007/978-3-662-46681-0\_41}
\showDOI{\tempurl}


\bibitem[Henzinger et~al\mbox{.}(2004)]%
        {DBLP:conf/popl/HenzingerJMM04}
\bibfield{author}{\bibinfo{person}{Thomas~A. Henzinger},
  \bibinfo{person}{Ranjit Jhala}, \bibinfo{person}{Rupak Majumdar}, {and}
  \bibinfo{person}{Kenneth~L. McMillan}.} \bibinfo{year}{2004}\natexlab{}.
\newblock \showarticletitle{Abstractions from proofs}. In
  \bibinfo{booktitle}{\emph{Proceedings of the 31st {ACM} {SIGPLAN-SIGACT}
  Symposium on Principles of Programming Languages, {POPL} 2004, Venice, Italy,
  January 14-16, 2004}}, \bibfield{editor}{\bibinfo{person}{Neil~D. Jones}
  {and} \bibinfo{person}{Xavier Leroy}} (Eds.). \bibinfo{publisher}{{ACM}},
  \bibinfo{pages}{232--244}.
\newblock
\urldef\tempurl%
\url{https://doi.org/10.1145/964001.964021}
\showDOI{\tempurl}


\bibitem[{Hojjat} and {Rümmer}(2018)]%
        {Eldarica}
\bibfield{author}{\bibinfo{person}{H. {Hojjat}} {and} \bibinfo{person}{P.
  {Rümmer}}.} \bibinfo{year}{2018}\natexlab{}.
\newblock \showarticletitle{The {ELDARICA} {Horn} Solver}. In
  \bibinfo{booktitle}{\emph{2018 Formal Methods in Computer Aided Design
  (FMCAD)}}. \bibinfo{pages}{1--7}.
\newblock


\bibitem[Ikeda et~al\mbox{.}(2023)]%
        {DBLP:conf/aplas/IkedaSK23}
\bibfield{author}{\bibinfo{person}{Ryo Ikeda}, \bibinfo{person}{Ryosuke Sato},
  {and} \bibinfo{person}{Naoki Kobayashi}.} \bibinfo{year}{2023}\natexlab{}.
\newblock \showarticletitle{Argument Reduction of Constrained {Horn} Clauses
  Using Equality Constraints}. In \bibinfo{booktitle}{\emph{Programming
  Languages and Systems - 21st Asian Symposium, {APLAS} 2023, Taipei, Taiwan,
  November 26-29, 2023, Proceedings}} \emph{(\bibinfo{series}{Lecture Notes in
  Computer Science}, Vol.~\bibinfo{volume}{14405})},
  \bibfield{editor}{\bibinfo{person}{Chung{-}Kil Hur}} (Ed.).
  \bibinfo{publisher}{Springer}, \bibinfo{pages}{246--265}.
\newblock
\urldef\tempurl%
\url{https://doi.org/10.1007/978-981-99-8311-7_12}
\showDOI{\tempurl}


\bibitem[Itzhaky et~al\mbox{.}(2014)]%
        {10.1007/978-3-319-08867-9_3}
\bibfield{author}{\bibinfo{person}{Shachar Itzhaky}, \bibinfo{person}{Nikolaj
  Bj{\o}rner}, \bibinfo{person}{Thomas Reps}, \bibinfo{person}{Mooly Sagiv},
  {and} \bibinfo{person}{Aditya Thakur}.} \bibinfo{year}{2014}\natexlab{}.
\newblock \showarticletitle{Property-Directed Shape Analysis}. In
  \bibinfo{booktitle}{\emph{Computer Aided Verification}},
  \bibfield{editor}{\bibinfo{person}{Armin Biere} {and}
  \bibinfo{person}{Roderick Bloem}} (Eds.). \bibinfo{publisher}{Springer
  International Publishing}, \bibinfo{address}{Cham}, \bibinfo{pages}{35--51}.
\newblock
\showISBNx{978-3-319-08867-9}


\bibitem[Jiang et~al\mbox{.}(1995)]%
        {JIANG199553}
\bibfield{author}{\bibinfo{person}{T. Jiang}, \bibinfo{person}{A. Salomaa},
  \bibinfo{person}{K. Salomaa}, {and} \bibinfo{person}{S. Yu}.}
  \bibinfo{year}{1995}\natexlab{}.
\newblock \showarticletitle{Decision Problems for Patterns}.
\newblock \bibinfo{journal}{\emph{J. Comput. System Sci.}}
  \bibinfo{volume}{50}, \bibinfo{number}{1} (\bibinfo{year}{1995}),
  \bibinfo{pages}{53--63}.
\newblock
\showISSN{0022-0000}
\urldef\tempurl%
\url{https://doi.org/10.1006/jcss.1995.1006}
\showDOI{\tempurl}


\bibitem[K. et~al\mbox{.}(2022)]%
        {DBLP:journals/pacmpl/KSG22}
\bibfield{author}{\bibinfo{person}{Hari Govind~V. K.}, \bibinfo{person}{Sharon
  Shoham}, {and} \bibinfo{person}{Arie Gurfinkel}.}
  \bibinfo{year}{2022}\natexlab{}.
\newblock \showarticletitle{Solving constrained {Horn} clauses modulo algebraic
  data types and recursive functions}.
\newblock \bibinfo{journal}{\emph{Proc. {ACM} Program. Lang.}}
  \bibinfo{volume}{6}, \bibinfo{number}{{POPL}} (\bibinfo{year}{2022}),
  \bibinfo{pages}{1--29}.
\newblock
\urldef\tempurl%
\url{https://doi.org/10.1145/3498722}
\showDOI{\tempurl}


\bibitem[Kahsai et~al\mbox{.}(2016)]%
        {DBLP:conf/cav/KahsaiRSS16}
\bibfield{author}{\bibinfo{person}{Temesghen Kahsai}, \bibinfo{person}{Philipp
  R{\"{u}}mmer}, \bibinfo{person}{Huascar Sanchez}, {and}
  \bibinfo{person}{Martin Sch{\"{a}}f}.} \bibinfo{year}{2016}\natexlab{}.
\newblock \showarticletitle{{JayHorn}: {A} Framework for Verifying {Java}
  programs}. In \bibinfo{booktitle}{\emph{Computer Aided Verification - 28th
  International Conference, {CAV} 2016, Toronto, ON, Canada, July 17-23, 2016,
  Proceedings, Part {I}}} \emph{(\bibinfo{series}{Lecture Notes in Computer
  Science}, Vol.~\bibinfo{volume}{9779})},
  \bibfield{editor}{\bibinfo{person}{Swarat Chaudhuri} {and}
  \bibinfo{person}{Azadeh Farzan}} (Eds.). \bibinfo{publisher}{Springer},
  \bibinfo{pages}{352--358}.
\newblock
\urldef\tempurl%
\url{https://doi.org/10.1007/978-3-319-41528-4\_19}
\showDOI{\tempurl}


\bibitem[Karbyshev et~al\mbox{.}(2017)]%
        {10.1145/3022187}
\bibfield{author}{\bibinfo{person}{Aleksandr Karbyshev},
  \bibinfo{person}{Nikolaj Bj\o{}rner}, \bibinfo{person}{Shachar Itzhaky},
  \bibinfo{person}{Noam Rinetzky}, {and} \bibinfo{person}{Sharon Shoham}.}
  \bibinfo{year}{2017}\natexlab{}.
\newblock \showarticletitle{Property-Directed Inference of Universal Invariants
  or Proving Their Absence}.
\newblock \bibinfo{journal}{\emph{J. ACM}} \bibinfo{volume}{64},
  \bibinfo{number}{1}, Article \bibinfo{articleno}{7} (\bibinfo{date}{March}
  \bibinfo{year}{2017}), \bibinfo{numpages}{33}~pages.
\newblock
\showISSN{0004-5411}
\urldef\tempurl%
\url{https://doi.org/10.1145/3022187}
\showDOI{\tempurl}


\bibitem[Karhum\"{a}ki et~al\mbox{.}(2000)]%
        {10.1145/337244.337255}
\bibfield{author}{\bibinfo{person}{Juhani Karhum\"{a}ki},
  \bibinfo{person}{Filippo Mignosi}, {and} \bibinfo{person}{Wojciech
  Plandowski}.} \bibinfo{year}{2000}\natexlab{}.
\newblock \showarticletitle{The expressibility of languages and relations by
  word equations}.
\newblock \bibinfo{journal}{\emph{J. ACM}} \bibinfo{volume}{47},
  \bibinfo{number}{3} (\bibinfo{date}{May} \bibinfo{year}{2000}),
  \bibinfo{pages}{483–505}.
\newblock
\showISSN{0004-5411}
\urldef\tempurl%
\url{https://doi.org/10.1145/337244.337255}
\showDOI{\tempurl}


\bibitem[Katsura et~al\mbox{.}(2025)]%
        {KatsuraSAS25}
\bibfield{author}{\bibinfo{person}{Hiroyuki Katsura}, \bibinfo{person}{Naoki
  Kobayashi}, \bibinfo{person}{Ken Sakayori}, {and} \bibinfo{person}{Ryosuke
  Sato}.} \bibinfo{year}{2025}\natexlab{}.
\newblock \showarticletitle{Automated Catamorphism Synthesis for Solving
  Constrained {Horn} Clauses over Algebraic Data Types}. In
  \bibinfo{booktitle}{\emph{Static Analysis - 32nd International Symposium,
  {SAS} 2025, Singapore, October 13-14, 2025, Proceedings}}
  \emph{(\bibinfo{series}{Lecture Notes in Computer Science},
  Vol.~\bibinfo{volume}{16100})}, \bibfield{editor}{\bibinfo{person}{Hakjoo Oh}
  {and} \bibinfo{person}{Yulei Sui}} (Eds.). \bibinfo{publisher}{Springer},
  \bibinfo{pages}{305--327}.
\newblock
\urldef\tempurl%
\url{https://doi.org/10.1007/978-3-032-07106-4_13}
\showDOI{\tempurl}


\bibitem[Kobayashi et~al\mbox{.}(2026)]%
        {Artifact}
\bibfield{author}{\bibinfo{person}{Naoki Kobayashi}, \bibinfo{person}{Ryosuke
  Sato}, \bibinfo{person}{Ayumi Shinohara}, {and} \bibinfo{person}{Yoshinaka
  Ryo}.} \bibinfo{year}{2026}\natexlab{}.
\newblock \bibinfo{booktitle}{\emph{Solvable Tuple Patterns and Their
  Applications to Program Verification}}.
\newblock
\urldef\tempurl%
\url{https://doi.org/10.5281/zenodo.19907200}
\showDOI{\tempurl}


\bibitem[Kobayashi et~al\mbox{.}(2025a)]%
        {STPlong}
\bibfield{author}{\bibinfo{person}{Naoki Kobayashi}, \bibinfo{person}{Ryosuke
  Sato}, \bibinfo{person}{Ayumi Shinohara}, {and} \bibinfo{person}{Ryo
  Yoshinaka}.} \bibinfo{year}{2025}\natexlab{a}.
\newblock \showarticletitle{Solvable Tuple Patterns and Their Applications to
  Program Verification}.
\newblock \bibinfo{journal}{\emph{CoRR}}  \bibinfo{volume}{abs/2508.20365}
  (\bibinfo{year}{2025}).
\newblock
\urldef\tempurl%
\url{https://doi.org/10.48550/ARXIV.2508.20365}
\showDOI{\tempurl}
\showeprint[arXiv]{2508.20365}


\bibitem[Kobayashi et~al\mbox{.}(2025b)]%
        {KobayashiNeuGuS}
\bibfield{author}{\bibinfo{person}{Naoki Kobayashi}, \bibinfo{person}{Taro
  Sekiyama}, \bibinfo{person}{Issei Sato}, {and} \bibinfo{person}{Hiroshi
  Unno}.} \bibinfo{year}{2025}\natexlab{b}.
\newblock \showarticletitle{Towards Neural-Network-Guided Program Synthesis and
  Verification}.
\newblock \bibinfo{journal}{\emph{Formal Methods in System Design}}
  (\bibinfo{year}{2025}).
\newblock


\bibitem[Komuravelli et~al\mbox{.}(2016)]%
        {DBLP:journals/fmsd/KomuravelliGC16}
\bibfield{author}{\bibinfo{person}{Anvesh Komuravelli}, \bibinfo{person}{Arie
  Gurfinkel}, {and} \bibinfo{person}{Sagar Chaki}.}
  \bibinfo{year}{2016}\natexlab{}.
\newblock \showarticletitle{SMT-based model checking for recursive programs}.
\newblock \bibinfo{journal}{\emph{Formal Methods Syst. Des.}}
  \bibinfo{volume}{48}, \bibinfo{number}{3} (\bibinfo{year}{2016}),
  \bibinfo{pages}{175--205}.
\newblock
\urldef\tempurl%
\url{https://doi.org/10.1007/s10703-016-0249-4}
\showURL{%
\tempurl}


\bibitem[Kostyukov et~al\mbox{.}(2021)]%
        {DBLP:conf/pldi/KostyukovMF21}
\bibfield{author}{\bibinfo{person}{Yurii Kostyukov}, \bibinfo{person}{Dmitry
  Mordvinov}, {and} \bibinfo{person}{Grigory Fedyukovich}.}
  \bibinfo{year}{2021}\natexlab{}.
\newblock \showarticletitle{Beyond the elementary representations of program
  invariants over algebraic data types}. In \bibinfo{booktitle}{\emph{{PLDI}
  '21: 42nd {ACM} {SIGPLAN} International Conference on Programming Language
  Design and Implementation, Virtual Event, Canada, June 20-25, 2021}},
  \bibfield{editor}{\bibinfo{person}{Stephen~N. Freund} {and}
  \bibinfo{person}{Eran Yahav}} (Eds.). \bibinfo{publisher}{{ACM}},
  \bibinfo{pages}{451--465}.
\newblock
\urldef\tempurl%
\url{https://doi.org/10.1145/3453483.3454055}
\showDOI{\tempurl}


\bibitem[Kuncak et~al\mbox{.}(2006)]%
        {DBLP:journals/jar/KuncakNR06}
\bibfield{author}{\bibinfo{person}{Viktor Kuncak}, \bibinfo{person}{Huu~Hai
  Nguyen}, {and} \bibinfo{person}{Martin~C. Rinard}.}
  \bibinfo{year}{2006}\natexlab{}.
\newblock \showarticletitle{Deciding Boolean Algebra with Presburger
  Arithmetic}.
\newblock \bibinfo{journal}{\emph{J. Autom. Reason.}} \bibinfo{volume}{36},
  \bibinfo{number}{3} (\bibinfo{year}{2006}), \bibinfo{pages}{213--239}.
\newblock
\urldef\tempurl%
\url{https://doi.org/10.1007/S10817-006-9042-1}
\showDOI{\tempurl}


\bibitem[Losekoot et~al\mbox{.}(2023)]%
        {DBLP:conf/fscd/LosekootGJ23}
\bibfield{author}{\bibinfo{person}{Th{\'{e}}o Losekoot},
  \bibinfo{person}{Thomas Genet}, {and} \bibinfo{person}{Thomas~P. Jensen}.}
  \bibinfo{year}{2023}\natexlab{}.
\newblock \showarticletitle{Automata-Based Verification of Relational
  Properties of Functions over Algebraic Data Structures}. In
  \bibinfo{booktitle}{\emph{8th International Conference on Formal Structures
  for Computation and Deduction, {FSCD} 2023, July 3-6, 2023, Rome, Italy}}
  \emph{(\bibinfo{series}{LIPIcs}, Vol.~\bibinfo{volume}{260})},
  \bibfield{editor}{\bibinfo{person}{Marco Gaboardi} {and}
  \bibinfo{person}{Femke van Raamsdonk}} (Eds.). \bibinfo{publisher}{Schloss
  Dagstuhl - Leibniz-Zentrum f{\"{u}}r Informatik}, \bibinfo{pages}{7:1--7:22}.
\newblock
\urldef\tempurl%
\url{https://doi.org/10.4230/LIPICS.FSCD.2023.7}
\showDOI{\tempurl}


\bibitem[Losekoot et~al\mbox{.}(2024)]%
        {DBLP:conf/sas/LosekootGJ24}
\bibfield{author}{\bibinfo{person}{Th{\'{e}}o Losekoot},
  \bibinfo{person}{Thomas Genet}, {and} \bibinfo{person}{Thomas~P. Jensen}.}
  \bibinfo{year}{2024}\natexlab{}.
\newblock \showarticletitle{Verification of Programs with {ADTs} Using Shallow
  {Horn} Clauses}. In \bibinfo{booktitle}{\emph{Static Analysis - 31st
  International Symposium, {SAS} 2024, Pasadena, CA, USA, October 20-22, 2024,
  Proceedings}} \emph{(\bibinfo{series}{Lecture Notes in Computer Science},
  Vol.~\bibinfo{volume}{14995})}, \bibfield{editor}{\bibinfo{person}{Roberto
  Giacobazzi} {and} \bibinfo{person}{Alessandra Gorla}} (Eds.).
  \bibinfo{publisher}{Springer}, \bibinfo{pages}{242--267}.
\newblock
\urldef\tempurl%
\url{https://doi.org/10.1007/978-3-031-74776-2\_10}
\showDOI{\tempurl}


\bibitem[Lothaire(2002)]%
        {Lothaire}
\bibfield{author}{\bibinfo{person}{M. Lothaire}.}
  \bibinfo{year}{2002}\natexlab{}.
\newblock \bibinfo{booktitle}{\emph{Algebraic Combinatorics on Words}}.
  \bibinfo{series}{Encyclopedia of Mathematics and its Applications},
  Vol.~\bibinfo{volume}{90}.
\newblock \bibinfo{publisher}{Cambridge University Press}.
\newblock


\bibitem[{Makanin}(1977)]%
        {Makanin}
\bibfield{author}{\bibinfo{person}{G.~S. {Makanin}}.}
  \bibinfo{year}{1977}\natexlab{}.
\newblock \showarticletitle{{The Problem of Solvability of Equations in a Free
  Semigroup}}.
\newblock \bibinfo{journal}{\emph{Sbornik: Mathematics}} \bibinfo{volume}{32},
  \bibinfo{number}{2} (\bibinfo{date}{Feb.} \bibinfo{year}{1977}),
  \bibinfo{pages}{129--198}.
\newblock
\urldef\tempurl%
\url{https://doi.org/10.1070/SM1977v032n02ABEH002376}
\showDOI{\tempurl}


\bibitem[Matsushita et~al\mbox{.}(2021)]%
        {DBLP:journals/toplas/MatsushitaTK21}
\bibfield{author}{\bibinfo{person}{Yusuke Matsushita}, \bibinfo{person}{Takeshi
  Tsukada}, {and} \bibinfo{person}{Naoki Kobayashi}.}
  \bibinfo{year}{2021}\natexlab{}.
\newblock \showarticletitle{{RustHorn}: {CHC}-based Verification for Rust
  Programs}.
\newblock \bibinfo{journal}{\emph{{ACM} Trans. Program. Lang. Syst.}}
  \bibinfo{volume}{43}, \bibinfo{number}{4} (\bibinfo{year}{2021}),
  \bibinfo{pages}{15:1--15:54}.
\newblock
\urldef\tempurl%
\url{https://doi.org/10.1145/3462205}
\showDOI{\tempurl}


\bibitem[McMillan(2005)]%
        {DBLP:journals/tcs/McMillan05}
\bibfield{author}{\bibinfo{person}{Kenneth~L. McMillan}.}
  \bibinfo{year}{2005}\natexlab{}.
\newblock \showarticletitle{An interpolating theorem prover}.
\newblock \bibinfo{journal}{\emph{Theor. Comput. Sci.}} \bibinfo{volume}{345},
  \bibinfo{number}{1} (\bibinfo{year}{2005}), \bibinfo{pages}{101--121}.
\newblock
\urldef\tempurl%
\url{https://doi.org/10.1016/J.TCS.2005.07.003}
\showDOI{\tempurl}


\bibitem[Mitchell(1998)]%
        {Mitchell98}
\bibfield{author}{\bibinfo{person}{Andrew~R. Mitchell}.}
  \bibinfo{year}{1998}\natexlab{}.
\newblock \showarticletitle{Learnability of a Subclass of Extended Pattern
  Languages}. In \bibinfo{booktitle}{\emph{Proceedings of the Eleventh Annual
  Conference on Computational Learning Theory, {COLT} 1998, Madison, Wisconsin,
  USA, July 24-26, 1998}}, \bibfield{editor}{\bibinfo{person}{Peter~L.
  Bartlett} {and} \bibinfo{person}{Yishay Mansour}} (Eds.).
  \bibinfo{publisher}{{ACM}}, \bibinfo{pages}{64--71}.
\newblock
\urldef\tempurl%
\url{https://doi.org/10.1145/279943.279955}
\showDOI{\tempurl}


\bibitem[Mordvinov and Fedyukovich(2017)]%
        {DBLP:conf/lpar/MordvinovF17}
\bibfield{author}{\bibinfo{person}{Dmitry Mordvinov} {and}
  \bibinfo{person}{Grigory Fedyukovich}.} \bibinfo{year}{2017}\natexlab{}.
\newblock \showarticletitle{Synchronizing Constrained {Horn} Clauses}. In
  \bibinfo{booktitle}{\emph{LPAR-21, 21st International Conference on Logic for
  Programming, Artificial Intelligence and Reasoning, Maun, Botswana, May 7-12,
  2017}} \emph{(\bibinfo{series}{EPiC Series in Computing},
  Vol.~\bibinfo{volume}{46})}, \bibfield{editor}{\bibinfo{person}{Thomas Eiter}
  {and} \bibinfo{person}{David Sands}} (Eds.). \bibinfo{publisher}{EasyChair},
  \bibinfo{pages}{338--355}.
\newblock
\urldef\tempurl%
\url{https://doi.org/10.29007/GR5C}
\showDOI{\tempurl}


\bibitem[Nelson and Oppen(1980)]%
        {10.1145/322186.322198}
\bibfield{author}{\bibinfo{person}{Greg Nelson} {and} \bibinfo{person}{Derek~C.
  Oppen}.} \bibinfo{year}{1980}\natexlab{}.
\newblock \showarticletitle{Fast Decision Procedures Based on Congruence
  Closure}.
\newblock \bibinfo{journal}{\emph{J. ACM}} \bibinfo{volume}{27},
  \bibinfo{number}{2} (\bibinfo{date}{April} \bibinfo{year}{1980}),
  \bibinfo{pages}{356–364}.
\newblock
\showISSN{0004-5411}
\urldef\tempurl%
\url{https://doi.org/10.1145/322186.322198}
\showDOI{\tempurl}


\bibitem[Nowotka and Wiedenh\"{o}ft(2025)]%
        {nowotka_et_al:LIPIcs.CPM.2025.4}
\bibfield{author}{\bibinfo{person}{Dirk Nowotka} {and} \bibinfo{person}{Max
  Wiedenh\"{o}ft}.} \bibinfo{year}{2025}\natexlab{}.
\newblock \showarticletitle{{The Equivalence Problem of E-Pattern Languages
  with Length Constraints Is Undecidable}}. In \bibinfo{booktitle}{\emph{36th
  Annual Symposium on Combinatorial Pattern Matching (CPM 2025)}}
  \emph{(\bibinfo{series}{Leibniz International Proceedings in Informatics
  (LIPIcs)}, Vol.~\bibinfo{volume}{331})},
  \bibfield{editor}{\bibinfo{person}{Paola Bonizzoni} {and}
  \bibinfo{person}{Veli M\"{a}kinen}} (Eds.). \bibinfo{publisher}{Schloss
  Dagstuhl -- Leibniz-Zentrum f{\"u}r Informatik}, \bibinfo{address}{Dagstuhl,
  Germany}, \bibinfo{pages}{4:1--4:23}.
\newblock
\showISBNx{978-3-95977-369-0}
\showISSN{1868-8969}
\urldef\tempurl%
\url{https://doi.org/10.4230/LIPIcs.CPM.2025.4}
\showDOI{\tempurl}


\bibitem[Okasaki(1995)]%
        {FuncQueue}
\bibfield{author}{\bibinfo{person}{Chris Okasaki}.}
  \bibinfo{year}{1995}\natexlab{}.
\newblock \showarticletitle{Simple and efficient purely functional queues and
  deques}.
\newblock \bibinfo{journal}{\emph{Journal of Functional Programming}}
  \bibinfo{volume}{5}, \bibinfo{number}{4} (\bibinfo{year}{1995}),
  \bibinfo{pages}{583–592}.
\newblock
\urldef\tempurl%
\url{https://doi.org/10.1017/S0956796800001489}
\showDOI{\tempurl}


\bibitem[Pettorossi and Proietti(2000)]%
        {PettorossiP00}
\bibfield{author}{\bibinfo{person}{Alberto Pettorossi} {and}
  \bibinfo{person}{Maurizio Proietti}.} \bibinfo{year}{2000}\natexlab{}.
\newblock \showarticletitle{Perfect Model Checking via Unfold/Fold
  Transformations}. In \bibinfo{booktitle}{\emph{Computational Logic - {CL}
  2000, First International Conference, London, UK, 24-28 July, 2000,
  Proceedings}} \emph{(\bibinfo{series}{Lecture Notes in Computer Science},
  Vol.~\bibinfo{volume}{1861})}, \bibfield{editor}{\bibinfo{person}{John~W.
  Lloyd}, \bibinfo{person}{Ver{\'{o}}nica Dahl}, \bibinfo{person}{Ulrich
  Furbach}, \bibinfo{person}{Manfred Kerber}, \bibinfo{person}{Kung{-}Kiu Lau},
  \bibinfo{person}{Catuscia Palamidessi}, \bibinfo{person}{Lu{\'{\i}}s~Moniz
  Pereira}, \bibinfo{person}{Yehoshua Sagiv}, {and} \bibinfo{person}{Peter~J.
  Stuckey}} (Eds.). \bibinfo{publisher}{Springer}, \bibinfo{pages}{613--628}.
\newblock
\urldef\tempurl%
\url{https://doi.org/10.1007/3-540-44957-4\_41}
\showDOI{\tempurl}


\bibitem[Piskac and Kuncak(2008)]%
        {multiset}
\bibfield{author}{\bibinfo{person}{Ruzica Piskac} {and} \bibinfo{person}{Viktor
  Kuncak}.} \bibinfo{year}{2008}\natexlab{}.
\newblock \showarticletitle{Decision Procedures for Multisets with Cardinality
  Constraints}. In \bibinfo{booktitle}{\emph{Verification, Model Checking, and
  Abstract Interpretation}}, \bibfield{editor}{\bibinfo{person}{Francesco
  Logozzo}, \bibinfo{person}{Doron~A. Peled}, {and} \bibinfo{person}{Lenore~D.
  Zuck}} (Eds.). \bibinfo{publisher}{Springer Berlin Heidelberg},
  \bibinfo{address}{Berlin, Heidelberg}, \bibinfo{pages}{218--232}.
\newblock
\showISBNx{978-3-540-78163-9}


\bibitem[Plandowski(1999)]%
        {Plandowski}
\bibfield{author}{\bibinfo{person}{W. Plandowski}.}
  \bibinfo{year}{1999}\natexlab{}.
\newblock \showarticletitle{Satisfiability of word equations with constants is
  in PSPACE}. In \bibinfo{booktitle}{\emph{40th Annual Symposium on Foundations
  of Computer Science (Cat. No.99CB37039)}}. \bibinfo{pages}{495--500}.
\newblock
\urldef\tempurl%
\url{https://doi.org/10.1109/SFFCS.1999.814622}
\showDOI{\tempurl}


\bibitem[Pugh(1990)]%
        {DBLP:journals/cacm/Pugh90}
\bibfield{author}{\bibinfo{person}{William~W. Pugh}.}
  \bibinfo{year}{1990}\natexlab{}.
\newblock \showarticletitle{Skip Lists: {A} Probabilistic Alternative to
  Balanced Trees}.
\newblock \bibinfo{journal}{\emph{Commun. {ACM}}} \bibinfo{volume}{33},
  \bibinfo{number}{6} (\bibinfo{year}{1990}), \bibinfo{pages}{668--676}.
\newblock
\urldef\tempurl%
\url{https://doi.org/10.1145/78973.78977}
\showDOI{\tempurl}


\bibitem[Reidenbach(2006)]%
        {REIDENBACH200691}
\bibfield{author}{\bibinfo{person}{Daniel Reidenbach}.}
  \bibinfo{year}{2006}\natexlab{}.
\newblock \showarticletitle{A non-learnable class of E-pattern languages}.
\newblock \bibinfo{journal}{\emph{Theoretical Computer Science}}
  \bibinfo{volume}{350}, \bibinfo{number}{1} (\bibinfo{year}{2006}),
  \bibinfo{pages}{91--102}.
\newblock
\showISSN{0304-3975}
\urldef\tempurl%
\url{https://doi.org/10.1016/j.tcs.2005.10.017}
\showDOI{\tempurl}
\newblock
\shownote{Algorithmic Learning Theory (ALT 2002)}.


\bibitem[Reidenbach(2008)]%
        {REIDENBACH2008166}
\bibfield{author}{\bibinfo{person}{Daniel Reidenbach}.}
  \bibinfo{year}{2008}\natexlab{}.
\newblock \showarticletitle{Discontinuities in pattern inference}.
\newblock \bibinfo{journal}{\emph{Theoretical Computer Science}}
  \bibinfo{volume}{397}, \bibinfo{number}{1} (\bibinfo{year}{2008}),
  \bibinfo{pages}{166--193}.
\newblock
\showISSN{0304-3975}
\urldef\tempurl%
\url{https://doi.org/10.1016/j.tcs.2008.02.029}
\showDOI{\tempurl}
\newblock
\shownote{Forty Years of Inductive Inference: Dedicated to the 60th Birthday of
  Rolf Wiehagen}.


\bibitem[Ryan et~al\mbox{.}(2020)]%
        {DBLP:conf/iclr/RyanWYGJ20}
\bibfield{author}{\bibinfo{person}{Gabriel Ryan}, \bibinfo{person}{Justin
  Wong}, \bibinfo{person}{Jianan Yao}, \bibinfo{person}{Ronghui Gu}, {and}
  \bibinfo{person}{Suman Jana}.} \bibinfo{year}{2020}\natexlab{}.
\newblock \showarticletitle{{CLN2INV:} Learning Loop Invariants with Continuous
  Logic Networks}. In \bibinfo{booktitle}{\emph{8th International Conference on
  Learning Representations, {ICLR} 2020}}. \bibinfo{publisher}{OpenReview.net}.
\newblock


\bibitem[Sharma et~al\mbox{.}(2013)]%
        {10.1007/978-3-642-37036-6_31}
\bibfield{author}{\bibinfo{person}{Rahul Sharma}, \bibinfo{person}{Saurabh
  Gupta}, \bibinfo{person}{Bharath Hariharan}, \bibinfo{person}{Alex Aiken},
  \bibinfo{person}{Percy Liang}, {and} \bibinfo{person}{Aditya~V. Nori}.}
  \bibinfo{year}{2013}\natexlab{}.
\newblock \showarticletitle{A Data Driven Approach for Algebraic Loop
  Invariants}. In \bibinfo{booktitle}{\emph{Programming Languages and
  Systems}}, \bibfield{editor}{\bibinfo{person}{Matthias Felleisen} {and}
  \bibinfo{person}{Philippa Gardner}} (Eds.). \bibinfo{publisher}{Springer
  Berlin Heidelberg}, \bibinfo{address}{Berlin, Heidelberg},
  \bibinfo{pages}{574--592}.
\newblock
\showISBNx{978-3-642-37036-6}


\bibitem[Shinohara(1982)]%
        {SHINOHARA1982}
\bibfield{author}{\bibinfo{person}{Takeshi Shinohara}.}
  \bibinfo{year}{1982}\natexlab{}.
\newblock \showarticletitle{Polynomial time inference of pattern languages and
  its application}. In \bibinfo{booktitle}{\emph{Seventh IBM Symposium on
  Mathematical Foundations of Computer Science}}. \bibinfo{pages}{191--209}.
\newblock


\bibitem[Shinohara(1983)]%
        {10.1007/3-540-11980-9_19}
\bibfield{author}{\bibinfo{person}{Takeshi Shinohara}.}
  \bibinfo{year}{1983}\natexlab{}.
\newblock \showarticletitle{Polynomial time inference of extended regular
  pattern languages}. In \bibinfo{booktitle}{\emph{RIMS Symposia on Software
  Science and Engineering}}, \bibfield{editor}{\bibinfo{person}{Eiichi Goto},
  \bibinfo{person}{Koichi Furukawa}, \bibinfo{person}{Reiji Nakajima},
  \bibinfo{person}{Ikuo Nakata}, {and} \bibinfo{person}{Akinori Yonezawa}}
  (Eds.). \bibinfo{publisher}{Springer Berlin Heidelberg},
  \bibinfo{address}{Berlin, Heidelberg}, \bibinfo{pages}{115--127}.
\newblock
\showISBNx{978-3-540-39442-6}


\bibitem[Shinohara(1994)]%
        {SHINOHARA1994175}
\bibfield{author}{\bibinfo{person}{Takeshi Shinohara}.}
  \bibinfo{year}{1994}\natexlab{}.
\newblock \showarticletitle{Rich Classes Inferable from Positive Data:
  Length-Bounded Elementary Formal Systems}.
\newblock \bibinfo{journal}{\emph{Information and Computation}}
  \bibinfo{volume}{108}, \bibinfo{number}{2} (\bibinfo{year}{1994}),
  \bibinfo{pages}{175--186}.
\newblock
\showISSN{0890-5401}
\urldef\tempurl%
\url{https://doi.org/10.1006/inco.1994.1006}
\showDOI{\tempurl}


\bibitem[Toman et~al\mbox{.}(2020)]%
        {DBLP:conf/esop/TomanSSI020}
\bibfield{author}{\bibinfo{person}{John Toman}, \bibinfo{person}{Ren Siqi},
  \bibinfo{person}{Kohei Suenaga}, \bibinfo{person}{Atsushi Igarashi}, {and}
  \bibinfo{person}{Naoki Kobayashi}.} \bibinfo{year}{2020}\natexlab{}.
\newblock \showarticletitle{ConSORT: Context- and Flow-Sensitive Ownership
  Refinement Types for Imperative Programs}. In
  \bibinfo{booktitle}{\emph{Programming Languages and Systems - 29th European
  Symposium on Programming, {ESOP} 2020, Held as Part of the European Joint
  Conferences on Theory and Practice of Software, {ETAPS} 2020, Dublin,
  Ireland, April 25-30, 2020, Proceedings}} \emph{(\bibinfo{series}{Lecture
  Notes in Computer Science}, Vol.~\bibinfo{volume}{12075})},
  \bibfield{editor}{\bibinfo{person}{Peter M{\"{u}}ller}} (Ed.).
  \bibinfo{publisher}{Springer}, \bibinfo{pages}{684--714}.
\newblock
\urldef\tempurl%
\url{https://doi.org/10.1007/978-3-030-44914-8\_25}
\showDOI{\tempurl}


\bibitem[Unno and Kobayashi(2009)]%
        {Unno09PPDP}
\bibfield{author}{\bibinfo{person}{Hiroshi Unno} {and} \bibinfo{person}{Naoki
  Kobayashi}.} \bibinfo{year}{2009}\natexlab{}.
\newblock \showarticletitle{Dependent type inference with interpolants}. In
  \bibinfo{booktitle}{\emph{Proceedings of the 11th International {ACM}
  {SIGPLAN} Conference on Principles and Practice of Declarative Programming,
  September 7-9, 2009, Coimbra, Portugal}}. \bibinfo{publisher}{{ACM}},
  \bibinfo{pages}{277--288}.
\newblock


\bibitem[Unno et~al\mbox{.}(2021)]%
        {DBLP:conf/cav/UnnoTK21}
\bibfield{author}{\bibinfo{person}{Hiroshi Unno}, \bibinfo{person}{Tachio
  Terauchi}, {and} \bibinfo{person}{Eric Koskinen}.}
  \bibinfo{year}{2021}\natexlab{}.
\newblock \showarticletitle{Constraint-Based Relational Verification}. In
  \bibinfo{booktitle}{\emph{Computer Aided Verification - 33rd International
  Conference, {CAV} 2021, Virtual Event, July 20-23, 2021, Proceedings, Part
  {I}}} \emph{(\bibinfo{series}{Lecture Notes in Computer Science},
  Vol.~\bibinfo{volume}{12759})}, \bibfield{editor}{\bibinfo{person}{Alexandra
  Silva} {and} \bibinfo{person}{K.~Rustan~M. Leino}} (Eds.).
  \bibinfo{publisher}{Springer}, \bibinfo{pages}{742--766}.
\newblock
\urldef\tempurl%
\url{https://doi.org/10.1007/978-3-030-81685-8\_35}
\showDOI{\tempurl}


\bibitem[Unno et~al\mbox{.}(2017)]%
        {DBLP:conf/cav/UnnoTS17}
\bibfield{author}{\bibinfo{person}{Hiroshi Unno}, \bibinfo{person}{Sho Torii},
  {and} \bibinfo{person}{Hiroki Sakamoto}.} \bibinfo{year}{2017}\natexlab{}.
\newblock \showarticletitle{Automating Induction for Solving {Horn} Clauses}.
  In \bibinfo{booktitle}{\emph{Computer Aided Verification - 29th International
  Conference, {CAV} 2017, Heidelberg, Germany, July 24-28, 2017, Proceedings,
  Part {II}}} \emph{(\bibinfo{series}{Lecture Notes in Computer Science},
  Vol.~\bibinfo{volume}{10427})}, \bibfield{editor}{\bibinfo{person}{Rupak
  Majumdar} {and} \bibinfo{person}{Viktor Kuncak}} (Eds.).
  \bibinfo{publisher}{Springer}, \bibinfo{pages}{571--591}.
\newblock
\urldef\tempurl%
\url{https://doi.org/10.1007/978-3-319-63390-9\_30}
\showDOI{\tempurl}


\bibitem[Yokomori(2003)]%
        {YOKOMORI2003179}
\bibfield{author}{\bibinfo{person}{Takashi Yokomori}.}
  \bibinfo{year}{2003}\natexlab{}.
\newblock \showarticletitle{Polynomial-time identification of very simple
  grammars from positive data}.
\newblock \bibinfo{journal}{\emph{Theoretical Computer Science}}
  \bibinfo{volume}{298}, \bibinfo{number}{1} (\bibinfo{year}{2003}),
  \bibinfo{pages}{179--206}.
\newblock
\showISSN{0304-3975}
\urldef\tempurl%
\url{https://doi.org/10.1016/S0304-3975(02)00423-1}
\showDOI{\tempurl}
\newblock
\shownote{Selected Papers in honour of Setsuo Arikawa}.


\bibitem[Yokomori and Kobayashi(1998)]%
        {LocallyTestableLanguage}
\bibfield{author}{\bibinfo{person}{T. Yokomori} {and} \bibinfo{person}{S.
  Kobayashi}.} \bibinfo{year}{1998}\natexlab{}.
\newblock \showarticletitle{Learning local languages and their application to
  DNA sequence analysis}.
\newblock \bibinfo{journal}{\emph{IEEE Transactions on Pattern Analysis and
  Machine Intelligence}} \bibinfo{volume}{20}, \bibinfo{number}{10}
  (\bibinfo{year}{1998}), \bibinfo{pages}{1067--1079}.
\newblock
\urldef\tempurl%
\url{https://doi.org/10.1109/34.722617}
\showDOI{\tempurl}


\bibitem[Yoshinaka(2006)]%
        {10.1007/11872436_5}
\bibfield{author}{\bibinfo{person}{Ryo Yoshinaka}.}
  \bibinfo{year}{2006}\natexlab{}.
\newblock \showarticletitle{Polynomial-Time Identification of an Extension of
  Very Simple Grammars from Positive Data}. In
  \bibinfo{booktitle}{\emph{Grammatical Inference: Algorithms and
  Applications}}, \bibfield{editor}{\bibinfo{person}{Yasubumi Sakakibara},
  \bibinfo{person}{Satoshi Kobayashi}, \bibinfo{person}{Kengo Sato},
  \bibinfo{person}{Tetsuro Nishino}, {and} \bibinfo{person}{Etsuji Tomita}}
  (Eds.). \bibinfo{publisher}{Springer Berlin Heidelberg},
  \bibinfo{address}{Berlin, Heidelberg}, \bibinfo{pages}{45--58}.
\newblock
\showISBNx{978-3-540-45265-2}


\bibitem[Yoshinaka(2008)]%
        {10.1007/978-3-540-88009-7_21}
\bibfield{author}{\bibinfo{person}{Ryo Yoshinaka}.}
  \bibinfo{year}{2008}\natexlab{}.
\newblock \showarticletitle{Identification in the Limit of k,l-Substitutable
  Context-Free Languages}. In \bibinfo{booktitle}{\emph{Grammatical Inference:
  Algorithms and Applications}}, \bibfield{editor}{\bibinfo{person}{Alexander
  Clark}, \bibinfo{person}{Fran{\c{c}}ois Coste}, {and}
  \bibinfo{person}{Laurent Miclet}} (Eds.). \bibinfo{publisher}{Springer Berlin
  Heidelberg}, \bibinfo{address}{Berlin, Heidelberg},
  \bibinfo{pages}{266--279}.
\newblock
\showISBNx{978-3-540-88009-7}


\bibitem[Yoshinaka(2011)]%
        {YOSHINAKA20111821}
\bibfield{author}{\bibinfo{person}{Ryo Yoshinaka}.}
  \bibinfo{year}{2011}\natexlab{}.
\newblock \showarticletitle{Efficient learning of multiple context-free
  languages with multidimensional substitutability from positive data}.
\newblock \bibinfo{journal}{\emph{Theoretical Computer Science}}
  \bibinfo{volume}{412}, \bibinfo{number}{19} (\bibinfo{year}{2011}),
  \bibinfo{pages}{1821--1831}.
\newblock
\showISSN{0304-3975}
\urldef\tempurl%
\url{https://doi.org/10.1016/j.tcs.2010.12.058}
\showDOI{\tempurl}
\newblock
\shownote{Algorithmic Learning Theory (ALT 2009)}.


\bibitem[Zhu et~al\mbox{.}(2018)]%
        {DBLP:conf/pldi/ZhuMJ18}
\bibfield{author}{\bibinfo{person}{He Zhu}, \bibinfo{person}{Stephen Magill},
  {and} \bibinfo{person}{Suresh Jagannathan}.} \bibinfo{year}{2018}\natexlab{}.
\newblock \showarticletitle{A data-driven {CHC} solver}. In
  \bibinfo{booktitle}{\emph{Proceedings of the 39th {ACM} {SIGPLAN} Conference
  on Programming Language Design and Implementation, {PLDI} 2018}}.
  \bibinfo{publisher}{{ACM}}, \bibinfo{pages}{707--721}.
\newblock
\urldef\tempurl%
\url{https://doi.org/10.1145/3192366.3192416}
\showDOI{\tempurl}


\bibitem[Zhu et~al\mbox{.}(2015)]%
        {zhu_2015}
\bibfield{author}{\bibinfo{person}{He Zhu}, \bibinfo{person}{Aditya~V. Nori},
  {and} \bibinfo{person}{Suresh Jagannathan}.} \bibinfo{year}{2015}\natexlab{}.
\newblock \showarticletitle{Learning refinement types}. In
  \bibinfo{booktitle}{\emph{Proceedings of the 20th {ACM} {SIGPLAN}
  International Conference on Functional Programming, {ICFP} 2015, Vancouver,
  BC, Canada, September 1-3, 2015}}. \bibinfo{publisher}{{ACM}},
  \bibinfo{pages}{400--411}.
\newblock
\urldef\tempurl%
\url{https://doi.org/10.1145/2784731.2784766}
\showDOI{\tempurl}


\bibitem[Zhu et~al\mbox{.}(2016)]%
        {Zhu16PLDI}
\bibfield{author}{\bibinfo{person}{He Zhu}, \bibinfo{person}{Gustavo Petri},
  {and} \bibinfo{person}{Suresh Jagannathan}.} \bibinfo{year}{2016}\natexlab{}.
\newblock \showarticletitle{Automatically learning shape specifications}. In
  \bibinfo{booktitle}{\emph{Proceedings of the 37th {ACM} {SIGPLAN} Conference
  on Programming Language Design and Implementation, {PLDI} 2016, Santa
  Barbara, CA, USA, June 13-17, 2016}},
  \bibfield{editor}{\bibinfo{person}{Chandra Krintz} {and}
  \bibinfo{person}{Emery~D. Berger}} (Eds.). \bibinfo{publisher}{{ACM}},
  \bibinfo{pages}{491--507}.
\newblock
\urldef\tempurl%
\url{https://doi.org/10.1145/2908080.2908125}
\showDOI{\tempurl}


\end{thebibliography}

\newpage
\appendix
\section*{Appendix}
\section{Proofs}
\label{sec:proofs}
\subsection{Proofs and Additional Definitions for Sections~\ref{sec:tpinf} and \ref{sec:stp}}
This section provides proofs omitted in Sections~\ref{sec:tpinf} and \ref{sec:stp}.

We first show basic properties in Sections~\ref{sec:properties-red}--\ref{sec:red-vs-pred}.
The dependency of each main property on basic properties is summarized in Table~\ref{tab:lemmas},
so that it is easy to discuss properties of the extensions and variations of STPs considered in Section~\ref{sec:ext}.

\begin{table}[bh]
  \caption{Required Basic Properties for Each Main Property}
  \label{tab:lemmas}
  \begin{center}
  \begin{tabular}{|l|l|}
  \hline
  Properties & Required Properties\\
  \hline
  \hline
  Soundness (Theorem~\ref{th:soundness})& Lemma~\ref{lem:red-invariant}\\
  \hline
  Soundness 2 (Theorem~\ref{th:soundness2}) &
  Lemmas~\ref{lem:rewriting-decomp}, \ref{lem:pred-closed-under-subst}, and \ref{lem:red-implies-pred}\\
  \hline
  Completeness (Theorems~\ref{th:completeness} an \ref{th:completeness-ctpinf}) & Lemmas~\ref{lem:rewriting-decomp} and \ref{lem:completeness} and Theorem~\ref{th:soundness}\\
  \hline
    Minimality (Theorem~\ref{th:minimality}) & Lemmas~\ref{lem:rewriting-decomp}, \ref{lem:pred-closed-under-subst}, \ref{lem:minimality-sim}, \ref{lem:pred-wconf}, \ref{lem:pred-decreases-measure}, Lemma~\ref{lem:pred-tpequiv},  and \ref{lem:red-implies-pred}\\
  \hline
  Characteristic Data (Theorem~\ref{th:data-size}) &
  \begin{minipage}{8cm}
    \ \\[-.5ex]
  Lemmas~\ref{lem:canonical} and \ref{lem:canonicalsubst},\\ in addition to all the lemmas above
    \ \\[-2ex]
  \end{minipage}\\
  \hline
  \begin{minipage}{5cm}
    \ \\[-.5ex]
  Learnability of STPs\\ (Theorems~\ref{th:learnability} and \ref{th:stp-lfp})
    \ \\[-2ex]
    \end{minipage}
  &
  \begin{minipage}{8cm}
    \ \\[-.5ex]
  Soundness and the existence of characteristic data\\
  (Theorems~\ref{th:soundness} and ~\ref{th:data-size})
    \ \\[-2ex]
  \end{minipage}\\
  \hline
  \begin{minipage}{5cm}
    \ \\[-.5ex]
  Learnability of CSTPs\\ (Theorems~\ref{th:learnability} and \ref{th:cstp-lfp})
    \ \\[-2ex]
  \end{minipage}
  & Soundness, completeness and
  Lemma~\ref{lem:data-increase}\\
  \hline
  \begin{minipage}{5cm}
    \ \\[-.5ex]
    Complexity of Decision Problems (Theorem~\ref{th:decision-problems})
    \ \\[-2ex]
    \end{minipage}
    & Lemmas~\ref{lem:minimality-sim}, \ref{lem:pred-wconf}, and  \ref{lem:pred-decreases-measure}\\
  \hline
  \begin{minipage}{5cm}
    \ \\[-.5ex]
    Decidability of Satisfiability\\ (Theorem~\ref{th:satisfiability})
    \ \\[-2ex]
    \end{minipage}
    & Lemma~\ref{lem:minimality-sim} and decidability of word equations\\
  \hline
\end{tabular}
  \end{center}
\end{table}

\subsubsection{Properties of \(\red\)}
\label{sec:properties-red}
\begin{lemma}
  \label{lem:rewriting-decomp}
  Let \(\seq{x}\) be mutually distinct variables.
  If \((\tp,\Subst{\seq{x}}{\dat})\red^n (\tp',\Dat')\nred\), then there exists
  \(\tp''\) such that 
  \((\seq{x},\Subst{\seq{x}}{\dat})\red^n (\tp'',\Dat')\nred\)
  and \(\tp' = [\tp''/\seq{x}]\tp\).

  Conversely, if \((\seq{x},\Subst{\seq{x}}{\dat})\red^n (\tp,\Dat')\),
  then \((\tp_0, \Subst{\seq{x}}{\dat})\red^n ([\tp/\seq{x}]\tp_0,\Dat')\) for any tuple pattern \(\tp_0\).
\end{lemma}
\begin{proof}
  This follows by a straightforward (simultaneous) induction on the length \(n\)
  of the rewriting sequence, with case analysis used in the first step.
  Note that in the rules for \((\tp,\Dat)\red (\tp',\Dat')\), there is no condition on \(\tp\)
  in the premise.
\end{proof}
\begin{lemma}
  \label{lem:red-invariant}
  If \((\tp_1,\Dat_1)\red (\tp_2,\Dat_2)\), then \(\Dat_1\tp_1=\Dat_2\tp_2\).
\end{lemma}
\begin{proof}
  This follows by a straightforward case analysis on the rule used for deriving
  \((\tp_1,\Dat_1)\red (\tp_2,\Dat_2)\).
\end{proof}

\subsubsection{Properties of \(\pred\)}
\label{sec:properties-pred}

\begin{lemma}
  \label{lem:pred-closed-under-subst}
  If \(\tp\pred \tp'\), then \([\tp_1/\seq{x}]\tp \preds [\tp_1/\seq{x}]\tp'\).
\end{lemma}
\begin{proof}
  Trivial by the definition of \(\pred\).
  We have actually \([\tp_1/\seq{x}]\tp \pred [\tp_1/\seq{x}]\tp'\) or
  \([\tp_1/\seq{x}]\tp = [\tp_1/\seq{x}]\tp'\), and the latter occurs only
  when empty sequences are substituted for some variables.
  For example, when \(\tp=(x_2x_1,x_2,\ldots)\pred (x_1,x_2,\ldots)=\tp'\)
  and \(\tp_1=(p_1,\epsilon,\ldots)\),
 \( [\tp_1/\seq{x}]\tp = (p_1,\epsilon,\ldots)=[\tp_1/\seq{x}]\tp'\). 
 \end{proof}

When \(\tp_1\pred \tp_2\) holds, we define a partial function \(\residualfun{\tp_1}{\tp_2}\) on tuple patterns
as follows. 
\begin{definition}[\(\residualfun{\tp_1}{\tp_2}\)]
  \label{df:residual-pattern}
  Let \(\tp_0, \tp_1\) and \(\tp_2\) be tuple patterns such that \(\tp_1\pred \tp_2\).
  We define
  the tuple pattern \(\residual{\tp_0}{\tp_1}{\tp_2}\) such that
  \(\tp_0\preds \residual{\tp_0}{\tp_1}{\tp_2}\) as follows,
  by case analysis on the rule used for deriving \(\tp_1\pred \tp_2\):
\begin{itemize}
\item Case  \rn{PR-Prefix}:
  In this case, \(\tp_1=(\pat_1,\ldots,\pat_n)\) and \(\tp_2=(\pat'_1,\ldots,\pat'_n)\)
  with \(\pat_j=\pat_i\cdot \pat'_j\) and \(\pat_i\ne \epsilon\)
  for some \(i, j\) and \(\pat'_k=\pat_k\) for \(k\ne j\).
  If \(\tp_0\) is of the form \((q_1,\ldots,q_n)\) and \(q_i\) is a prefix of \(q_j\),
  then \(\residual{\tp_0}{\tp_1}{\tp_2}:= (q'_1,\ldots,q'_n)\) where
  \(q_j' = \ldiff{q_i}{q_j}\) and \(q_k'=q_k\) for \(k\ne j\).
  (Recall that \(\ldiff{q_i}{q_j}\) denotes \(q\) such that \(q_i\cdot q=q_j\).  
  We have \(\tp_0\pred \residual{\tp_0}{\tp_1}{\tp_2}\) if \(q_2\ne \epsilon\),
  and \(\tp_0=\residual{\tp_0}{\tp_1}{\tp_2}\) if \(q_2=\epsilon\).)
Otherwise, \(\residual{\tp_0}{\tp_1}{\tp_2}\) is undefined.
\item Case  \rn{PR-Suffix}:
  In this case, \(\tp_1=(\pat_1,\ldots,\pat_n)\) and \(\tp_2=(\pat'_1,\ldots,\pat'_n)\)
  with \(\pat_j=\pat'_j\cdot \pat_i\) and \(\pat_i\ne \epsilon\)
  for some \(i, j\) and \(\pat'_k=\pat_k\) for \(k\ne j\).
  If \(\tp_0\) is of the form \((q_1,\ldots,q_n)\) and \(q_i\) is a suffix of \(q_j\),
  then \(\residual{\tp_0}{\tp_1}{\tp_2}:= (q'_1,\ldots,q'_n)\) where
  \(q_j' = \rdiff{q_i}{q_j}\) and \(q_k'=q_k\) for \(k\ne j\).
Otherwise, \(\residual{\tp_0}{\tp_1}{\tp_2}\) is undefined.
\item Case  \rn{PR-CPrefix}:
  In this case, \(\tp_1=(\pat_1,\ldots,\pat_n)\) and \(\tp_2=(\pat'_1,\ldots,\pat'_n)\)
  with \(\pat_j= a\pat'_j\) for some \(j\) and \(\pat'_k=\pat_k\) for \(k\ne j\).
  If \(\tp_0\) is of the form \((q_1,\ldots,q_n)\) and \(a\) is a prefix of \(q_j\),
  then \(\residual{\tp_0}{\tp_1}{\tp_2}:= (q'_1,\ldots,q'_n)\) where
  \(q_j' = \ldiff{a}{q_j}\) and \(q_k'=q_k\) for \(k\ne j\).
  Otherwise,
\(\residual{\tp_0}{\tp_1}{\tp_2}\) is undefined.
\item Case  \rn{PR-CSuffix}:
  In this case, \(\tp_1=(\pat_1,\ldots,\pat_n)\) and \(\tp_2=(\pat'_1,\ldots,\pat'_n)\)
  with \(\pat_j= \pat'_j a\) for some \(j\) and \(\pat'_k=\pat_k\) for \(k\ne j\).
  If \(\tp_0\) is of the form \((q_1,\ldots,q_n)\) and \(a\) is a suffix of \(q_j\),
  then \(\residual{\tp_0}{\tp_1}{\tp_2}:= (q'_1,\ldots,q'_n)\) where
  \(q_j' = \rdiff{a}{q_j}\) and \(q_k'=q_k\) for \(k\ne j\).
  Otherwise,
\(\residual{\tp_0}{\tp_1}{\tp_2}\) is undefined.
\item Case  \rn{PR-Epsilon}:
  In this case, \(\tp_1=(\pat_1,\ldots,\pat_{j-1},\epsilon,\pat_{j+1},\ldots,\pat_n)\)
  and \(\tp_2=(\pat_1,\ldots,\pat_{j-1},\pat_{j+1},\ldots,\pat_n)\).
  If \(\tp_0\) is of the form \((\patalt_1,\ldots,\patalt_{j-1},\epsilon,\patalt_{j+1},\ldots,\patalt_n)\), then
  \(\residual{\tp_0}{\tp_1}{\tp_2}:=(\patalt_1,\ldots,\patalt_{j-1},\patalt_{j+1},\ldots,\patalt_n)\).
  Otherwise,
\(\residual{\tp_0}{\tp_1}{\tp_2}\) is undefined.
\end{itemize}
\end{definition}
\noindent
Note that \(\residualfun{\tp_1}{\tp_2}\) above depends on which rule has been used
for deriving \({\tp_1}\pred{\tp_2}\); for example, \(\residualfun{(ax,a)}{(x,a)}(ap_2p_1,ap_2)=(p_1,ap_2)\) if
\((ax,a)\pred(x,a)\) has been derived by \rn{R-Prefix} and \(\residualfun{(ax,a)}{(x,a)}(ap_2p_1,ap_2)=(p_2p_1,ap_2)\) if
\((ax,a)\pred(x,a)\) has been derived by \rn{R-CPrefix}.
Thus, when we write \(\residualfun{\tp_1}{\tp_2}\), we assume that
\(\tp_1\pred \tp_2\) carries information about the rule used for deriving \(\tp_1\pred \tp_2\).

We defined  \(\residual{\tp_0}{\tp_1}{\tp_2}\) above 
so that the reduction \(\tp_0\preds \residual{\tp_0}{\tp_1}{\tp_2}\)
``simulates'' the reduction \(\tp_1\pred \tp_2\) in the following sense.
\begin{lemma}
  \label{lem:minimality-sim}
  \begin{enumerate}
 \item  If \(\tp_1\pred \tp_1'\) and \(\lang(\tp_0)\subseteq \lang([\tp/\seq{x}]\tp_1)\), then
  \(\residual{\tp_0}{\tp_1}{\tp_1'}\) is well defined and \(\tp_0\preds \residual{\tp_0}{\tp_1}{\tp_1'}\).
\item If \(\residual{\tp_0}{\tp_1}{\tp_1'}\) is well defined,
  then \(\lang(\tp_0)\subseteq \lang([\tp/\seq{x}]\tp_1)\) if and only if
  \(\lang(\residual{\tp_0}{\tp_1}{\tp_1'})\subseteq \lang([\tp/\seq{x}]\tp_1')\);
  and \(\lang(\tp_0)\supseteq \lang([\tp/\seq{x}]\tp_1)\) if and only if
  \(\lang(\residual{\tp_0}{\tp_1}{\tp_1'})\supseteq \lang([\tp/\seq{x}]\tp_1')\).
\end{enumerate}
\end{lemma}
\begin{proof}
  Case analysis on the rule used for deriving \(\tp_1\pred \tp_1'\).
  Since the other cases are similar or simpler, we discuss only the case for
  \rn{PR-Prefix}.
We may assume, without loss of generality,
    \(\tp_1=(p_2p_1,p_2,\ldots,p_n)\) and \(\tp_1'=(p_1,p_2,\ldots,p_n)\) with
    \(p_2\ne \epsilon\).
    To show (1), assume \(\lang(\tp_0)\subseteq \lang([\tp/\seq{x}]\tp_1)\).
    Then for any \((s_1,s_2,\ldots,s_n)\in\lang(\tp_0)\),
     \(s_2\) must be a prefix of \(s_1\). Thus
     \(\tp_0\) must also be of the form \((q_2q_1,q_2,\ldots,q_n)\).
     (To see this, suppose \(\tp_0=(q_1',q_2,\ldots)\) and \(q_2\) is not a prefix of \(q_1'\).
     Then either (i) \(q_1'q_1''=q_2\) with \(q_1''\ne \epsilon\) (i.e., \(q_1'\) is a strict prefix of \(q_2\)) or (ii)
     \(q_1' = q \alpha_1 q_1''\) and \(q_2 = q \alpha_2 q_2''\) with \(\alpha_1,\alpha_2\in\Alpha\cup \Vars\)
     and \(\alpha_1\ne \alpha_2\) (i.e., \(\alpha_1,\alpha_2\) are the first letters or variables that differ between
     \(q_1'\) and \(q_2\)). In case (i), let \(\theta\) be a substitution such that \(\theta q_1''\ne \epsilon\),
     and in case (ii), let \(\theta\) be a substitution such that
     \(\theta \alpha_1\ne \theta \alpha_2\). 
     Then \(s_2=\theta  q_2\) is not a prefix of \(s_1=\theta q_1'\).  Thus, \(q_2\) must be a prefix of \(q_1'\). )
     Therefore,  \(\residual{\tp_0}{\tp_1}{\tp_1'}\) is well defined, and \(\tp_0\preds \residual{\tp_0}{\tp_1}{\tp_1'}\).

     To show (2), suppose  \(\residual{\tp_0}{\tp_1}{\tp_1'}\) is well defined.
     Then, \(\tp_0=(q_2q_1,q_2,\ldots,q_n)\) and
     \(\tp_0' := \residual{\tp_0}{\tp_1}{\tp_1'}=(q_1,q_2,\ldots,q_n)\) where
     \(q_2\) may be \(\epsilon\).
     Thus, we have:
    \begin{align*}
      \lang(\tp_0')&=\set{(\leftinv{w_2}{w_1},w_2,\ldots,w_n)\mid (w_1,\ldots,w_n)\in\lang(\tp_0)}\\
      \lang(\tp_0)&=\set{(w_2w_1,w_2,\ldots,w_n)\mid (w_1,\ldots,w_n)\in\lang(\tp'_0)}\\
      \lang([\tp/\seq{x}]\tp_1')&=\set{(\leftinv{w_2}{w_1},w_2,\ldots,w_n)\mid (w_1,\ldots,w_n)\in\lang([\tp/\seq{x}]\tp_1)}\\
      \lang([\tp/\seq{x}]\tp_1)&=\set{(w_2w_1,w_2,\ldots,w_n)\mid (w_1,\ldots,w_n)\in\lang([\tp/\seq{x}]\tp'_1)},
    \end{align*}
    from which the required properties follow immediately.
\end{proof}

Next, we will prove the weak confluence of \(\pred\) (up to permutations).
To that end, we prepare the following lemma.
 \begin{lemma}
   \label{lem:conjugacy}
   Suppose \(pq_1=q_2p\) for \(p,q_1,q_2\in (\Alpha\cup \Vars)^*\).
   Then, there exists \(\tp\) such that \((q_1,p)\preds \tp \rpreds (q_2,p)\).
 \end{lemma}
 \begin{proof}
   This follows by induction on \(|p|\).
   If \(p=\epsilon\), then \(q_1=q_2\). Therefore, the property holds for \(\tp=(q_1,p)\).
   Otherwise,
   by the assumption \(pq_1=q_2p\), either \(p\) is a prefix of \(q_2\), or \(q_2\) is a prefix of \(p\).
   
   In the former case, \(q_2=pr\) with \(q_1=rp\).
   Thus, we have \((q_1,p)\pred \tp \rpreds (q_2,p)\) for \(\tp = (r,p)\).

   In the latter case, \(p=q_2 p'=p'q_1\).
   We may assume that \(q_2\ne \epsilon\), since \((q_1,p)=(\epsilon,p)= (q_2,p)\) otherwise.
   Thus, we have \(|p'|<|p|\).
   By the induction hypothesis, there exists \(\tp\) such that \((q_1,p')\preds \tp \rpreds (q_2,p')\).
   Thus, we have
   \((q_1,p)\pred (q_1,p')\preds \tp \rpreds (q_2,p')\rpred (q_2,p)\) as required.
 \end{proof}
 \begin{remark}
   The lemma above also follows immediately from the following property of words (cf. \cite{Lothaire}, Proposition~12.1.2):
   ``If \(p,q_1,q_2\in \Alpha^*\), \(pq_1=q_2p\) and \(q_1\ne \epsilon\), then there exists \(r,s\in \Alpha^*\) and an integer \(n\) such that
   \(p=(rs)^nr, q_1=sr\), and \(q_2=rs\).''
   By using this property, we obtain  \((q_1,p)\preds (sr, r) \pred (s,r) \rpred (rs, r) \rpreds (q_2,p)\). \qed
\end{remark}

We write \(\tp \tpequiv \tp'\) if \(\tp'\) is a permutation of \(\tp\). For example,
\((x, xy)\tpequiv (xy, x)\).
The following lemma states that \(\pred\) is weakly confluent up to \(\tpequiv\).
\begin{lemma}[weak confluence of $\pred$ up to $\tpequiv$]
  \label{lem:pred-wconf}
If \(\tp\pred \tp_1\) and \(\tp\pred \tp_2\), then there exist \(\tp_1'\)
  and \(\tp_2'\) such that
  \(\tp_1\preds \tp_1' \tpequiv \tp_2' \rpreds \tp_2\).
\end{lemma}
\begin{proof}
  It is sufficient to consider the cases where \(\tp_1\ne \tp_2\) and
  the parts involved in the two reductions overlap.
  Without loss of generality, we may assume that the first element \(p_1\)
  is principal in \(\tp\pred \tp_1\).
  We thus need to consider only the following cases; due to the symmetry between
  prefix and suffix, the cases involving only suffixes are omitted.
  \begin{itemize}
  \item Case \(\tp=(p_2p'_1,p_2,p_3,\ldots,p_n)\pred (p'_1,p_2,p_3,\ldots,p_n)=\tp_1\)
    and \(\tp=(p_3p''_1,p_2,p_3,\ldots,p_n)\pred (p''_1,p_2,p_3,\ldots,p_n)=\tp_2\)
    (i.e., both reductions use \rn{PR-Prefix}, and \(p_1\) is principal in both
    reductions).\\
    By \(p_2p'_1=p_3p''_1\), \(p_2\) and \(p_3\) are in the prefix relation.
    By symmetry, we may assume that \(p_2=p_3p_2'\), with \(p''_1=p_2'p'_1\).
    Then we have:
    \begin{align*}
      & \tp_1 = (p'_1,p_3p_2',p_3,\ldots,p_n) \pred (p'_1,p_2',p_3,\ldots,p_n)\\
      & \tp_2=(p_2'p'_1,p_3p_2',p_3,\ldots,p_n)\pred
      (p_2'p'_1,p_2',p_3,\ldots,p_n)\preds(p'_1,p_2',p_3,\ldots,p_n),
    \end{align*}
    as required. In the second reduction step for \(\tp_2\),
    \((p_2'p'_1,p_2',p_3,\ldots,p_n)=(p'_1,p_2',p_3,\ldots,p_n)\) if \(p_2'=\epsilon\)
    and \((p_2'p'_1,p_2',p_3,\ldots,p_n)\pred (p'_1,p_2',p_3,\ldots,p_n)\) otherwise.
  \item Case \(\tp=(p_2p'_1,p_2,p_3,\ldots,p_n)\pred (p'_1,p_2,p_3,\ldots,p_n)=\tp_1\)
    and \(\tp=(p_1,p_2,p_1p'_3,\ldots,p_n)\pred (p_1,p_2,p'_3,\ldots,p_n)=\tp_2\)
    (i.e., both reductions use \rn{PR-Prefix}, and \(p_1\) is auxiliary in
    \(\tp\pred \tp_2\)).\\
    In this case, we have \(p_1=p_2p_1'\) and \(p_3=p_1p'_3=p_2p_1'p_3'\).
    Then we have:
    \begin{align*}
      & \tp_1 = (p'_1,p_2,p_2p_1'p_3',\ldots,p_n) \pred
      (p'_1,p_2,p_1'p_3',\ldots,p_n) \preds (p'_1,p_2,p_3',\ldots,p_n) \\
      & \tp_2=(p_2p_1',p_2,p'_3,\ldots,p_n)\pred (p'_1,p_2,p_3',\ldots,p_n),
    \end{align*}
    as required.
    In the second reduction step for \(\tp_1\),
   \( (p'_1,p_2,p_1'p_3',\ldots,p_n) = (p'_1,p_2,p_3',\ldots,p_n)\)  if \(p_1'=\epsilon\)
   and
\( (p'_1,p_2,p_1'p_3',\ldots,p_n) \pred (p'_1,p_2,p_3',\ldots,p_n)\)  otherwise.
  \item Case \(\tp=(p_2p'_1,p_2,\ldots,p_n)\pred (p'_1,p_2,\ldots,p_n)=\tp_1\)
and \(\tp=(ap''_1,p_2,\ldots,p_n)\pred (p''_1,p_2,\ldots,p_n)=\tp_2\)
    (i.e., \(\tp\pred\tp_1\) uses \rn{PR-Prefix}, while
    \(\tp\pred \tp_2\) uses \rn{PR-CPrefix}, and \(p_1\) is principal in both
    reductions).\\
    In this case, \(p_2=ap_2'\) with \(p''_1=p_2'p'_1\).
    Thus, we have:
    \begin{align*}
      & \tp_1 = (p'_1,ap_2',\ldots,p_n) \pred (p'_1,p_2',\ldots,p_n) \\
      & \tp_2=(p'_2p_1',ap_2',\ldots,p_n)\pred (p'_2p'_1,p'_2,\ldots,p_n)\preds
      (p'_1,p'_2,\ldots,p_n),
    \end{align*}
    as required. In the second step for \(\tp_2\),
    \((p'_2p'_1,p'_2,\ldots,p_n)=(p'_1,p'_2,\ldots,p_n)\) if \(p_2'=\epsilon\)
    and \((p'_2p'_1,p'_2,\ldots,p_n)\pred(p'_1,p'_2,\ldots,p_n)\) otherwise.
  \item Case \(\tp=(p_2p'_1,p_2,\ldots,p_n)\pred (p'_1,p_2,\ldots,p_n)=\tp_1\)
and \(\tp=(p_1,ap'_2,\ldots,p_n)\pred (p_1,p'_2,\ldots,p_n)=\tp_2\)
    (i.e., \(\tp\pred\tp_1\) uses \rn{PR-Prefix}, while
    \(\tp\pred \tp_2\) uses \rn{PR-CPrefix}, and the auxiliary element \(p_2\) in
    the former reduction is principal in the latter).\\
    In this case, \(p_2=ap_2'\).
    Thus, we have:
    \begin{align*}
      & \tp_1 = (p'_1,ap_2',\ldots,p_n) \pred (p'_1,p_2',\ldots,p_n)\\
      & \tp_2=(ap_2'p'_1,p_2',\ldots,p_n)\pred (p'_2p'_1,p'_2,\ldots,p_n)\preds
      (p'_1,p'_2,\ldots,p_n),
    \end{align*}
    as required.
    In the second step for \(\tp_2\),
    \((p'_2p'_1,p'_2,\ldots,p_n)=(p'_1,p'_2,\ldots,p_n)\) if \(p_2'=\epsilon\)
    and \((p'_2p'_1,p'_2,\ldots,p_n)\pred(p'_1,p'_2,\ldots,p_n)\) otherwise.
   \item Case where \(t=(p_2p_1',p_2,\ldots,p_n)\pred (p_1',p_2,\ldots,p_n)=t_1\) and
     \(t=(p_1''p_2,p_2,\ldots,p_n)\pred (p_1'',p_2,\ldots,p_n)=t_2\)
     (i.e., \(p_1\) is principal and \(p_2\) is auxiliary in both reductions, and \(t\pred t_1\) uses \rn{PR-Prefix}
     while \(t\pred t_2\) uses \rn{PR-Suffix}).
     In this case, \(p_2p_1'=p_1''p_2\).
    Thus, by Lemma~\ref{lem:conjugacy}, there exists \(\tp_3\) such that \(\tp_1\preds \tp_3\rpreds \tp_2\).
\item Case where \(t=(p_2p_1',p_2,\ldots,p_n)\pred (p_1',p_2,\ldots,p_n)=t_1\) and
     \(t=(p_1''p_3,p_2,p_3,\ldots,p_n)\pred (p_1'',p_2,p_3,\ldots,p_n)=t_2\)
     (i.e., \(p_1\) is principal in both reductions, and \(p_2\) is auxiliary in the former reduction while
     \(p_3\) is auxiliary in the latter;   \(t\pred t_1\) uses \rn{PR-Prefix}, \(t\pred t_2\) uses \rn{PR-Suffix}).
     In this case, \(p_2p_1'=p_1''p_3\); so, either (i) \(p_1''=p_2p_1'''\) and \(p_1'=p_1'''p_3\),
     or (ii) \(p_2=p_1''p_2'\) and \(p_3=p_2'p_1'\).
     In case (i), we have:
     \begin{align*}
       & t_1=(p_1'''p_3,p_2,p_3,\ldots,p_n)\pred (p_1''',p_2,p_3,\ldots,p_n)\\
       & t_2=(p_2p_1''',p_2,p_3,\ldots,p_n)\pred (p_1''',p_2,p_3,\ldots,p_n).
     \end{align*}
     In case (ii), we have:
     \begin{align*}
       & t_1=(p_1',p_1''p_2',p_2'p_1',\ldots,p_n)\pred
       (p_1',p_1''p_2',p_2',\ldots,p_n)\pred (p_1',p_1'',p_2',\ldots,p_n)\\
       & t_2=(p_1'',p_1''p_2',p_2'p_1',\ldots,p_n)\pred 
       (p_1'',p_2',p_2'p_1',\ldots,p_n)\pred (p_1'',p_2',p_1',\ldots,p_n)\\
      & (p_1',p_1'',p_2',\ldots,p_n)\tpequiv (p_1'',p_2',p_1',\ldots,p_n),
     \end{align*}
     as required.
     Here, notice that reductions are confluent only up to the permutation relation \(\tpequiv\).
   \item Case where
     \(t=(p_2p_1',p_2,\ldots,p_n)\pred (p_1',p_2,\ldots,p_n)=t_1\) and
     \(t=(p_1,p_2'p_3,p_3,\ldots,p_n)\pred (p_1,p'_2,p_3,\ldots,p_n)=t_2\)
     (i.e. the auxiliary element \(p_2\) in the former reduction is principal in the latter, and
     \(t\pred t_1\) uses \rn{PR-Prefix} while \(t\pred t_2\) uses \rn{PR-Suffix}).
     In this case, we have \(p_1=p_2p_1'\) and \(p_2=p_2'p_3\).
     Then, we have:
     \begin{align*}
       & t_1=(p_1',p_2'p_3,p_3,\ldots,p_n)\pred(p_1',p_2',p_3,\ldots,p_n)\\
       & t_2=(p_2'p_3p'_1,p'_2,p_3,\ldots,p_n)\pred
       (p_3p'_1,p'_2,p_3,\ldots,p_n)\pred (p'_1,p'_2,p_3,\ldots,p_n)
     \end{align*}
     as required.
  \end{itemize}
\end{proof}

Recall that \(\measure{(p_1,\ldots,p_k)}=|\pat_1\cdots\pat_k|+k\).
\begin{lemma}
  \label{lem:pred-decreases-measure}
  If \(\tp\pred \tp'\), then
  \(\measure{\tp}>\measure{\tp'}\).
\end{lemma}
\begin{proof}
  This follows immediately by the case analysis on the rule used for deriving
  \(\tp\pred \tp'\). 
\end{proof}  

\begin{lemma}
  \label{lem:pred-tpequiv}
  If \(t_1\tpequiv \pred t_2\), then \(t_1\pred \tpequiv t_2\).
\end{lemma}
\begin{proof}
  This follows immediately from the definition of 
  \(\pred \). 
\end{proof}
\begin{lemma}[Church-Rosser property of $\pred$ up to $\tpequiv$]
  \label{lem:pred-CR}
If \(\tp_{L,1}\rpreds \tp_L \tpequiv \tp_R \preds \tp_{R,1}\),
  then there exist \(\tp_{L,2}\) and \(\tp_{R,2}\) such that
  \(\tp_{L,1}\preds \tp_{L,2} \tpequiv \tp_{R,2}\rpreds \tp_{R,1}\).
\end{lemma}
\begin{proof}
This follows immediately from 
the standard fact that weak confluence (Lemma~\ref{lem:pred-wconf}) and strong normalization (implied by
Lemma~\ref{lem:pred-decreases-measure}) implies the Church-Rosser property.
For completeness, we provide a proof.
  We show the lemma by induction on \(\#\tp_L (= \#\tp_R)\).
  Suppose \(\tp_{L,1}\rpreds \tp_L \tpequiv \tp_R \preds \tp_{R,1}\).
  If \(\tp_L=\tp_{L,1}\) or \(\tp_R=\tp_{R,1}\), the result follows immediately by Lemma~\ref{lem:pred-tpequiv}.
  Otherwise, we have \(\tp_L\pred \tp'_{L,1}\preds \tp_{L,1}\) and \(\tp_R\pred \tp'_{R,1}\preds \tp_{R,1}\).
By Lemma~\ref{lem:pred-wconf}, there exist \(\tp''_{L,1}\) and \(\tp''_{R,1}\) such that
  \(\tp_{L,1}'\preds \tp''_{L,1}\tpequiv \tp_{R,1}''\rpreds \tp_{R,1}'\).
  By the induction hypothesis, we also have:
  \begin{align*}
  &  \tp_{L,1}\preds \tp_{L,3}\tpequiv \tp'_{L,3}\rpreds \tp''_{L,1}\\
    &  \tp_{R,1}\preds \tp_{R,3}\tpequiv \tp'_{R,3}\rpreds \tp''_{R,1}.
  \end{align*}
  By applying the induction hypothesis to
  \(\tp'_{L,3}\rpreds \tp''_{L,1}\tpequiv \tp''_{R,1}\preds \tp'_{R,3}\), we obtain
  \(\tp'_{L,3} \preds \tp'_{L,2} \tpequiv \tp'_{R,2} \rpreds \tp'_{R,3}\).
As \(\tp_{L,1}\preds \tp_{L,3}\tpequiv \tp'_{L,3}\preds \tp'_{L,2}\)
  and
\(\tp_{R,1}\preds \tp_{R,3}\tpequiv \tp'_{R,3}\preds \tp'_{R,2}\),
  by Lemma~\ref{lem:pred-tpequiv}, 
  there exist \(\tp_{L,2}\) and \(\tp_{R,2}\) such that
  \(\tp_{L,1}\preds \tp_{L,2}\tpequiv \tp'_{L,2}\) and
  \(\tp_{R,1}\preds \tp_{R,2}\tpequiv \tp'_{R,2}\).
  Thus, we have 
  \[\tp_{L,1}\preds \tp_{L,2} \tpequiv \tp'_{L,2}\tpequiv \tp'_{R,2}\tpequiv \tp_{R,2}\rpreds \tp_{R,1},\]
as required. 
\end{proof}

\subsubsection{Correspondence between \(\red\) and \(\pred\)}
\label{sec:red-vs-pred}

\begin{lemma}
  \label{lem:red-implies-pred}
  If \((\seq{x},\Subst{\seq{x}}{\dat})\red (\tp,\Subst{\seq{y}}{\dat'})\),
  then \(\tp\pred \seq{y}\). \end{lemma}
\begin{proof}
  This follows by a straightforward case analysis on 
  \((\seq{x},\Subst{\seq{x}}{\dat})\red (\tp,\Subst{\seq{y}}{\dat'})\). 
 \end{proof}
\begin{lemma}
  \label{lem:pred-implies-red}
  \label{lem:completeness}
  Suppose \(\dat \models \tp\).
  If \(\tp \pred\tp'\),
  then \((\seq{x},\Subst{\seq{x}}{\dat})\red (\tp'',\Subst{\seq{y}}{\dat'})\) for some \(\tp''\), \(\seq{y}\), and \(\dat'\).
  Furthermore, if \(\dat\smodels \tp\), then \(\tp=[\tp'/\seq{y}]\tp''\) and \(\dat'\smodels\tp'\).
\end{lemma}
\begin{proof}
   Suppose \(\dat,\Theta \models \tp\) and \(\tp\pred \tp'\).
   We perform case analysis on the rule used for deriving \(\tp \pred\tp'\).
  \begin{itemize}
  \item  Case \rn{PR-Prefix}: 
     We may assume, without loss of generality,
     \(\tp=(p_2p_1,p_2,\ldots,p_n)\) and \(\tp'=(p_1,p_2,\ldots,p_n)\), with \(p_2\ne \epsilon\).
    By \(\dat,\Theta\models \tp\), we have \(\dat[*][1]=\dat[*][2]\cdot \seq{s}\)
where \(\seq{s}=\Theta p_1\).
    If \(\dat[*][2]=\seq{\epsilon}\), then \rn{R-Epsilon} is applicable to 
    \((\seq{x},\Subst{\seq{x}}{\dat})\).
    Otherwise,
    let \(\seq{y}=(x_1',x_2,\ldots,x_n)\) with
    \(\tp_1=(x_2x_1',x_2,\ldots,x_n)\) and
    \(\dat'=(\seq{s},\dat[*][2],\ldots,\dat[*][n])\).
    Then, we have
    \((\seq{x},\Subst{\seq{x}}{\dat})\red (\tp_1,\Subst{\seq{y}}{\dat'})\).
    Furthermore, if \(\dat,\Theta\smodels \tp\), then \(\dat[*][2]\ne\seq{\epsilon}\).
    Therefore, the latter case applies and we have 
    \(\tp=[\tp'/\seq{y}]\tp_1\) and \(\dat',\Theta\models \tp'\) as required.
  \item  Case \rn{PR-CPrefix}: 
    In this case \(\tp=(ap_1,p_2,\ldots,p_n)\) and \(\tp'=(p_1,\ldots,p_n)\).
    By \(\dat,\Theta \models \tp\), we have \(\dat[*][i]=a\cdot \seq{s}\). Thus, \rn{R-CPrefix} is applicable to obtain
    \[(\seq{x},\Subst{\seq{x}}{\dat})\red (\tp'', \Subst{(x'_1,x_2,\ldots,x_n)}{\dat'})\]
    for \(\tp''=(ax'_1,x_2,\ldots,x_n)\), with
     \(\dat'[*][i]=\seq{s}\) and \(\dat'[*][j]=\dat[*][j]\) for \(j\ne i\).
    Thus, we also have \(\tp=[\tp'/\seq{y}]\tp''\), and \(\dat,\Theta\smodels \tp\) implies \(\dat',\Theta\smodels\tp'\).
  \item  Cases for \rn{PR-Suffix} and \rn{PR-CSuffix}: analogous to the cases for \rn{PR-Prefix} and \rn{PR-CPrefix} respectively.
  \item  Case \rn{PR-Epsilon}:
    We may assume, without loss of generality,
    \(\tp=(\epsilon,p_2,\ldots,p_n)\) and \(\tp'=(p_2,\ldots,p_n)\).
    The required result holds for \(\tp''=(\epsilon,x_2,\ldots,x_n)\), \(\dat'=\remove{\dat}{1}\) and \(\seq{y}=(x_2,\ldots,x_n)\).
\end{itemize}
 \end{proof}

\subsubsection{Proof of Soundness (Theorems~\ref{th:soundness} and \ref{th:soundness2})}

Theorem~\ref{th:soundness} is a trivial consequence of Lemma~\ref{lem:red-invariant}.
\begin{proof}[Proof of Theorem~\ref{th:soundness}]
  Suppose
  \((\seq{x},\Subst{\seq{x}}{\dat})\reds (t,\Dat)\).
  By Lemma~\ref{lem:red-invariant}, we have
  \(\dat = \Subst{\seq{x}}{\dat}\seq{x}= \Dat t\), i.e., \(\dat,\Dat\models t\) as required.
\end{proof}

\begin{proof}[Proof of Theorem~\ref{th:soundness2}]
  This follows by induction on the length of  the rewriting sequence
  \((\seq{x},\Subst{\seq{x}}{\dat})\red^* (t,\Dat)\).
  The base case where
  \((\seq{x},\Subst{\seq{x}}{\dat})= (t,\Dat)\) follows immediately.
  Suppose \[(\seq{x},\Subst{\seq{x}}{\dat})\red (t_1,\Subst{\seq{y}}{\dat_1})
  \red^k  (t,\Dat).\]
  By Lemma~\ref{lem:rewriting-decomp}, there exist \(\tp'\) such that
  \((\seq{y},\Subst{\seq{y}}{\dat_1})\red^k (\tp',\Dat)\) with \(t=[\tp'/\seq{y}]t_1\).
  By the induction hypothesis, \(\tp'\in \Tp\).
  By Lemma~\ref{lem:red-implies-pred}, we have \(t_1\pred \seq{y}\).
  By Lemma~\ref{lem:pred-closed-under-subst},
  we have \(\tp=[\tp'/\seq{y}]\tp_1\preds [\tp'/\seq{y}]\seq{y}=\tp'\).
  Thus, we have \(\tp\in \Tp\) as required. 
\end{proof}

\subsubsection{Proofs of Completeness (Theorems~\ref{th:completeness} and \ref{th:completeness-ctpinf})}

\begin{proof}[Proof of Theorem~\ref{th:completeness}]
  For (I), suppose \(\dat \smodels t\) and \(t\in \Tp_n\). By the latter condition, there exists a rewriting sequence
  \(t\preds \seq{y}\) for mutually distinct variables \(\seq{y}\).
  We show 
  \((\seq{x},\Subst{\seq{x}}{\dat})\reds (t,\Dat)\)
  for some \(\Dat\) by induction on the length of the rewriting sequence \(t\preds \seq{y}\).
  The base case where \(t=\seq{y}\) follows immediately.
  Suppose
  \[ t \pred t' \preds \seq{y}.\]
  By Lemma~\ref{lem:completeness}, we have
  \begin{align*}
    & (\seq{x},\Subst{\seq{x}}{\dat})\red (\tp_1,\Subst{\seq{z}}{\dat'})\qquad \tp = [\tp'/\seq{z}]\tp_1 \qquad \dat'\smodels \tp'.
  \end{align*}
  By the induction hypothesis, we have
  \begin{align*}
    & (\seq{z},\Subst{\seq{z}}{\dat'})\reds (\tp',\Dat)
  \end{align*}
  for some \(\Dat\).
  By Lemma~\ref{lem:rewriting-decomp},
  we have \((\tp_1,\Subst{\seq{z}}{\dat'})\reds ([\tp'/\seq{z}]\tp_1,\Dat)=(\tp, \Dat)\).
  Thus, we have
  \((\seq{x},\Subst{\seq{x}}{\dat})\reds (\tp,\Dat)\) as required.

  For (II), suppose \(\dat,\Dat\models \tp\).
  Let \(\seq{x} = x_1,\ldots,x_\ell\) be the variables \(x\) such that
  \(\Dat(x)=\seq{\epsilon}\).
  Then we have \(\dat=\Dat([\seq{\epsilon}/\seq{x}]\tp)\),
  i.e., \(\dat\smodels [\seq{\epsilon}/\seq{x}]\tp\).
  By (I), we have \((\seq{y},\Subst{\seq{y}}{\dat})\reds ([\seq{\epsilon}/\seq{x}]\tp,\Theta')\) for some \(\Theta'\).
  Let \(\tp'\) be an STP such that \(([\seq{\epsilon}/\seq{x}]\tp,\Theta')\reds (\tp',\Theta'')\nred\).
  Then \(\lang(\tp')\subseteq \lang([\seq{\epsilon}/\seq{x}]\tp) \subseteq \lang(\tp)\) as required.
\end{proof}

\begin{proof}[Proof of Theorem~\ref{th:completeness-ctpinf}]
  The condition \(\dat\models \CTPinf(\dat)\) follows immediately from Theorem~\ref{th:soundness}.
  
  Suppose \(\ctp= \tp_1\land \cdots \land \tp_m\) is a CSTP and \(\dat \models \ctp\).
  Then there exists \(\Theta_k\) such that \(\Theta_k\tp_k=\dat\) for each \(k\in\set{1,\ldots,m}\).
  Let \(\seq{x} = x_1,\ldots,x_\ell\) be the variables \(x\) such that
  \(\Theta_k(x)=\seq{\epsilon}\).
  Then we have \(\dat=\Theta_k([\seq{\epsilon}/\seq{x}]\tp_k)\),
  i.e., \(\dat\smodels [\seq{\epsilon}/\seq{x}]\tp_k\).
  By Theorem~\ref{th:completeness}, we have \((\seq{y},\Subst{\seq{y}}{\dat})\reds ([\seq{\epsilon}/\seq{x}]\tp_k,\Theta)\).
  Let \(\tp'_k\) be an STP such that \(([\seq{\epsilon}/\seq{x}]\tp_k,\Theta)\reds (\tp'_k,\Theta')\nred\).
  Then \(\tp'_k\) is a conjunct of
    \(\CTPinf(\dat)\) and \(\lang(\tp'_k)\subseteq \lang([\seq{\epsilon}/\seq{x}]\tp_k)\).
    Thus, we have \(\lang(\CTPinf(\dat))\subseteq \lang(\tp'_k)\subseteq
    \lang([\seq{\epsilon}/\seq{x}]\tp_k) \subseteq \lang(\tp_k)\) for every \(k\in\set{1,\ldots,m}\),
    which implies \(\lang(\CTPinf(\dat))\subseteq \lang(\ctp)\).
\end{proof}
\subsubsection{Proof of Minimality (Theorem~\ref{th:minimality})}

We prepare a few lemmas.
\begin{lemma}
  \label{lem:minimality-base}
  If \((\seq{x},\Subst{\seq{x}}{\dat})\nred\) and \(\dat\models \tp\) with \(\tp\in \Tp_n\),
  then \(\tp\) is of the form \((x_1,\ldots,x_n)\) where \(x_1,\ldots,x_n\) are mutually distinct variables.
\end{lemma}
\begin{proof}
  If \(\tp\in \Tp_n\) is not of the given form, then \(\tp\pred \tp'\) for some \(\tp'\).
  By Lemma~\ref{lem:pred-implies-red},
  \((\seq{x},\Subst{\seq{x}}{\dat})\red (\tp'',\Dat)\) for some \(\tp''\) and \(\Dat\),
  which contradicts the assumption \((\seq{x},\Subst{\seq{x}}{\dat})\nred\). 
\end{proof}

The following is a key lemma.
\begin{lemma}
  \label{lem:pred-commut}
  Suppose \((p_1,\ldots,p_n)\pred (p'_1,\ldots,p'_m)\).
  Then, \((p_1,\ldots,p_n)\in \Tp_n\) if and only if
  \((p'_1,\ldots,p'_m)\in \Tp_m\).
\end{lemma}
\begin{proof}
  The 'if' direction follows immediately from the definition.
  To show the converse, suppose
  \((p_1,\ldots,p_n)\pred (p'_1,\ldots,p'_m)\) and
  \((p_1,\ldots,p_n)\) is solvable, i.e.,
  \((p_1,\ldots,p_n)\preds (x_1,\ldots,x_k)\) for some mutually distinct variables \(x_1,\ldots,x_k\).
  By Lemma~\ref{lem:pred-CR}, there exist \(\tp_3\) and \(\tp_4\) such that
  \((p'_1,\ldots,p'_m)\preds \tp_3\tpequiv \tp_4 \rpreds(x_1,\ldots,x_k)\).
   The last condition implies
\((x_1,\ldots,x_k)\tpequiv \tp_3\), and thus, \((p'_1,\ldots,p'_m)\) is solvable. 
\end{proof}
\begin{proof}[Proof of Theorem~\ref{th:minimality}]
    The proof proceeds by induction on the length of the rewriting sequence
  \((\seq{x},\Subst{\seq{x}}{\dat})\reds (\tp_1,\Dat_1)\).

  In the base case,
  \((\seq{x},\Subst{\seq{x}}{\dat})= (\tp_1,\Dat_1)\not\red\).
  By Lemma~\ref{lem:minimality-base}, \(\tp_0\) must be a tuple of mutually distinct variables \(\seq{y}\).
  Thus, \(\lang(\tp_1)=\lang(\tp_0)=\underbrace{\Alpha^*\times\cdots\times\Alpha^*}_n\) as required.

  In the induction step, we have \((\seq{x},\Subst{\seq{x}}{\dat})\red
  (\tp',\Subst{\seq{y}}{\dat'})\red^k (\tp_1,\Dat_1)\).
  By Lemma~\ref{lem:red-implies-pred}, we have \(\tp'\pred \seq{y}\).
  By Lemma~\ref{lem:rewriting-decomp}, 
  there exists \(\tp_2\) such that 
  \((\seq{y},\Subst{\seq{y}}{\dat'})\red^k (\tp_2,\Dat_1)\) and \(\tp_1=[\tp_2/\seq{y}]\tp'\).
  By the assumption \(\lang(\tp_0)\subseteq \lang(\tp_1)=[\tp_2/\seq{y}]\tp'\)
  and Lemma~\ref{lem:minimality-sim} with \(\tp_1\pred \seq{y}\),
  we have (i) \(\tp'_0 := \residual{\tp_0}{\tp'}{\seq{y}}\) is well-defined;
    (ii) \(\lang(\tp'_0)\subseteq \lang([\tp_2/\seq{y}]\seq{y}) = \lang(\tp_2)\);
  and (iii) \(\lang(\tp'_0)\subseteq \lang(\tp_2)\) implies \(\lang(\tp_0)=\lang(\tp_1)\).
  Furthermore,  by  \(\Tp\ni \tp_0\preds \tp_0'\) and
  Lemma~\ref{lem:pred-commut}, \(\tp_0'\) is also solvable.
    By (ii) and \(\tp_0'\in\Tp\), we can apply the induction hypothesis to obtain \(\lang(\tp_0')= \lang(\tp_2)\), and
    by (iii), we obtain \(\lang(\tp_0)=\lang(\tp_1)\) as required.
\end{proof}

\subsubsection{Proof of the Existence of Characteristic Data (Theorem~\ref{th:data-size})}
\label{sec:proof-others}
To prove Theorem~\ref{th:data-size}, we first define
\emph{characteristic} learning data \(\dat^\tp\) for each solvable tuple pattern \(\tp\).
Let \(\tp=(p_1,\ldots,p_n)\) be a solvable tuple pattern.
Since \(\tp\) is solvable, \(\tp\) contains at most \(n\) distinct variables \(x_1,\ldots,x_k\) with \(k\le n\).
  (Note that, at each reduction step by \(\pred\), the set of variables occurring in the tuple pattern
  does not change, and the number of elements may only monotonically decrease.)
  Let \(a, b\) be two distinct variables of \(\Alpha\), and
  \(\alpha_1,\ldots,\alpha_n\) be binary codewords for \(1,\ldots,n\) using \(a,b\)
  such that \(|\alpha_i|=O(\log n)\) and for all \(i,j (i\neq j)\), \(\alpha_i\) is neither a prefix nor a suffix of \(\alpha_j\).
  Let \(\beta_i\, (i=1,\ldots,n)\) be \([a/b, b/a]\alpha_i\).
  We define \(\dat^\tp\) as:\footnote{The characteristic data \(\dat^\tp\) depends on
  the choice of \(a,b\in\Alpha\) and binary prefix coding, but  the choice does not matter for the discussion below.}
  \[
  \left(
  \begin{array}{lll}
       [\seq{\alpha}/\seq{x}]p_1 & \cdots & [\seq{\alpha}/\seq{x}]p_n\\{}
           [\seq{\beta}/\seq{x}]p_1 & \cdots & [\seq{\beta}/\seq{x}]p_n
  \end{array}
  \right).
  \]
  For example, let \(\tp\) be \((ax_1x_2, x_3bx_2, x_1x_2x_3)\) and
  \(\alpha_1=aa, \alpha_2=ab, \alpha_3=ba\).
  Then, 
  \[
  \dat^\tp = \left(\begin{array}{rrr}
    aaaab & babab & aaabba\\
   abbba & abbba & bbbaab
    \end{array}\right).
  \]
  
   We use the following properties of characteristic data.
   \begin{lemma}
     \label{lem:canonical}
    Let \(\tp\) and \(\tp_1\) be solvable tuple patterns.
    If \(\dat^\tp \models \tp_1\), then \(\lang(\tp)\subseteq \lang(\tp_1)\).
  \end{lemma}
   \begin{proof}
     This follows by induction on \(\measure{\tp_1}\).
     If \(\tp_1\npred\), then \(\tp_1\) must be mutually distinct variables. Thus, \(\lang(\tp_1)=\Alpha^*\times \cdots \times \Alpha^*\supseteq
     \lang(\tp)\).
     Suppose \(\dat^\tp,\Dat\models \tp_1\) and \(\tp_1\pred \tp_1'\). We perform case analysis on the rule used in the reduction
     \(\tp_1\pred \tp_1'\).
     \begin{itemize}
     \item Case \rn{PR-Prefix}:
       We may assume without loss of generality that \(\tp_1=(q_2q'_1,q_2,\ldots,q_n)\) and \(\tp_1'=(q'_1,q_2,\ldots,q_n)\), with
       \(q_2\ne \epsilon\).
       By the assumption \(\dat^\tp,\Dat \models \tp_1\), \(\dat^\tp[*][1]\) must be of the form
       \(\dat^\tp[*][2]\cdot \seq{s}\).
       By the definition of \(\dat^\tp\), \(\tp\) must be of the form \((p_2p'_1,p_2,\ldots,p_n)\), and
       \(\dat^{\tp'},\Dat\models \tp_1'\) for \(\tp'=(p'_1,p_2,\ldots,p_n)\).
       Since \(\measure{\tp_1'}<\measure{\tp_1}\) (by Lemma~\ref{lem:pred-decreases-measure}),
       by the induction hypothesis, we have \(\lang(\tp')\subseteq \lang(\tp_1')\).
       Thus, we have \(\lang(\tp)\subseteq \lang(\tp_1)\).
     \item Cases \rn{PR-CPrefix}, \rn{PR-Suffix}, and \rn{PR-CSuffix}: Similar to the above case.
     \item Case \rn{PR-Epsilon}:
       We may assume without loss of generality that \(\tp_1=(\epsilon,q_2,\ldots,q_n)\) and \(\tp_1'=(q_2,\ldots,q_n)\).
       By the assumption \(\dat^\tp,\Dat \models \tp_1\), \(\dat^\tp[*][1]=\seq{\epsilon}\).
       Therefore, \(\tp\) must be of the form \((\epsilon,p_2,\ldots,p_n)\), with
        \(\dat^{\tp'}, \Dat \models \tp_1'\) for \(\tp'=(p_2,\ldots,p_n)\).
       By the induction hypothesis, we have \(\lang(\tp')\subseteq \lang(\tp_1')\).
       Thus, we have \(\lang(\tp)\subseteq \lang(\tp_1)\) as required.
     \end{itemize}
  \end{proof}

\begin{lemma}
  \label{lem:canonicalsubst}
  Let \(\tp\) be a solvable tuple pattern, and \(\Dat\) be
      \( \left(\begin{array}{ccc}
      x_1 & \cdots & x_k\\
      \alpha_1 & \cdots & \alpha_k\\
      \beta_1 & \cdots & \beta_k\\
      \end{array}\right)\), where \(\alpha_i,\beta_i\) are the codewords for
      \(x_i\)      as defined above.
  If \(\Dat'\tp = \dat^\tp\), 
  then \(\Dat'(x_i) = \Dat(x_i)\) for every \(x_i\in\FV(\tp)\).
\end{lemma}
\begin{proof}
  Since \(\tp\) is solvable, 
  \(\tp\preds (x_1,\ldots,x_k)\).
  We prove the required property by induction on the length of 
  the reduction sequence \(\tp\preds (x_1,\ldots,x_k)\).
  If \(\tp= (x_1,\ldots,x_k)\), the result follows immediately, since \(\dat^\tp=\Dat\,\tp\).
  Otherwise, we have \(\tp\pred \tp' \preds (x_1,\ldots,x_k)\).
  We perform the case analysis on the rule used for deriving \(\tp\pred \tp'\).
  Since the other cases are simpler or simpler, we discuss only the
  case for \rn{PR-Prefix}, where
   we may assume \(\tp=(\pat_2\pat_1',\pat_2,\ldots,\pat_n)\) and
    \(\tp'=(\pat_1',\pat_2,\ldots,\pat_n)\).
    By the assumption \(\Dat'\tp = \dat^\tp\), we have \(\Dat'\tp' = \dat^{\tp'}\).
    By the induction hypothesis, \(\Dat'(x_i) = \Dat(x_i)\) for every
    \(x_i\in\FV(\tp')=\FV(\tp)\), as required.
\end{proof}
  \begin{proof}[Proof of Theorem~\ref{th:data-size}]
  Let \(\tp=(p_1,\ldots,p_n)\) be a solvable tuple pattern such that \(|p_1\cdots p_n|=m\).
    We show \(\dat=\dat^\tp\) satisfies the required properties.
    By the definition of \(\dat^\tp\), \(\sizeof{\dat^\tp} = O((m+n)\log n)\).
    Suppose \(\dat\subseteq \dat'\subseteq \lang(\tp)\).
    We can assume that the first two rows of \(\dat'\) (as a matrix) coincide with \(\dat\).
    By \(\dat'\smodels \tp\) and Theorem~\ref{th:completeness}, we have
    \((\seq{x},\Subst{\seq{x}}{\dat^\tp})\reds (\tp, \Dat)\) for some \(\Dat\).
    For (ii), it remains to show \((\tp, \Dat)\nred\).
    By Theorem~\ref{th:soundness}, we have \(\Dat\tp = \dat'\).
    By Lemma~\ref{lem:canonicalsubst}, \(\Dat\) must be:
    \[ \left(\begin{array}{ccc}
      x_1 & \cdots & x_k\\
      \alpha_1 & \cdots & \alpha_k\\
      \beta_1 & \cdots & \beta_k\\
      \cdots & \cdots & \cdots 
    \end{array}\right).\]
    Because there exist no \(i\) and \(j\) such that \(\alpha_i\) is a prefix or suffix of \(\alpha_j\),
    the rules \rn{R-Prefix} and \rn{R-Suffix} are inapplicable.
    Since \(a\) and \(b\) are swapped between \(\alpha_i\)'s and \(\beta_i\)'s,
    the rules \rn{R-CPrefix} and \rn{R-CSuffix} are also inapplicable.
Thus, we have \((\tp, \Dat)\nred\) as required.
    To show (iii), suppose \((\seq{x},\Subst{\seq{x}}{\dat'})\reds (\tp', \Dat')\nred\).
    By Theorem~\ref{th:soundness}, \(\tp'\models \dat'\).
    By Lemma~\ref{lem:canonical}, we have \(\lang(\tp)\subseteq \lang(\tp')\).
    By Theorem~\ref{th:minimality}, we have \(\lang(\tp)=\lang(\tp')\) as required.
    Obviously, the size of the characteristic data \(\dat^\tp\) is \(O((m+n)\log n)\), and it can be computed 
    also in time \(O((m+n)\log n)\).
\end{proof}

  \subsubsection{Proofs for Section~\ref{sec:learnability}}

  \begin{lemma}
    \label{lem:data-increase}
    Suppose \(\dat' \subseteq \dat\) and \((\seq{x},\Subst{\seq{x}}{\dat})\reds (\tp, \Theta)\).
    Then there exist \(\tp'\) and \(\Theta'\) such that
    \((\seq{x},\Subst{\seq{x}}{\dat'})\reds (\tp', \Theta')\) and \(\tp' = [\seq{\epsilon}/\seq{y}]\tp\)
    for some variables \(\seq{y}\).
  \end{lemma}
  \begin{proof}
    This follows by induction on the length \(n\) of the reduction sequence
    \((\seq{x},\dat)\reds (\tp, \Theta)\).
    If \(n=0\), then \((\seq{x},\Subst{\seq{x}}{\dat})= (\tp, \Theta)\nred\).
    Thus, the required result holds for the empty sequence of variables \(\seq{y}\).

    If \(n>0\), then we have 
    \((\seq{x},\Subst{\seq{x}}{\dat})\red (\tp_1,\Subst{\seq{y}}{\dat_1})\red^{n-1} (\tp, \Theta)\).
    By Lemma~\ref{lem:rewriting-decomp},
    we have
    \((\seq{y},\Subst{\seq{y}}{\dat_1})\red^{n-1} (\tp_2, \Theta)\)
    with \(\tp = [\tp_2/\seq{y}]\tp_1\).
    
    We perform case analysis on the rule used for deriving
    \((\seq{x},\dat)\red (\tp_1,\Theta_1)\). Since the other cases are similar or
    easier, we discuss only the cases for \rn{R-Epsilon} and \rn{R-Prefix}.
    \begin{itemize}
    \item Case \rn{R-Epsilon}:
      In this case, we have:
      \begin{align*}
        & \dat[*][i]=\seq{\epsilon}\\
        & (\seq{y}) = (x_1,\ldots,x_{i-1},x_{i+1},\ldots, x_k)\\
        & \tp_1 = (x_1,\ldots,x_{i-1},\epsilon,x_{i+1},\ldots, x_k)\\
        & \dat_1 = \remove{\dat}{i}.
      \end{align*}
      Since \(\dat'\subseteq \dat\), we have \(\dat'[*][i]=\seq{\epsilon}\).
      Thus, we can apply the same rule to obtain:
      \((\seq{x},\Subst{\seq{x}}{\dat'})\red (\tp_1,\Subst{\seq{y}}{\remove{\dat'}{i}})\). Since \(\remove{\dat'}{i}\subseteq \remove{\dat}{i}=\dat_1\), we can apply induction hypothesis to obtain
      \(\tp_2'\) such that \((\seq{y}, \Subst{\seq{y}}{\remove{\dat'}{i}})\reds (\tp_2',\Theta')\) and \(\tp_2'=[\seq{\epsilon}/\seq{z}]\tp_2\).
      Thus, the required result holds for \(\tp'=[\tp_2'/\seq{y}]\tp_1=[[\seq{\epsilon}/\seq{z}]\tp_2/\seq{y}]\tp_1 = [\seq{\epsilon}/\seq{z}][\tp_2/\seq{y}]\tp_1 = [\seq{\epsilon}/\seq{z}]\tp\).
    \item Case \rn{R-Prefix}:
     In this case, We have:
     \begin{align*}
       & \dat[*][i]\ne\seq{\epsilon}\\
       & \dat[*][j]=\dat[*][i]\cdot \seq{s}\\
        & (\seq{y}) = (x_1,\ldots,x_{j-1},x_j',x_{j+1},\ldots, x_k)\\
        & \tp_1 = (x_1,\ldots,x_{j-1},x_ix_j',x_{j+1},\ldots, x_k)\\
        & \dat_1 = \replace{\dat}{j}{\seq{s}}.
     \end{align*}
     For simplicity, we assume  \(j=2\) and \(i=1\) below.
     If \(\dat'[*][1]\ne \seq{\epsilon}\), then we can apply the same rule
     to \((\seq{x},\Subst{\seq{x}}{\dat'})\) and obtain the required result.
     If \(\dat'[*][1]= \seq{\epsilon}\), then we have
     \(\dat'\subseteq \dat_1\).
     Thus, by the induction hypothesis, we have
     \((\seq{x},\Subst{\seq{x}}{\dat'})\reds (\tp',\Theta')\)
     with \(\tp'=[\seq{\epsilon}/\seq{z}]\tp_2\) for some \(\tp'\).
     Let \(\tp'=(p_1',\ldots,p_k')\).
     Since \(\Theta'\tp'=\dat'\) and \(\dat'[*][1]=\seq{\epsilon}\),
     for every variable \(y\) occurring in \(p_1'\), \(\Theta'(y)=\seq{\epsilon}\).
     Thus, by applying \rn{R-Epsilon} for such variables,
     we obtain \(\tp''\) such that \((\tp',\Theta')\reds (\tp'',\Theta'')\),
     \(\tp''=(\epsilon,p_2'',\ldots,p_k'')\), and \(\tp''=[\seq{\epsilon}/\seq{u}]\tp'=[\seq{\epsilon}/\seq{v}]\tp_2\)
     for \(\seq{v} = \seq{u},\seq{z}\).
     Let \(\tp_2=(p_1,\ldots,p_k)\). Then \(\tp= (p_1,p_1p_2,p_3,\ldots,p_k)\).
     Since \(\epsilon=[\seq{\epsilon}/\seq{v}]p_1\),
       we have 
     \(\tp''=[\seq{\epsilon}/\seq{v}]\tp_2 = [\seq{\epsilon}/\seq{v}]\tp\),
     as required.
    \end{itemize}
  \end{proof}

  For data \(\dat\) such that \(\colsof{\dat}=n\),
  we define \(\CSTPs(\dat)\) as \(\CSTPnorm(\ctpset_0,\ldots,\ctpset_n)\) where:
  \begin{align*}
 &   \ctpset_i = \set{\tp \mid (\seq{x},\Subst{\seq{x}}{\dat}\reds (\tp,\Theta)
      \mbox{ and } |\FV(\tp)|=i}\\
    &  \CSTPnorm(\ctpset_0,\ldots,\ctpset_n)=(\ctpset'_0,\ldots,\ctpset'_n)\\
    & \ctpset'_i = \set{\tp\in \ctpset_i\mid \mbox{there exists no \(\tp'\in\ctpset_j\)
        such that \(j<i\) and \(\tp'=[\seq{\epsilon}/\seq{y}]\tp\).}}
  \end{align*}
  For example, for \(\dat=\left(\begin{array}{rrr} a & ab & abab\\
    bb & b & bb
  \end{array}\right)\),
  \(\CSTPs(\dat)\) is:
  \[(\emptyset, \emptyset, \set{(x,y, yy)}, \set{(x,y, z), (x,y,xz), (x,y,yz),(x,y,zy)}).\]
  (Note that we also have \(((x,y,z),\Subst{(x,y,z)}{\dat})\reds ((x,y,yyz),\Theta)\), but
\((x,y,yyz)\) is subsumed by \((x,y,yy)\), as \([\epsilon/z](x,y,yyz)=(x,y,yy)\).)
  
  We claim:
  \begin{lemma}
    \label{lem:CSTPs}
  Suppose \(\dim(\dat)=n\). If \(\CSTPs(\dat)=(\ctpset_0,\ldots,\ctpset_n)\),
  then \(\lang(\CTPinf(\dat)) = \lang(\bigwedge_{\tp\in \ctpset_0\cup \cdots \cup \ctpset_n} \tp)\).
  \end{lemma}
  \begin{proof}
    To show \(\lang(\CTPinf(\dat)) \supseteq
    \lang(\bigwedge_{\tp\in \ctpset_0\cup \cdots \cup \ctpset_n} \tp)\),
    suppose \(\tp\) is a conjunct of \(\CTPinf(\dat)\).
    Since \(\tp\) contains at most \(n\) variables, by the definition of
    \(\CSTPs(\dat)\), 
    there exists \(\tp'\in \ctpset_j\)
    such that \(\tp'=[\seq{\epsilon}/\seq{y}]\tp\).
    Thus, \(\lang(\tp)\supseteq \lang(\tp')\supseteq \lang(\bigwedge_{\tp\in \ctpset_0\cup \cdots \cup \ctpset_n} \tp)\). Therefore, we have
    \(\lang(\CTPinf(\dat)) \supseteq
    \lang(\bigwedge_{\tp\in \ctpset_0\cup \cdots \cup \ctpset_n} \tp)\), as required.

    For the converse, suppose \(\tp\in \ctpset_i\). Then,
    we have \((\seq{x},\Subst{\seq{x}}{\dat})\reds (\tp,\Theta)\) for some
    \(\Theta\). Let \(\tp'\) be a tuple pattern such that
    \((\tp,\Theta)\reds (\tp',\Theta')\nred\). Then, \(\tp'\) is a conjunct of \(\CTPinf(\dat)\),
    and therefore, we have \(\lang(\CTPinf(\dat) \subseteq \lang(\tp')\subseteq \lang(\tp)\).
    Thus, we have \(\lang(\CTPinf(\dat)) \subseteq
    \lang(\bigwedge_{\tp\in \ctpset_0\cup \cdots \cup \ctpset_n} \tp)\), as required.
  \end{proof}
  We define the order \((\ctpset_0,\ldots,\ctpset_n)<(\ctpset'_0,\ldots,\ctpset'_n)\)
  by: \((\ctpset_0,\ldots,\ctpset_n)<(\ctpset'_0,\ldots,\ctpset'_n)\)
  iff \(\exists i.
  (\ctpset_i\subsetneq \ctpset'_i \land \forall j<i.\ctpset_j= \ctpset'_j )\).
  Note that \(<\) is well-founded.

  \begin{lemma}
    \label{lem:convergence_of_CTPinf}
    Suppose \(\dat_0,\dat_1,\dat_2,\ldots\) is an infinite, strictly decreasing sequence
    of non-empty data,  i.e.,
    \(\emptyset\subsetneq \dat_0\subsetneq \dat_1\subsetneq \dat_2\subsetneq \cdots\).
    Then, there exists \(j\) such that for all \(i\ge j\), \(\lang(\CTPinf(\dat_i))=\lang(\CTPinf(\dat_j))\).
  \end{lemma}
  \begin{proof}
    By Lemma~\ref{lem:data-increase}, we have:
    \[\CSTPs(\dat_0)\ge \CSTPs(\dat_1)\ge \CSTPs(\dat_2)\ge \cdots\]
    Since \(<\) is well-founded,
    there exists \(j\) such that for all \(i\ge j\),  \(\CSTPs(\dat_i)=\CSTPinf(\dat_j)\).
   By Lemma~\ref{lem:CSTPs}, we have
\(\lang(\CTPinf(\dat_i))=\lang(\CTPinf(\dat_j))\) for all \(i\ge j\).
  \end{proof}

  \begin{proof}[Proof of Theorem~\ref{th:learnability}]
    Suppose \(\tp\) is an STP, and \(\lang(\tp) = \set{\seq{s}_i\mid i\in\omega}\) with
    \(\dat_i = \set{\seq{s}_j\mid 0\le j\le i}\). By Theorem~\ref{th:data-size},
    there exists \(\dat\) such that \(\TPinf(\dat')\) always returns an STP equivalent to \(\tp\)
    for any \(\dat'\supseteq \dat\),
    and by the soundness
    of \(\TPinf\) (Theorem~\ref{th:soundness}), we have \(\dat\subseteq \lang(\tp)\).
    Since \(\lang(\tp) = \set{\seq{s}_i\mid i\in\omega}\), there exists \(j\) such that
    \(\dat\subseteq \dat_j\). Therefore, for every \(k\ge j\), \(\TPinf(\dat_k)\) always returns an STP equivalent to \(\tp\).

    To show the case for CSTPs, suppose \(\ctp\) is a CSTP and \(\lang(\ctp)=\set{\seq{s}_i\mid i\in\omega}\) with \(\dat_i = \set{\seq{s}_j\mid 0\le j\le i}\).
    By Lemma~\ref{lem:convergence_of_CTPinf},
    there exist \(\ctp'\) and \(j\) such that for all \(i\ge j\), \(\lang(\CTPinf(\dat_i))=\lang(\ctp')\).
    By the soundness of \(\TPinf\) (Theorem~\ref{th:soundness}),
    we have \(\lang(\ctp)=\bigcup_{i\in\omega} \dat_i \subseteq \lang(\ctp')\).
    By Theorem~\ref{th:completeness-ctpinf}, we also have \(\lang(\ctp)\supseteq \lang(\ctp')\).
Therefore, we have \(\lang(\ctp)= \lang(\ctp')\), as required.
  \end{proof}

  \begin{proof}[Proof of Theorem~\ref{th:cstp-lfp}]
    We first show the termination of the procedure by contradiction.
    Suppose the procedure does not terminate.
    Then, there exist infinite sequences \(\seq{s}_0,\seq{s}_1,\seq{s}_2,\ldots\)
    and \(\ctp_0,\ctp_1,\ctp_2,\ldots\) such that:
    \begin{align*}
      & \ctp_0 = (a,\seq{\epsilon})\land (\epsilon,\seq{\epsilon})\qquad \dat_j = \set{\seq{s}_i \mid i<j}.\\
      & \CTPinf(\dat_j)=\ctp_j \mbox{ for $j\ge 1$}\\
      &   \oracle(\ctp_j)=Some(\seq{s}_j) \mbox{ for $j\ge 0$}.
    \end{align*}
    These \(\ctp_j\), \(\seq{s}_j\), and \(\dat_j\) are the values of \(\ctp\), \(\seq{s}\), and \(\dat\)
    at the \(j\)-th iteration of the loop in the procedure.
By the last condition and the assumption on \(\oracle\), we have
    \(\seq{s}_j\in \mathcal{F}(\lang(\ctp_j))\setminus \lang(\ctp_j)\).
    By the soundness of \(\CTPinf\) (Theorem~\ref{th:soundness}),
    we have \(\dat_j\subseteq \lang(\ctp_j)\). Therefore, \(\seq{s}_j\not\in \dat_j\),
    and hence, \(\dat_0=\emptyset \subsetneq \dat_1\subsetneq \dat_2\subsetneq \dat_3\subsetneq\cdots\)
    is a strictly increasing sequence. By Lemma~\ref{lem:convergence_of_CTPinf},
    there exist \(k\) and \(\ctp\) such that \(\lang(\ctp_j)=\lang(\ctp)\) for all \(j\ge k\).
    This however, contradicts 
    \(\seq{s}_j\in \mathcal{F}(\lang(\ctp_j))\setminus \lang(\ctp_j)\) and
    \(\seq{s}_j\in \dat_{j+1}\subseteq \lang(\ctp_{j+1})\).

    Next, we show that the procedure returns the least CSTP \(\ctp\)
    such that \(\lang(\ctp)\supseteq \mathcal{F}(\lang(\ctp))\).
    Suppose \(\lang(\ctp')\supseteq \mathcal{F}(\lang(\ctp'))\) for a CSTP \(\ctp'\).
    Let \(\ctp_j\), \(\seq{s}_j\), and \(\dat_j\) be the values of \(\ctp\), \(\seq{s}\), and \(\dat\)
    at the \(j\)-th iteration of the loop in the procedure, as given above.
    It suffices to show \(\lang(\ctp_j)\subseteq \lang(\ctp')\) by induction on \(j\).
    The case \(j=0\) is obvious, as \(\lang(\ctp_j)=\emptyset\).
    For the induction step, suppose \(\lang(\ctp_j)\subseteq \lang(\ctp')\)  and
    \(\oracle(\ctp_j)\) returns \(Some(\seq{s}_j)\). 
    By the monotonicity of \(\mathcal{F}\), we have 
    \[\mathcal{F}(\lang(\ctp_j))\subseteq \mathcal{F}(\lang(\ctp'))\subseteq \lang(\ctp').\]
    Thus, \(\seq{s}_j\in \mathcal{F}(\lang(\ctp_j))\setminus \lang(\ctp_j) \subseteq \lang(\ctp')\).
    Therefore, we have \[\dat_{j+1}=\dat_j\cup\set{\seq{s}_j}\subseteq
    \lang(\ctp_j)\cup\set{\seq{s}_j}\subseteq \lang(\ctp').\]
    By Theorem~\ref{th:completeness-ctpinf}, we have
    \(\ctp_{j+1}=\CTPinf(\dat_{j+1})\subseteq \lang(\ctp')\), as required.
  \end{proof}

  \begin{proof}[Proof of Theorem~\ref{th:stp-lfp}]
    Let \(\mathtt{enumerate}\) be the procedure defined as below, and suppose \(\seq{s}_0\in \mathcal{F}(\emptyset)\).
    Then \(\mathtt{enumerate}(\set{\seq{s}_0})\) outputs a superset of
    all the minimal \(\tp\)'s such that \(\lang(\tp)\supseteq \mathcal{F}(\lang(\tp))\).
    We can then filter out non-minimal STPs by using the algorithm of Theorem~\ref{th:decision-problems} (3).
    \begin{algorithm}
      \SetKwFunction{Enumerate}{enumerate}
      \SetKwProg{Fn}{Function}{:}{}
      \Fn{\Enumerate{$\dat$}}{
{$U$ := $\set{\tp \mid (\seq{x},\Subst{\seq{x}}{\dat})\reds (\tp,\Dat)\nred}$;}\\
        \While{$U\ne \emptyset$}{
          $\tp$ := pick($U$); $U$ := $U\setminus\set{\tp}$;\\
          \uIf{$\oracle(\tp)=Some(\seq{s})$}
              {$U := U\setminus \set{\tp' \mid \seq{s}\in \lang(\tp')}$;\\
                \Enumerate($\dat\cup\set{\seq{s}}$);}
              \uElse
                  {$\mathit{output}$($\tp$);}
          }
        }
    \end{algorithm}
    \\
    Note that $U$ in the above algorithm is a local variable, prepared for each recursive call.\\
    We first show the termination of the algorithm.
    Suppose \(\mathtt{enumerate}(\set{\seq{s}_0})\) does not terminate.
    Since $\set{\tp \mid (\seq{x},\Subst{\seq{x}}{\dat})\reds (\tp,\Dat)\nred}$
    is a finite set, there must be an infinite chain of recursive calls
    \(\mathtt{enumerate}(\dat_0), \mathtt{enumerate}(\dat_1), \mathtt{enumerate}(\dat_2), \ldots\)
    where \(\dat_0=\set{\seq{s}_0}\) and \(\dat_{i}=\dat_{i-1}\cup\set{\seq{s}_i}\) for \(i>0\).
    Let \(U_i\) be the value of \(U\) initialized at the beginning of the call \(\mathtt{enumerate}(\dat_i)\),
    and \(\tp_i\) be the value of \(\tp\) when \(\mathtt{enumerate}(\dat_{i+1})\) is called.
    Thus, \(\dat_i\subseteq \lang(\tp_i)\) and \(\seq{s}_{i+1}\not\in \lang(\tp_i)\).
    By Theorem~\ref{th:learnability}, there exists \(k\) and \(\tp'\) such that
    $\lang(\tp_i) = \lang(\tp')$ for all \(i\ge k\).
    But this contradicts the conditions \(\seq{s}_{i}\in \dat_{i}\subseteq \lang(\tp_{i})\) but
    \(\seq{s}_{i}\not\in\lang(\tp_{i-1})\).

    Next, suppose that \(\tp\) is a minimal STP such that \(\lang(\tp)\supseteq \mathcal{F}(\lang(\tp))\).
    It suffices to show that for any \(\dat\) such that \(\dat\subseteq \lang(\tp)\),
    \(\mathtt{enumerate}(\dat)\) outputs (an STP equivalent to) \(\tp\), or makes a recursive call \(\mathtt{enumerate}(\dat')\)
    for \(\dat'\)
    such that \(\dat\subsetneq \dat'\subseteq \lang(\tp)\) (as we have already shown the termination of the algorithm).
    
    By the completeness (Theorem~\ref{th:completeness}) and minimality (Theorem~\ref{th:minimality}),
    the value of \(U\) initialized at the beginning of the call of \(\mathtt{enumerate}(\dat)\)
    must contain \(\tp'\) such that \(\dat\subseteq \lang(\tp')\subseteq \lang(\tp)\).
    We perform case analysis on whether \(\tp'\) is picked up in the while-loop. 
    \begin{itemize}
      \item
    Suppose \(\tp'\) is eventually picked up in the while-loop. Then if \(\lang(\tp)=\lang(\tp')\), then
    \(\oracle(\tp')\) returns \(\None\) and \(\tp'\) is output.
    If \(\lang(\tp')\subsetneq \lang(\tp)\), then 
    \(\oracle(\tp')\) returns \(Some(\seq{s})\) such that \(\seq{s}\in \mathcal{F}(\lang(\tp'))\subseteq
    \mathcal{F}(\lang(\tp))\subseteq \lang(\tp)\). Thus, \(\mathtt{enumerate}(\dat')\) is called for
    \(\dat'=\dat\cup\set{\seq{s}}\subseteq \lang(\tp)\), as required.
    \item
    Suppose \(\tp'\) is never picked up in the while-loop. Then \(\tp'\) must have been removed by the assignment
    \(U := U\setminus \set{\tp' \mid \seq{s}\in \lang(\tp')}\).
    But then it must be the case that \(\seq{s}\in\lang(\tp')\). Therefore, \(\mathtt{enumerate}(\dat')\) is called for
    \(\dat'=\dat\cup\set{\seq{s}}\subseteq \lang(\tp')\subseteq \lang(\tp)\), as required.
    \end{itemize}
  \end{proof}

  \subsubsection{Proofs for Section~\ref{sec:decision}}
  \label{sec:proof-decidability}
\begin{proof}[Proof of Theorem~\ref{th:decision-problems}]
  (1): By Lemma~\ref{lem:pred-wconf}, if \(\tp\pred\tp'\), then
  \(\tp\) is solvable if and only if \(\tp'\) is solvable.
  By Lemma~\ref{lem:pred-decreases-measure}, the length
  of a reduction sequence \(\tp\pred \tp_1\pred \cdots \pred \tp_n\) is bounded
  by \(\measure{\tp}\). Thus, to check whether \(\tp\) is solvable,
  it suffices to repeatedly reduce
  \(\tp\) (in a non-deterministic manner) until no reduction is applicable,
  and check whether the resulting tuple pattern is of the form
  \((x_1,\ldots,x_k)\). As each step can be performed in polynomial
  time, the algorithm runs in polynomial time.

  (3): Based on Lemma~\ref{lem:minimality-sim},
  we can use the algorithm in Algorithm~\ref{alg:inclusion}.
\iftrue
  \begin{algorithm}
\SetKwFunction{DecideInc}{decide\_inc}
\SetKwProg{Fn}{Function}{:}{}
\Fn{\DecideInc{$\tp_1,\tp_2$}}{
  \uIf{$\tp_2$ is of the form $(x_1,\ldots,x_n)$ where $x_1,\ldots,x_n$ are pairwise distinct}{
    \KwRet{\textsc{True}}
  }
  pick $\tp_2'$ such that $\tp_2 \pred \tp_2'$ \;
  \uIf{$\tp_1' := \residual{\tp_1}{\tp_2}{\tp_2'}$ is well defined}{
    \KwRet{\DecideInc{$\tp_1',\tp_2'$}}
  }\uElse{
    \KwRet{\textsc{False}}
  }
}
\caption{Algorithm to decide \(\lang(\tp_1)\stackrel{?}{\subseteq} \lang(\tp_2)\).}
\label{alg:inclusion}
\end{algorithm}
\else
\begin{algorithmic}[1]
  \Function{decide\_inc}{$\tp_1, \tp_2$}
  \If{$\tp_2$ is of the form $(x_1,\ldots,x_n)$ (where $x_1,\ldots,x_n$ are mutually different)}
  \State{return TRUE};
\EndIf
\State{pick \(\tp_2'\) such that \(\tp_2\pred \tp_2'\);}
\If{\(\tp_1' := \residual{\tp_1}{\tp_2}{\tp_2'}\) is well defined}
\State{return \textsc{decide\_inc}($\tp_1',\tp_2'$)};
\Else \State{return FALSE;}
\EndIf
\EndFunction
\end{algorithmic}
\fi
\ \\
\noindent
Here, we assume \(|\tp_1|=|\tp_2|\); otherwise, we can immediately conclude
that \(\lang(\tp_1)\not\subseteq \lang(\tp_2)\).
The algorithm repeatedly reduces \(\tp_2\) and checks whether the reduction
of \(\tp_2\) can be ``simulated by'' the corresponding reduction of \(\tp_1\).
If \(\tp_2\) is of the form \((x_1,\ldots,x_n)\),
we can immediately conclude that \(\tp_1\subseteq \tp_2\) holds,
as \(\lang(\tp_2)=\Alpha^*\times \cdots\times \Alpha^*\).
Otherwise, we can always pick \(\tp_2'\) such that $\tp_2 \pred \tp_2'$ by the assumption that \(\tp_2\) is solvable.
By Lemma~\ref{lem:minimality-sim}, if \(\residual{\tp_1}{\tp_2}{\tp_2'}\)
is well defined, then it suffices to check whether \(\lang(\tp_1')\subseteq \lang(\tp_2')\);
 otherwise, we can immediately conclude that \(\lang(\tp_1)\not\subseteq
\lang(\tp_2)\) (on the last line).

Each operation can be clearly performed in time polynomial in the size of
the initial inputs \(\tp_1\) and \(\tp_2\).
Since \(\measure{\tp_2}> \measure{\tp_2'}\ge 0\), the depth of the recursive call
is bounded by \(\measure{\tp_2}\). Thus, the algorithm runs in polynomial time.
(As this proof shows, \(\tp_1\) need not be solvable.)

(2): This is a special case of (3), where \(\tp_2=\tp\) and \(\tp_1\) is
the constant pattern \((s_1,\ldots,s_n)\) (whose language is a singleton set
\(\set{(s_1,\ldots,s_n)}\)).

(4): This is an immediate consequence of (3), as \(\lang(\tp_1)=\lang(\tp_2)
\IFF \lang(\tp_1)\subseteq \lang(\tp_2)\land
\lang(\tp_2)\subseteq \lang(\tp_1)\).
\end{proof}

Before proving Theorem~\ref{th:satisfiability}, we define the semantics of formulas.
A \emph{valuation} is a map from the set of variables to \(\Alpha^*\).
For a valuation \(\rho\) and \(w\in (\Alpha\cup \Vars)^*\) such that
\(\FV(w)\subseteq \dom(\rho)\), we define \(\rho w\) by:
\[ \rho \epsilon = \epsilon \qquad \rho (aw) = a\cdot (\rho w)\qquad \rho(xw) = \rho(x)\cdot (\rho w).\]
For a valuation \(\rho\) and a quantifier-free STP-formula \(\form\),
the relation \(\rho\models \form\) is defined by induction on \(\form\), as follows.
\begin{align*}
&  \rho \models w_1=w_2 \mbox{ if } \rho w_1 = \rho w_2\\
&  \rho \models w_1\ne w_2 \mbox{ if } \rho w_1 \ne \rho w_2\\
&  \rho \models (w_1,\ldots,w_n)\in\lang(\tp) \mbox{ if } (\rho w_1,\ldots,\rho w_n)\in \lang(\tp)\\
&  \rho \models (w_1,\ldots,w_n)\not\in\lang(\tp) \mbox{ if } (\rho w_1,\ldots,\rho w_n)\not\in \lang(\tp)\\
&  \rho \models \form_1\land \form_2 \mbox{ if } \rho\models \form_1 \mbox{ and }\rho\models \form_2\\
&  \rho \models \form_1\lor \form_2 \mbox{ if } \rho\models \form_1 \mbox{ or }\rho\models \form_2.
\end{align*}
We say a formula \(\form\) is \emph{satisfiable} if \(\rho\models \form\) for some \(\rho\).
Without subformulas \(\seq{w}\in\lang(\tp)\) and \(\seq{w}\not\in\lang(\tp)\),
it is known that the satisfiability of quantifier-free formulas is decidable~\cite{Makanin,Plandowski,10.1145/337244.337255}.
\begin{proof}[Proof of Theorem~\ref{th:satisfiability}]
  Since the satisfiability of word equations is decidable and in PSPACE~\cite{Makanin,Plandowski,10.1145/337244.337255},
  it suffices to show that
  atomic formulas of the form \(w_1\ne w_2\), \((w_1,\ldots,w_k)\in \lang(\tp)\), or
  \((w_1,\ldots,w_k)\not\in \lang(\tp)\) can be encoded into formulas using only word equations of polynomial size.
  We can assume, without loss of generality, that the alphabet \(\Alpha\) is finite:
  note that if \(\form\) is satisfiable for an infinite alphabet, then
  \(\form\) is also satisfiable for \(\Alpha_\form\cup \set{a,b}\) where \(\Alpha_\form\) is the set of letters
  occurring in \(\form\) and \(a,b\) are distinct letters not occurring in \(\form\),
  and vice versa.
  First, we encode the relation \(w_1\not\preceq w_2\), which means ``\(w_1\) is not a prefix of \(w_2\).''
  It can be expressed by \((\bigvee_{a,b\in \Alpha, a\ne b}w_1=xay\land w_2=xbz)\lor
  \bigvee_{a\in \Sigma} w_1=w_2ax\), where \(x,y,z\) are fresh variables. (This encoding of
  \(\not\preceq\) is standard; see, e.g. \cite{DIEKERT2005105,10.1145/337244.337255}.)
  We can then express
  \(w_1\ne w_2\) as \(w_1\not\preceq w_2\lor w_2\not\preceq w_1\).

  The encoding of \((w_1,\ldots,w_k)\in \lang(\tp)\) is straightforward.
  Suppose \(\tp = (\pat_1,\ldots,\pat_\ell)\). If  \(k\ne \ell\), then
  \((w_1,\ldots,w_k)\in \lang(\tp)\) can be replaced with \(\false{}\) (which can be expressed as \(a=\epsilon\)).
  If \(k=\ell\), then \((w_1,\ldots,w_k)\in \lang(\tp)\) is expressed as
  \(w_1=\pat'_1\land \cdots \land w_k=\pat'_k\), where \((\pat'_1,\ldots,\pat'_k)\) are obtained from
  \(\tp\) by \(\alpha\)-renaming, so that
  the variables in \(\pat'_1,\ldots,\pat'_k\) are fresh.

  Now we consider \((w_1,\ldots,w_k)\not\in \lang(\tp)\), where \(\tp=(\pat_1,\ldots,\pat_\ell)\).
  If \(k\ne \ell\), then it can be replaced with \(\true{}\) (which can be expressed as \(\epsilon=\epsilon\)).
  Thus, we assume \(k=\ell\). Since \(\tp\) is solvable,
  there exists a reduction sequence \(\tp \preds (x_1,\ldots,x_m)\).
  We encode the formula by induction on the length \(n\) of the reduction sequence.
  If \(n=0\), then \(\lang(\tp)=\lang(x_1,\ldots,x_k) = \Alpha^*\times \cdots \times \Alpha^*\).
  Thus, \((w_1,\ldots,w_k)\not\in \lang(\tp)\) can be replaced with \(\false\).
  If \(n>0\), then \(\tp \pred \tp'\preds (x_1,\ldots,x_m)\).
  We perform case analysis on the rule used for deriving \(\tp\pred \tp'\).
  Since the other cases are similar, we discuss only the cases for \rn{PR-Prefix} and \rn{PR-Epsilon}.
  \begin{itemize}
  \item Case \rn{PR-Prefix}:
  In this case, there exist \(i,j\) such that \(\tp'=(\pat'_1,\ldots,\pat'_k)\) with
  \(\pat_j = \pat_i\pat'_j\) and \(\pat'_\ell=\pat_\ell\) for all \(\ell\ne j\).
  Thus, \((w_1,\ldots,w_k)\not\in \lang(\tp)\) can be replaced with
  \(w_i\not\preceq w_j \lor (w_j=w_i x\land (w'_1,\ldots,w'_k)\not\in\lang(\tp'))\),
  where \(x\) is a fresh variable, 
  \(w'_j=x\), and \(w'_\ell=w_\ell\) for all \(\ell\ne j\). By the induction hypothesis,
  \((w'_1,\ldots,w'_k)\not\in\lang(\tp')\) can also be encoded into word equations.

  \item Case \rn{PR-Epsilon}:
    In this case,
    there exists \(j\) such that \(\pat_j=\epsilon\) and \(\tp'=(\pat_1,\ldots,\pat_{j-1},\pat_{j+1},\ldots,\pat_k)\).
  Thus, \((w_1,\ldots,w_k)\not\in \lang(\tp)\) can be replaced with
  \(w_j\ne \epsilon \lor (w_1,\ldots,w_{j-1},w_{j+1},\ldots,w_k)\not\in\lang(\tp')\).
  By the induction hypothesis,
  \((w'_1,\ldots,w'_k)\not\in\lang(\tp')\) can also be encoded into word equations.
  \end{itemize}
  The size of the resulting formula is polynomial in the size of the original formula, as
  the length of the reduction sequence \(\tp \preds (x_1,\ldots,x_m)\) is linear in \(\measure{\tp}\) and
  the increase of the formula size in each step of the encoding above is polynomial in the size of \(\seq{w}\not\in\lang(\tp)\).
\end{proof}

\begin{remark}
  \label{rem:satisfiability}
  The encoding of \(\seq{w}\not\in\lang(\tp)\) above for the inductive case 
  is actually based on a special case of Lemma~\ref{lem:minimality-sim}
  where \(\tp_0\) is a constant pattern \(\seq{s}\).
  By Lemma~\ref{lem:minimality-sim}, \(\seq{s}\not\in \lang(\tp)\) if and only if
   \(\residual{\seq{s}}{\tp}{\tp'}\) is undefined or  \(\lang(\residual{\seq{s}}{\tp}{\tp'})
   \not\subseteq \lang(\tp')\). The (un)definedness of \(\residual{\seq{s}}{\tp}{\tp'}\)
   can be expressed in terms of the prefix/suffix relations, which can further be
   expressed by word equations, as discussed above.
\end{remark}
 \subsection{Additional Definitions and Proofs for Section~\ref{sec:ext}}
\label{sec:extproof}

\subsubsection{The Extension in Section~\ref{sec:ext-rev}}

First, we extend \(\sim\) (used in Lemma~\ref{lem:pred-wconf})
to the least equivalence relation closed under not only the permutations and
but also the rule \((p_1,\ldots,p_{k}) \sim (p_1,\ldots,p_{i-1},p_i^R,p_{i+1},\ldots, p_{k})\),
and slightly relax the requirement for solvability:
a tuple pattern \(\tp\) is \emph{solvable} if \(\tp \preds \sim (x_1,\ldots,x_n)\) for
some mutually distinct variables \(x_1,\ldots,x_n\).
Thus, \((x^R, y)\) is also solvable, as \((x^R,y)\sim (x,y)\). Without the reverse pattern constructor,
the definition of solvability is the same.

All the main properties are preserved by the extension of patterns
with the reverse constructor.
This can be intuitively understood by noting that
 the rules \rn{R-RPrefix} and \rn{R-RSuffix} 
can be emulated by \rn{R-Prefix} and \rn{R-Suffix}, if we augment the data \(\dat\)
by adding, for each column \(\dat[*][i]\) of \(\dat\), a new column consisting of the reverse of elements of \(\dat[*][i]\).
For example, 
  \begin{align*}
 &   ((x_1,x_2), \Subst{(x_1,x_2)}{\left(\begin{array}{rrr}
    ab &\ bac \\
  \end{array}\right)})\\
    &  \red ((x_1, \rpat{x_1}x'_2), \Subst{(x_1, x'_2)}{
    \left(\begin{array}{rrr}
    ab &\ c \\
    \end{array}\right)}) & \mbox{ (\rn{R-RPrefix})}
  \end{align*}
can be emulated by:
  \begin{align*}
 &   ((x_1,x_2,y_1,y_2), \Subst{(x_1,x_2,y_1,y_2)}{\left(\begin{array}{rrrr}
    ab &\ bac &\ ba &\ cab\\
  \end{array}\right)})\\
    &  \red ((x_1, y_1x'_2,y_1,y_2), \Subst{(x_1, x'_2,y_1,y_2)}{
    \left(\begin{array}{rrrr}
    ab &\ c &\ ba &\ cab\\
    \end{array}\right)}) & \mbox{ (\rn{R-Prefix})}\\
    &  \red ((x_1, y_1x'_2,y_1,y'_2x_1), \Subst{(x_1, x'_2,y_1,y'_2)}{
    \left(\begin{array}{rrrr}
    ab &\ c &\ ba &\ c\\
    \end{array}\right)}) & \mbox{ (\rn{R-Suffix})}
  \end{align*}
  It should be obvious that
  the extended pattern \((\pat_1,\ldots,\pat_n)\) (consisting of variables \(z_1,\ldots,z_m\))
  is inferable if and only if
  \((\theta_1\pat_1,\ldots,\theta_1\pat_n, \theta_2\rpat{\pat_1},\ldots,\theta_2\rpat{\pat_n})\) is inferable from the augmented data,
  where \(\theta_1 = [x_1/z_1,\ldots,x_m/z_m,y_1/\rpat{z_1},\ldots,y_m/\rpat{z_m}]\) and
  \(\theta_2 = [y_1/z_1,\ldots,y_m/z_m,x_1/\rpat{z_1},\ldots,x_m/\rpat{z_m}]\).

 A little more formally, we can confirm that the main properties listed in Table~\ref{tab:lemmas},
 by checking that all the lemmas in the table remain to hold (modulo a minor adjustment due to the extension of
 solvability).  The lemmas required for soundness and completeness (Theorem~\ref{th:soundness},
 Theorem~\ref{th:soundness2} and Theorem~\ref{th:completeness}) trivially hold.
 More care is necessary for minimality (Theorem~\ref{th:minimality}).
 For Lemma~\ref{lem:minimality-sim}, we extend the definition of
 \(\residual{\tp_0}{\tp_1}{\tp_1'}\) (Definition~\ref{df:residual-pattern}) by the following cases:
 \begin{itemize}
   \item Case  \rn{PR-RPrefix}:
  In this case, \(\tp_1=(\pat_1,\ldots,\pat_n)\) and \(\tp_2=(\pat'_1,\ldots,\pat'_n)\)
  with \(\pat_j=\pat_i^R\cdot \pat'_j\) and \(\pat_i\ne \epsilon\)
  for some \(i, j\) and \(\pat'_k=\pat_k\) for \(k\ne j\).
  If \(\tp_0\) is of the form \((q_1,\ldots,q_n)\) and \(q_j=q_i^Rq_j'\), then \(\residual{\tp_0}{\tp_1}{\tp_2}:= (q'_1,\ldots,q'_n)\) where
 for \(k\ne j\).
  Otherwise, \(\residual{\tp_0}{\tp_1}{\tp_2}\) is undefined.
   \item Case  \rn{PR-RSuffix}:
  In this case, \(\tp_1=(\pat_1,\ldots,\pat_n)\) and \(\tp_2=(\pat'_1,\ldots,\pat'_n)\)
  with \(\pat_j=\pat'_j\cdot \pat_i^R\) and \(\pat_i\ne \epsilon\)
  for some \(i, j\) and \(\pat'_k=\pat_k\) for \(k\ne j\).
  If \(\tp_0\) is of the form \((q_1,\ldots,q_n)\) and \(q_j=q_j'\cdot q_i^R\), then \(\residual{\tp_0}{\tp_1}{\tp_2}:= (q'_1,\ldots,q'_n)\) where
 for \(k\ne j\).
  Otherwise, \(\residual{\tp_0}{\tp_1}{\tp_2}\) is undefined.
 \end{itemize}
Then Lemma~\ref{lem:minimality-sim} holds as required.

Lemma~\ref{lem:pred-wconf} (weak confluence of \(\pred\) up to \(\sim\))
also remains to hold.
We discuss only one new case; the other cases are similar.
\begin{itemize}
  \item  Case \(\tp=(p_2p'_1,p_2,p_3,\ldots,p_n)\pred (p'_1,p_2,p_3,\ldots,p_n)=\tp_1\)
    and \(\tp=(p_3^Rp''_1,p_2,p_3,\ldots,p_n)\pred (p''_1,p_2,p_3,\ldots,p_n)=\tp_2\)
    (i.e., the two reductions use \rn{PR-Prefix} and \rn{PR-RPrefix} respectively,
    and \(p_1\) is principal in both   reductions).\\
    By \(p_2p'_1=p_3^Rp''_1\), \(p_2\) and \(p_3^R\) are in the prefix relation.
    Let us consider the case where \(p_3^R\) is a prefix of \(p_2\);
    the other case is similar.
    In that case, we have:
    \begin{align*}
      & p_2 = p_3^Rp_2' \qquad p''_1 = p_2'p'_1\\
      & \tp_1 = (p'_1,p_3^Rp_2',p_3,\ldots,p_n) \pred (p'_1,p_2',p_3,\ldots,p_n)\\
      & \tp_2=(p_2'p'_1,p_3^Rp_2',p_3,\ldots,p_n)\pred
      (p_2'p'_1,p_2',p_3,\ldots,p_n)\preds(p'_1,p_2',p_3,\ldots,p_n).
    \end{align*}
\end{itemize}

For the existence of characteristic data (Theorem~\ref{th:data-size}),
it suffices to strengthen the requirement for the codewords
\(\alpha_1,\ldots,\alpha_n\), by further requiring that 
for all \(i,j\in\set{1,\ldots,n}\), \(\alpha_i\) is neither a prefix nor a suffix of \(\alpha_j^R\).
Note that we can still choose such codewords such that \(|\alpha_i|=O(\log{n})\).
(For example, for the codewords \(\alpha_1',\ldots,\alpha_n'\) of the same length that satisfy
the original conditions, let \(\alpha_i = a\alpha_ib\).)
Then, we have a slightly weakened version of Lemma~\ref{lem:canonicalsubst},
where ``\(\Dat'(x_i) = \Dat(x_i)\) for every \(x_i\in\FV(\tp)\)'' is replaced with
``\(\Dat'(x_i)\) is \(\Dat(x_i)\) or \((\Dat(x_i))^R\) for every \(x_i\in\FV(\tp)\)'',
which is sufficient for proving Theorem~\ref{th:data-size}.

The lemmas for learnability in Table~\ref{tab:lemmas} remain to hold,
and therefore, Theorems~\ref{th:learnability}
and \ref{th:cstp-lfp} also remain to hold.

For Theorem~\ref{th:satisfiability} (the decidability of
quantifier-free STP formulas), the set of word expressions are extended to
\(w\in (\Alpha\cup \Vars\cup \Vars^R)^*\), and the definition of \(\rho w\) is extended in an obvious
manner.
Then, Theorem~\ref{th:satisfiability} remains to hold,
thanks to the decidability
of word equations with involution~\cite{Diekert06,DBLP:journals/iandc/DiekertJP16}.

\subsubsection{The Extensions with Multisets and Sets}
We now discuss the extensions sketched in Section~\ref{sec:set}.

We first provide more precise definitions of multiset patterns.
The sets of \emph{multiset tuple patterns} (MTPs)
and \emph{conjunctive multiset tuple patterns} (CMTPs), ranged over by \(\tp\) and \(\ctp\) respectively,
are defined by:
\begin{align*}
  & \tp ::= (\pat_1,\ldots,\pat_k) \qquad
   \pat \in (\Alpha\cup \Vars)^* \qquad \gamma ::= \tp_1\land \cdots \land \tp_m. 
\end{align*}
Patterns are identified up to permutations; e.g., \(xyz = yxz\)
(thus, a pattern \(\pat\) may actually be viewed as a multiset consisting of
elements of \(\Alpha\) and \(\Vars\)).
A multiset tuple pattern is \emph{solvable} if
\(\tp \preds (x_1,\ldots,x_k)\) for some mutually distinct variables \(x_1,\ldots,x_k\).
A conjunctive multiset tuple pattern \(\ctp = \tp_1\land \cdots \land \tp_k\)
is \emph{solvable} if \(\tp_i\) is solvable for every \(i\in\set{1,\ldots,k}\).
Solvable multiset tuple patterns and conjunctive multiset tuple patterns are
abbreviated as SMTP and CSMTP respectively.

For a map \(\rho\) from \(\Vars\) to the set of multisets over \(\Alpha\),
and a pattern \(p\) such that \(\FV(p)\subseteq \dom(\rho)\), we define \(\rho p\) by:
\[ \rho \epsilon = \emptyset\qquad
\rho (ap) = \set{a}\uplus (\rho p)\qquad
\rho (xp) = \rho(x) \uplus (\rho p),\]
where \(\uplus\) denotes the multiset union.

For an MTP \(\tp\) and a CMTP \(\ctp\), \(\lang(\tp)\) and \(\lang(\ctp)\)
are defined by:
\begin{align*}
&  \lang(p_1,\ldots,p_n) = \set{(\rho p_1,\ldots,\rho p_n)\mid \dom(\rho)\supseteq
    \Vars(p_1\cdots p_n)}\\
  & \lang(\tp_1\land \cdots \land \tp_m) = \lang(\tp_1)\cap \cdots \cap \lang(\tp_m).
\end{align*}

Among the properties listed in Table~\ref{tab:lemmas},
the minimality property (Theorem~\ref{th:minimality}) fails.
For example, consider the data: \(\dat = (abcd, ab, ac, ccd)\).
We can reduce \((\seq{x},\Subst{\seq{x}}{\dat})\) in the following two ways:
\begin{align*}
  & ((x,y,z,w), \left(\begin{array}{rrrr}
    x & y & z & w\\
    abcd & ab & ac & ccd
  \end{array}
    \right)
    )
    \red 
    ((x'y,y,z,w), \left(\begin{array}{rrrr}
    x' & y & z & w\\
    cd & ab & ac & ccd
    \end{array}
    \right)
    )\\&
    \red 
    ((x'y,y,z,x'w'), \left(\begin{array}{rrrr}
    x' & y & z & w'\\
    cd & ab & ac & c
    \end{array}
    \right)
    )
    \red 
    ((x'y,y,z'w',x'w'), \left(\begin{array}{rrrr}
    x' & y & z' & w'\\
    cd & ab & a & c
    \end{array}
    \right)
    )\\&
    \red 
    ((x'y'z',y'z',z'w',x'w'), \left(\begin{array}{rrrr}
    x' & y' & z' & w'\\
    cd & b & a & c
    \end{array}
    \right)
    )\\&
    \red 
    ((x''w'y'z',y'z',z'w',x''w'w'), \left(\begin{array}{rrrr}
    x'' & y' & z' & w'\\
    d & b & a & c
    \end{array}
    \right)
    )\nred \\&
   ((x,y,z,w), \left(\begin{array}{rrrr}
    x & y & z & w\\
    abcd & ab & ac & ccd
  \end{array}
    \right)
    )
    \red 
    ((x'z,y,z,w), \left(\begin{array}{rrrr}
    x' & y & z & w\\
    bd & ab & ac & ccd
    \end{array}
    \right)
    )\nred \\&
 \end{align*}
Thus, \(\TPinf(\dat)\) outputs
\((xyzw,yz,zw,xww)\) or \((xz,y,z,w)\) non-deterministically,
but the latter is not minimal; in fact,
\(\lang(xwyz,yz,zw,xww) \subsetneq \lang(xz,y,z,w)\).

The failure of the minimality is attributed to the failure of
weak confluence of \(\pred\) (Lemma~\ref{lem:pred-wconf}).
In fact, the reductions
\((xyzw,yz,zw,xww) \pred (xw,yz,zw,xww)\)
and 
\((xyzw,yz,zw,xww) \pred (xy,yz,zw,xww)\)
are not confluent, because \((xy,yz,zw,xww)\npred\).

Due to the failure of minimality, 
Theorem~\ref{th:data-size} (about characteristic data) and
the first part of Theorem~\ref{th:learnability}
(about the learnability of SMTPs) fail.
However, the following weaker version of Theorem~\ref{th:data-size}
(obtained by weakening the condition (iii)) holds.

\begin{theorem}[Characteristic Data (Weaker Version)]
  \label{th:data-size-weak}
  Let \(\tp=(p_1,\ldots,p_n)\) be an STP such that \(|p_1\cdots p_n|=m\).
  Then, there exists \(\dat\) such that (i) \(\sizeof{\dat}\) is polynomial
  in \(m+n\),
  (ii) for any \(\dat'\) such that \(\dat\subseteq \dat'\subseteq \lang(\tp)\),
  there exists \(\Dat\) such that \((\seq{x},\Subst{\seq{x}}{\dat'})\reds (\tp, \Dat)\nred{}\),
  and  (iii) for any \(\dat'\) such that \(\dat\subseteq \dat'\subseteq \lang(\tp)\),
  \((\seq{x},\Subst{\seq{x}}{\dat'})\reds (\tp', \Dat')\nred{}\)
  implies \(\lang(\tp)\subseteq\lang(\tp')\).
  Furthermore, given \(\tp\), \(\dat\) can be constructed in polynomial time. \end{theorem}
\begin{proof}[Proof Sketch]
  For an SMTP \(\tp\), we construct the characteristic data \(\dat^\tp\)
  as follows.
  Let \(\tp = (p_1,\ldots,p_n)\), with \(\FV(p_1\cdots p_n) = \set{x_1,\ldots,x_\ell}\).
  Let \(\rho_i\;(i\in\set{1,\ldots,\ell})\) be the map such that
  \(\rho_i(x_i)=a\) and \(\rho_i(x_j)=\epsilon\) for every \(j\ne i\).
  Let \(\dat^\tp\) be:
  \begin{align*}
    & \left(
    \begin{array}{ccc}
      \rho_1 p_1 & \cdots & \rho_1 p_n\\
      \cdots & \cdots & \cdots \\
      \rho_\ell p_1 & \cdots & \rho_\ell p_n
    \end{array}\right).
  \end{align*}
  For example, the characteristic data for \((xyzw,yz,zw,xww)\) above is
  \begin{align*}
    & \left(
    \begin{array}{cccc}
      a & \epsilon & \epsilon & a\\
      a & a & \epsilon & \epsilon\\
      a & a & a & \epsilon\\
      a & \epsilon & a & aa
    \end{array}\right).
  \end{align*}
  
  Then, \(\dat:=\dat^\tp\) satisfies the conditions required in
  Theorem~\ref{th:data-size-weak}.
    \end{proof}

The class of SMTPs is still learnable if \(\TPinf\) is modified
so that it outputs a minimal STP \(\tp_i\) such that \((\seq{x},\Subst{\seq{x}}{M_i})\reds (\tp_i,\Theta)\nred\).
To do so, it suffices for \(\TPinf\) to first compute
the set \(U' := \set{\tp\mid (\seq{x},\Subst{\seq{x}}{M_i})\reds (\tp_i,\Theta)\nred}\),
filter out non-minimal STPs by computing \(U := \set{\tp\in U' \mid \neg\exists \tp'\in U'.\lang(\tp')\subsetneq \lang(\tp)}\),
and output an element of \(U\) non-deterministically. Then Theorem~\ref{th:learnability} holds for the modified version of \(\TPinf\).

 As for Theorem~\ref{th:decision-problems},
 the decidability remains to hold, but the complexity of
 deciding whether a given multiset tuple pattern \(\tp\) is solvable may not belong to P,
 due to the failure of weak confluence of \(\pred\).
 The membership and inclusion problems can still be solved in
 polynomial time if we are given a witness for the solvability
 of \(\tp\) (i.e., a reduction sequence \(\tp \preds (x_1,\ldots,x_k)\)).

 The other properties---soundness (Theorem~\ref{th:soundness}),
 completeness (Theorems~\ref{th:completeness} and \ref{th:completeness-ctpinf}),
 learnability of CSTPs (the second part of Theorem~\ref{th:learnability} and
 Theorems~\ref{th:cstp-lfp} and \ref{th:stp-lfp}), and the decidability of quantifier-free formulas
 (Theorem~\ref{th:satisfiability})---remain to hold.
 Indeed, we can confirm that
 the required lemmas listed in Table~\ref{tab:lemmas} remain to hold.
 In the proof of Theorem~\ref{th:stp-lfp}, the first line of the procedure \(\mathtt{enumerate}\):
 \[U := \set{\tp \mid (\seq{x},\Subst{\seq{x}}{\dat})\reds (\tp,\Dat)\nred}\]
 should be replaced with:
 \[U' := \set{\tp \mid (\seq{x},\Subst{\seq{x}}{\dat})\reds (\tp,\Dat)\nred};\;
   U := \set{\tp\in U' \mid \neg\exists \tp'\in U'. \lang(\tp')\subsetneq \lang(\tp)},\]
to filter out non-minimal STPs.
     
 The decidability of (the satisfiability of) quantifier-free formulas
 relies for the multiset-case relies on that of linear integer arithmetic,
 instead of word equations.
 First, let us fix the syntax of quantifier-free SMTP formulas:\footnote{We
   can further add primitive constraints accommodated by \cite{multiset}.}
 \begin{definition}
   \label{df:SMTP-formula}
   The set of (quantifier-free) SMTP formulas, ranged over by \(\form\), is defined by:
 \[ \form ::= |w_1|=|w_2| \mid w_1\subseteq w_2 \mid \seq{w}\in \lang(\tp)\mid
 \form_1\land \form_2 \mid \form_1 \lor \form_2\mid \neg \form.\]
 Here, \(w\in (\Alpha\cup\Vars)^*\), where the empty sequence is interpreted as
 the empty multiset, and the concatenation is interpreted as the multiset union.
 \(|w|\) denotes the size of the multiset (expressed by) \(w\).
 \end{definition}
 
 The semantics of quantifier-free SMTP formulas is defined as follows.
 Let \(\rho\) be a map from \(\Vars\) to the set of finite multisets over \(\Alpha\).
 We define \(\rho w\) (where we assume \(\dom(\rho)\supseteq \FV(w)\))
 and \(\rho \models \form\) by:
 \begin{align*}
   & \rho \epsilon = \emptyset\qquad \rho (a w) = \set{a}\uplus \rho w\qquad
   \rho (xw) = \rho(x) \uplus \rho w\\
   & \rho \models |w_1|=|w_2| \mbox{ if }|\rho w_1| = |\rho w_2|\\
   & \rho \models w_1\subseteq w_2 \mbox{ if }\rho w_1 \subseteq \rho w_2\\
   & \rho \models \seq{w}\in \lang(\tp) \mbox{ if }\rho \seq{w}\in \lang(\tp)\\
   & \rho \models \seq{w}\not\in \lang(\tp) \mbox{ if }\rho \seq{w}\not\in \lang(\tp)\\
&  \rho \models \form_1\land \form_2 \mbox{ if } \rho\models \form_1 \mbox{ and }\rho\models \form_2\\
&  \rho \models \form_1\lor \form_2 \mbox{ if } \rho\models \form_1 \mbox{ or }\rho\models \form_2.
&  \rho \models \neg\form \mbox{ if } \rho\models \form \mbox{ does not hold.}
 \end{align*}
 We say that a quantifier-free SMTP formula \(\form\) is \emph{satisfiable}
 if \(\rho \models \form\) for some \(\rho\).
 The following is a multiset version of Theorem~\ref{th:satisfiability}.

\begin{theorem}
  \label{th:satisfiability-multiset}
  Given a quantifier-free SMTP formula \(\form\),
  one can effectively construct an equi-satisfiable SMTP formula \(\form'\) that
  contains no subformulas of the form \(\seq{w}\in \lang(\tp)\), or \(\seq{w}\not\in \lang(\tp)\).
 Therefore,  the satisfiability of quantifier-free SMTP formulas is decidable.
\end{theorem}
\begin{proof}
  It suffices to show encodings of \(\seq{w}\in\lang(\tp)\) and
  \(\seq{w}\not\in\lang(\tp)\) (by considering the negation normal form).
  Let \(\seq{w}=(w_1,\ldots,w_n)\) and \(\tp=(p_1,\ldots,p_\ell)\).
  If \(n\ne \ell\), then
  \(\seq{w}\in\lang(\tp)\) and
  \(\seq{w}\not\in\lang(\tp)\) can be replaced by \(\false\) and \(\true\) respectively.
  Suppose \(n=\ell\).
  Then, \(\seq{w}\in\lang(\tp)\) can be replaced by
  \(w_1=p_1'\land \cdots w_n=p_n'\), where \((p_1',\ldots,p_n')\) is obtained
  by \(\alpha\)-renaming of \(\tp\) and all the variables in
  \((p_1',\ldots,p_n')\) are fresh.
  \(\seq{w}\not\in\lang(\tp)\) can be encoded by induction
  on the sequence \(\tp \preds (x_1,\ldots,x_k)\),  in
  the same manner as the case for STPs.
  For example, if \(\tp = (p_2p'_1,p_2,\ldots)\pred (p'_1,p_2,\ldots)\preds
  (x_1,\ldots,x_k)\), then
  \(\seq{w}\not\in\lang(\tp)\) can be replaced by:
  \(\neg (w_2\subseteq w_1)\lor (w_1=w_2x \land (x,w_2,\ldots,w_n)\not\in (p'_1,p_2,\ldots))\).

  By the decidability of multiset constraints~\cite{multiset},
  we can conclude that the satisfiability of quantifier-free SMTP formulas is decidable.
\end{proof}

Analogous results hold for the case of sets, where the conjunction is interpreted as
the disjoint union. The decidability of quantifier-free SSTP formulas
follows from the decidability of set constraints (e.g., see \cite{DBLP:journals/jar/KuncakNR06}).

 \subsection{Additional Definitions and Proofs for Section~\ref{sec:chc}}
\label{sec:proof-chc}

\subsubsection{Additional Details on Section~\ref{sec:chc-words}}

As defined in Section~\ref{sec:chc-words},
a CHC over words is of the form
\(P_1(\seq{y}_1)\land \cdots \land
P_n(\seq{y}_n)\land \cform \imp H\),
where \(P_i\) is a predicate variable, \(\cform\) is a quantifier-free STP-formula,
and \(H\) is of the form \(P(\seq{y})\) or \(\false\).
(Here, \(\false\) may be viewed as an STP formula \(a=\epsilon\).)

An \emph{interpretation} for predicate variables is
a function that maps each \(k\)-ary predicate variable to a subset of \((\Alpha^*)^k\).
Let \(\pmodel\) be an interpretation for predicates,
and \(\rho\) be a map from \(\Vars\) to \(\Alpha^*\).
We write \(\pmodel, \rho \models P(y_1,\ldots,y_k)\) if \((\rho y_1,\ldots,\rho y_k)\in \pmodel(P)\)
and \(\pmodel,\rho\models \cform\) if \(\rho\models \cform\) (recall Section~\ref{sec:proof-decidability} for the definition of the latter).
Let \(\CHCone\) be a CHC 
\(P_1(\seq{y}_1)\land \cdots \land
P_n(\seq{y}_n)\land \cform \imp H\).
We write  \(\pmodel\models \CHCone\)
if \(\pmodel,\rho\models H\) holds
for every \(\rho\) such that
\(\pmodel,\rho\models \cform\) and 
\(\pmodel,\rho \models P_i(\seq{y}_i)\)
for each \(i\in\set{1,\ldots,n}\).
Let \(\CHC\) be a system of CHCs \(\set{\CHCone_1,\ldots,\CHCone_m}\).
We write \(\pmodel\models \CHC\), and call \(\pmodel\) a \emph{model} of
\(\CHC\) if \(\pmodel\models \CHCone_i\) for every \(\CHCone_i\in\CHC\).
A system \(\CHC\) of CHCs is \emph{satisfiable} if there exists a model of \(\CHC\).

A \emph{CSTP-interpretation} for predicates is 
a map \(\set{P_1\mapsto \cstp_1,\ldots,P_n\mapsto \cstp_n}\)
where \(\cstp_i\)'s are CSTPs.
A {CSTP-interpretation} \(\CSTPmodel = \set{P_1\mapsto \cstp_1,\ldots,P_n\mapsto \cstp_n}\)
is a \emph{CSTP-model} of \(\CHC\), written \(\CSTPmodel\models \CHC\),
if \(\set{P_1\mapsto \lang(\cstp_1),\ldots,P_n\mapsto \lang(\cstp_n)}\models \CHC\).
By abuse of notation, we often just write
\(\set{P_1\mapsto \cstp_1,\ldots,P_n\mapsto \cstp_n}\) for
\(\set{P_1\mapsto \lang(\cstp_1),\ldots,P_n\mapsto \lang(\cstp_n)}\).

\begin{lemma}
  \label{lem:cstp-model}
  Given a CSTP-interpretation \(\CSTPmodel\) and a system \(\CHC\) of CHCs,
  it is decidable whether \(\CSTPmodel\) is a CSTP-model of \(\CHC\).
  Furthermore, there is an algorithm which,
  given \(\CSTPmodel\) and a finite set \(\CHC_D\) of definite clauses,
  outputs ``\(\mathit{None}\)'' if \(\CSTPmodel\models \CHC_D\), and
  otherwise outputs 
  \(Some(P,\seq{s})\) such that there exists a valuation \(\rho\) and a clause
  \(P_1(\seq{y}_1)\land \cdots \land
  P_n(\seq{y}_n)\land \cform \imp P(\seq{y})\) in \(\CHC_D\),
  with \(\CSTPmodel,\rho\models \form\)  and \(\CSTPmodel,\rho \models P_i(\seq{y}_i)\) for every \(i\),
  but \(\rho(\seq{y})=\seq{s}\not\in \CSTPmodel(P)\).
\end{lemma}
\begin{proof}
  This follows immediately from Theorem~\ref{th:satisfiability}.
  Note that \(\CSTPmodel\) is a CSTP-model of a CHC
  \(P_1(\seq{y}_1)\land \cdots \land
  P_n(\seq{y}_n)\land \cform \imp H\), if and only if
  \(\CSTPmodel(P_1(\seq{y}_1))\land \cdots \land
  \CSTPmodel(P_n(\seq{y}_n))\land \cform \land \neg \CSTPmodel(H)\) is unsatisfiable,
  where \(\CSTPmodel(P(\seq{y}))\) denotes the formula \(\seq{y}\in \CSTPmodel(P)\).
  Thus, if \(\CSTPmodel \models \CHC_D\) does not hold, then
  \(\CSTPmodel(P_1(\seq{y}_1))\land \cdots \land\CSTPmodel(P_n(\seq{y}_n))\land \cform \land \neg \CSTPmodel(P(\seq{x}))\) is
  satisfiable for some definite clause \( P_1(\seq{y}_1)\land \cdots \land
  P_n(\seq{y}_n)\land \cform \imp P(\seq{y})\).
  We can construct a model \(\rho\) for the formula and return \(Some(P,\rho(\seq{x}))\).
\end{proof}

To prove Theorem~\ref{th:decidability-chc} for the case where CHCs have
multiple predicates, we need to extend Theorem~\ref{th:cstp-lfp}.
Let \(\mathcal{I}\) be the set of interpretations for predicate variables.
\begin{theorem}
  \label{th:chc-lfp}
  Let \(\CHC_D\) be a finite set of definite clauses consisting of predicates \(P_1,\ldots,P_n\).
  Let \(g\) be the algorithm of Lemma~\ref{lem:cstp-model}.
  Then the procedure below eventually terminates and
  returns the least CSTP-model of \(\CHC_D\).
\begin{algorithm}
  \SetKwProg{Fn}{Function}{}{}
  \SetKwIF{If}{ElseIf}{Else}{if}{}{else if}{else}{end}
  \SetKwBlock{uSwitch}{switch}{}
  \SetKwBlock{uCase}{case}{}
$\pmodel \gets \set{P_1\mapsto \emptyset,\ldots,P_n\mapsto \emptyset}$;
  $\CSTPmodel \gets \set{P_1\mapsto (a,\seq{\epsilon})\land (\epsilon,\seq{\epsilon}),
    \ldots, P_n\mapsto (a,\seq{\epsilon})\land (\epsilon,\seq{\epsilon})}$;\\
\textbf{\upshape while} true \textbf{\upshape do}\\ \quad   \textbf{\upshape if }{$g(\CSTPmodel,\CHC_D)= Some(P,\seq{s})$}\\ \quad \textbf{\upshape then}{      $(\pmodel \gets \pmodel\set{P\mapsto
      \pmodel(P) \cup \{\seq{s}\}};  \CSTPmodel \gets \CSTPmodel\set{P\mapsto
      \CTPinf(\pmodel(P))}$} \\\quad \textbf{\upshape else }{\Return{$\CSTPmodel$}  }
\end{algorithm}
\end{theorem}
\begin{proof}
  We first prove the termination of the procedure.  Suppose that the procedure does not terminate.
  Let \(\pmodel_i, \CSTPmodel_i, \seq{s}_i\) be the values of
  \(\pmodel, \CSTPmodel, \seq{s}\) at the \(i\)-th iteration of the loop.
  We have a strictly increasing infinite sequence
  \(\pmodel_0\subsetneq \pmodel_1\subsetneq \pmodel_2 \subsetneq \cdots\).
For every \(P_i\)
 such that 
  \(\pmodel_0(P_i), \pmodel_1(P_i),\pmodel_2(P_i),\ldots\) strictly increases infinitely often,
  by Lemma~\ref{lem:convergence_of_CTPinf},  there exists \(k_i\) such that
  \(\lang(\CTPinf(\pmodel_j(P_i)))=\lang(\CTPinf(\pmodel_{k_i}(P_i)))\) for all \(j\ge k_i\).
  For every \(P_i\) such that 
  \(\pmodel_0(P_i), \pmodel_1(P_i),\pmodel_2(P_i),\ldots\) strictly increases only finitely often,
  let \(k_i\) be the index such that \(\pmodel_j(P_i))=\pmodel_{k_i}(P_i)\) for all \(j\ge k_i\).
  Let \(k = \max(k_1,\ldots,k_n)\). Then, we have \(\lang(\CSTPmodel_j)=\lang(\CSTPmodel_{k})\) for all \(k\ge j\).
  But then that contradicts \(g(\CSTPmodel_k,\CHC_D)=Some(P_i,\seq{s}_k)\) (which implies
  \(\seq{s}_k\not\in \lang(\CSTPmodel_k(P_i))\)) and
  \(\seq{s}_k\in \pmodel_{k+1}(P_i)\subseteq \lang(\CSTPmodel_{k+1}(P_i))\).

  By the assumption on \(g\),
  when the algorithm terminates at the \(k\)-th step, \(\CSTPmodel_k\) is a model of \(\CHC_D\).
  Suppose there exists another CSTP-model \(\CSTPmodel'\) of \(\CHC_D\).
  Then it follows by induction on \(j\) that \(\CSTPmodel_j \le \CSTPmodel'\) for every \(j\le k\)
  (by the same argument as in the proof of Theorem~\ref{th:cstp-lfp}).
\end{proof}

We can now prove Theorem~\ref{th:decidability-chc}.
\begin{proof}[Proof of Theorem~\ref{th:decidability-chc}.]
    Let \(\CHC = \CHC_D\cup\CHC_G\), where \(\CHC_D\) and \(\CHC_G\) respectively consist of definite and
  goal clauses. 
  By Theorem~\ref{th:chc-lfp}, we can compute the least CSTP-model \(\CSTPmodel\) of
  \(\CHC_D\). Then, \(\CHC\) has a CSTP-model if and only if \(\CSTPmodel\) is also a CSTP-model
  of \(\CHC_G\). (Note that ``only if'' direction holds since
  for any goal clause \(\CHCone\), \(\pmodel'\models \CHCone\) and \(\pmodel\le \pmodel'\)
  imply \(\pmodel\models \CHCone\).)
  The latter is decidable by Theorem~\ref{lem:cstp-model}.
\end{proof}

Theorem~\ref{th:decidability-chc-stp} is an immediate corollary of the following lemma,
which is obtained as a straightforward generalization of Theorem~\ref{th:stp-lfp}.
\begin{lemma}
  \label{th:chc-lfp-stp}
  Let \(\CHC_D\) be a finite set of definite clauses consisting of predicates \(P_1,\ldots,P_n\).
  Then there exists an algorithm that enumerates all the minimal STP-models of \(\CHC_D\).
\end{lemma}
\begin{proof}
      Let \(\mathtt{enumModel}\) be the procedure defined as below, and
      let \(P_1,\ldots,P_n\) be the predicate variables occurring in \(\CHC_D\).
      Let \(g\) be the algorithm of Lemma~\ref{lem:cstp-model}.
      Then it suffices to call
      \(\mathtt{enumModel}(\set{P_1\mapsto \emptyset,\ldots,P_n\mapsto \emptyset})\) to obtain a superset of
    all the minimal STP-models of \(\CHC_D\), and then filter out non-minimal models.
      \begin{algorithm}\small
      \SetKwFunction{Enumerate}{enumModel}
      \SetKwProg{Fn}{Function}{:}{}
      \Fn{\Enumerate{$\pmodel$}}{
{$U$ := $\set{\set{P_1\mapsto \tp_1,\ldots,P_n\mapsto \tp_n}
              \mid {\forall i.(\pmodel(P_i)=\emptyset\land \tp_i=\tpempty) \lor (\seq{x},\Subst{\seq{x}}{\pmodel(P_i)})\reds (\tp_i,\Dat)\nred}$;}}\\
        \While{$U\ne \emptyset$}{
          $\CSTPmodel$ := pick($U$); $U$ := $U\setminus\set{\CSTPmodel}$;\\
          \uIf{$g(\CSTPmodel,\CHC_D)=Some(P,\seq{s})$}
              {$U := U\setminus \set{\CSTPmodel'\in U \mid \seq{s}\in \lang(\CSTPmodel'(P))}$;\\
                \Enumerate($\pmodel\set{P\mapsto \pmodel(P)\cup\set{\seq{s}}}$);}
              \uElse
                  {$\mathit{output}$($\tp$);}
          }
        }
    \end{algorithm}\\
The correctness of the algorithm follows by the same argument as that of Theorem~\ref{th:stp-lfp}.
\end{proof}

  \begin{example}
    Let us apply the algorithm in the proof above to the \(\Reva{}\) example
    (cf. Example~\ref{ex:solving-reva} and Remark~\ref{rem:CSTPvsSTP}).
    The computation of \(\enumM(\set{\Reva\mapsto \tpempty})\) proceeds as follows.
    \begin{enumerate}
      \item     \(U\) is first set to \(\set{\set{\Reva\mapsto \tpempty}}\) and \(\CSTPmodel_0 :=\set{\Reva\mapsto \tpempty}\) is picked.
    Suppose \(g(\CSTPmodel_0,\CHC_D)\) returns \(Some(\Reva,(\epsilon,ab,ab))\).
  \item     Then \(\enumM(\set{\Reva\mapsto \set{(\epsilon,ab,ab)}})\) is recursively called, and
    \(U\) is set to \(\set{\set{\Reva\mapsto (\epsilon,x,x)}}\), and \(\CSTPmodel_1 := \set{\Reva\mapsto (\epsilon,x,x)}\) is picked.
    Now \(g(\CSTPmodel_1,\CHC_D)\) may return \(Some(\Reva,(a,b,ab))\).
    \item Then \(\enumM(\set{\Reva\mapsto \set{(\epsilon,ab,ab),(a,b,ab)}})\) is recursively called, and
      \(U\) is set to \(\set{\set{\Reva\mapsto (x,y,xy)},\set{\set{\Reva\mapsto (x,y,x^Ry)}}}\).
      Suppose \(\CSTPmodel_{2,1} := \set{\Reva\mapsto (x,y,xy)}\) is picked and \(U\) is updated to
      \(\set{\set{\Reva\mapsto (x,y,x^Ry)}}\) (inside this recursive call).
      \(g(\CSTPmodel_{2,1},\CHC_D)\) may return \(Some(\Reva,(cab,\epsilon,abc))\).
    \item  Now, \(\enumM(\set{\Reva\mapsto \set{(\epsilon,ab,ab),(a,b,ab),(cab,\epsilon,abc)}})\) is recursively called, and
      \(U\) is set to \(\set{\set{\Reva\mapsto (x,y,zy)}}\), and \(\CSTPmodel_{3}:=\set{\Reva\mapsto (x,y,zy)}\) is picked.
      Now \(g(\CSTPmodel_3,\CHC_D)\) returns \(\None\). So, the algorithm outputs \(\CSTPmodel_3\) and returns to the middle of Step 3 above.
    \item From \(U=\set{\set{\Reva\mapsto (x,y,x^Ry)}}\), \(\CSTPmodel_{2,2}:= \set{\Reva\mapsto (x,y,x^Ry)}\) is picked.
      Since \(g(\CSTPmodel_{2,2},\CHC_D)\) returns \(\None\), the algorithm outputs \(\CSTPmodel_{2,2}\).
      Now \(U\) becomes empty in chain of recursive calls, and the algorithm terminates.
    \end{enumerate}
    Above, we have obtained \(\set{\CSTPmodel_3,\CSTPmodel_{2,2}}\) as models of \(\CHC_D\). By filtering out the non-minimal model
    \(\CSTPmodel_3\), we obtain \(\set{\CSTPmodel_{2,2}}\) as the set of minimal models of \(\CHC_D\).
    Since \(\CSTPmodel_{2,2}=\set{\Reva\mapsto (x,y,x^Ry)}\) is also a model of \(\CHC_G\), we can conclude that \(\CHC_D\cup\CHC_G\)
    is satisfiable.
    \qed
  \end{example}

\subsubsection{Additional Details on Section~\ref{sec:chc-multisets}}

A \emph{CHC over multisets} (CHC over sets, resp.) is of the form
\(P_1(\seq{y}_1)\land \cdots \land
P_n(\seq{y}_n)\land \cform \imp H\),
where \(P_i\) is a predicate variable, \(\cform\) is a quantifier-free SMTP-formula
(SSTP-formula) as defined in Definition~\ref{df:SMTP-formula},
and \(H\) is of the form \(P(\seq{y})\) or \(\false\).

The semantics and satisfiability of CHCs over multisets (sets, resp.) is defined in the same manner
as those of CHCs over words, except that \(\rho\) maps each variable to a multiset (a set, resp.)
over \(\Alpha\)
and \(\pmodel\) maps each \(k\)-ary predicate to a \(k\)-ary relation on multisets (sets, resp.)
over \(\Alpha\).
Theorems~\ref{th:decidability-chc-set} and \ref{th:decidability-chc-set-stp} follow by the same argument as 
Theorems~\ref{th:decidability-chc} and Theorems~\ref{th:decidability-chc-stp}.

\begin{example}
 Recall the following CHCs over multisets.
  \begin{align*}
  & |x|=1\imp \Pred{Insert}(x, \epsilon, x).     \qquad l=yl' \land |x|=|y|=1\imp \Pred{Insert}(x,l,xl).  \\  &
  l=yl' \land |x|=|y|=1\land \Pred{Insert}(x,l',r)\imp \Pred{Insert}(x,l,yr).\\
  & \Pred{Sort}(\epsilon, \epsilon). \qquad |x|=1\land l=xl'\land \Pred{Sort}(l', r')\land \Pred{Insert}(x,r',r)\imp \Pred{Sort}(l,r). \qquad \Pred{Sort}(l,r)
    \imp l=r. \end{align*}
  We can compute the least CMTP-model for the definite clauses (i.e., the above CHCs
  except the last one) as follows. We write \(\CMTPinf\) for the multiset version of
  \(\CTPinf\), and disable the rule \rn{R-CSubset} for simplicity.

  We first set:
  \begin{align*}
    &   \pmodel_0 = \set{\Pred{Insert} \mapsto \emptyset, \Pred{Sort}\mapsto \emptyset}\\
    & \CSTPmodel_0 =
    \set{\Pred{Insert} \mapsto (a,\epsilon,\epsilon)\land(\epsilon,\epsilon,\epsilon),
      \Pred{Sort}\mapsto (a,\epsilon)\land(\epsilon,\epsilon)}.
  \end{align*}
  By calling \(g(\CSTPmodel_0,\CHC_D)\), we may obtain \(Some(\Pred{Insert},(a,\epsilon,a))\).
  As \(\CMTPinf(\set{(a,\epsilon,a)})=(x,\epsilon,x)\), we now have:
  \begin{align*}
    &   \pmodel_1 = \set{\Pred{Insert} \mapsto \set{(a,\epsilon,a)}, \Pred{Sort}\mapsto \emptyset}\\
    & \CSTPmodel_1 =
    \set{\Pred{Insert} \mapsto (x,\epsilon,x),
      \Pred{Sort}\mapsto (a,\epsilon)\land(\epsilon,\epsilon)}.
  \end{align*}
  By calling \(g(\CSTPmodel_1,\CHC_D)\), we may obtain \(Some(\Pred{Sort},(\epsilon,\epsilon))\).
  We now have:
  \begin{align*}
    &   \pmodel_2 = \set{\Pred{Insert} \mapsto \set{(a,\epsilon,a)}, \Pred{Sort}\mapsto (\epsilon,\epsilon)}\\
    & \CSTPmodel_2 =
    \set{\Pred{Insert} \mapsto (x,\epsilon,x),
      \Pred{Sort}\mapsto (\epsilon,\epsilon)}.
  \end{align*}
  By calling \(g(\CSTPmodel_2,\CHC_D)\), we may get \(Some(\Pred{Insert},(a,b,ab))\) and:
  \begin{align*}
    &   \pmodel_3 = \set{\Pred{Insert} \mapsto \set{(a,\epsilon,a),(a,b,ab)}, \Pred{Sort}\mapsto (\epsilon,\epsilon)}\\
    & \CSTPmodel_3 =
    \set{\Pred{Insert} \mapsto (x,y,xy),
      \Pred{Sort}\mapsto (\epsilon,\epsilon)}.
  \end{align*}
  By calling \(g(\CSTPmodel_3,\CHC_D)\), we may get \(Some(\Pred{Sort},(a,a))\).
  We now have:
  \begin{align*}
    &   \pmodel_4 = \set{\Pred{Insert} \mapsto \set{(a,\epsilon,a),(a,b,ab)}, \Pred{Sort}\mapsto (\epsilon,\epsilon), (a,a)}\\
    & \CSTPmodel_4 =
    \set{\Pred{Insert} \mapsto (x,y,xy),
      \Pred{Sort}\mapsto (x,x)}.
  \end{align*}
  At this point, \(g(\CSTPmodel_4,\CHC_D)\) returns \(\None\).
  Thus, we have obtained the least CSMTP-model \(\CSTPmodel_4 =
  \set{\Pred{Insert} \mapsto (x,y,xy),
    \Pred{Sort}\mapsto (x,x)}\).
  Since \(\CSTPmodel_4\) is a model of the goal clause
  \(\Pred{Sort}(l,r)
    \imp l=r\), we can conclude that \(\CHC\) is satisfiable. \qed
\end{example}
 \section{Verification of Functional Queues Using Piecewise Conjunctive Solvable Tuple Patterns}
\label{sec:pcstp}

We have considered conjunctive STPs, but not arbitrary Boolean combinations of
STPs. While inferring arbitrary Boolean combinations seems infeasible,
we can allow a restricted form of them called \emph{piecewise} CSTPs.
Let us fix a finite set \(\mathcal{P}\) of predicates of the form \(\lambda \seq{x}.\form\) where \(\form\) is
a (quantifier-free) STP-formula such that \(\FV(\form)\subseteq \set{\seq{x}}\).
We call \(|\seq{x}|\) the arity of the predicate \(\lambda \seq{x}.\form\), and write
\(\arity(\lambda \seq{x}.\form)\) for it. We sometimes omit ``\(\lambda \seq{x}.\)'' when it is clear from
context.
A (\(k\)-ary) \emph{piecewise CSTPs} over \(\mathcal{P}\) is an expression of the form
\[ (p_1\imp \ctp_1)\land \cdots \land (p_\ell\imp \ctp_\ell),\]
where \(p_1,\ldots,p_\ell\) are \(k\)-ary predicates, and \(\ctp_1,\ldots,\ctp_\ell\) are \(k\)-ary CSTPs.
The language represented by a \(k\)-ary piecewise CSTPs is defined by:
\[ \lang((p_1\imp \ctp_1)\land \cdots \land (p_\ell\imp \ctp_\ell))
= \bigcap_{i\in\set{1,\ldots,\ell}} \set{\seq{s} \mid p_i(\seq{s}) \imp \seq{s}\in\lang(\ctp_i)
  }.
  \]
Here, for \(p_i=\lambda \seq{x}.\form_i\), \(p_i(\seq{s})\) denotes \([\seq{s}/\seq{x}]\form_i\).
For example, if \(\mathcal{P}=\set{p_1,p_2}\) where
 \(p_1= \lambda (x_1,x_2,x_3).x_1=\epsilon\) and \(p_2=\lambda (x_1,x_2,x_3).x_1\ne \epsilon\).
then \((p_1\imp (x, y, \epsilon))\land (p_2\imp (x, y, y))\)
is a piecewise CSTP. It represents the set of triples of the form \((\epsilon,s_2,\epsilon)\)
or \((s_1,s_2,s_2)\) where \(s_1\ne \epsilon\).

Typically, the set \(\mathcal{P}\) of predicates would be
\(\set{\lambda \seq{x}.x_i=\epsilon \mid i\in\set{1,\ldots,k}}\cup
\set{\lambda \seq{x}.x_i\ne\epsilon \mid i\in\set{1,\ldots,k}}\).
CSTPs can be considered a special case of piecewise CSTPs where \(\mathcal{P}=\set{\true}\).

Theorem~\ref{th:decidability-chc} can be extended for piecewise CSTPs.
\begin{theorem}
  \label{th:decidability-piecewisechc}
   Given a finite set \(\mathcal{P}\) of predicates as defined above and 
   a system \(\CHC\) of CHCs on words, it is decidable whether \(\CHC\) has
  a piecewise CSTP-model over \(\mathcal{P}\).
\end{theorem}
\begin{proof}
This follows from the fact that \(\CHC\) has a piecewise CSTP-model if and only if
\(\CHC'\) obtained from \(\CHC\) by replacing each \(k\)-ary predicate
\(P(x_1,\ldots,x_k)\) with
\((p_1(x_1,\ldots,x_k)\imp P_1(x_1,\ldots,x_k))\land \cdots \land
(p_\ell(x_1,\ldots,x_k)\imp P_\ell(x_1,\ldots,x_k))\)
has a CSTP-model. Here, \(P_i\)'s are fresh predicates and \(\set{p_1,\ldots,p_\ell}\)
is the set of \(k\)-ary predicates in \(\mathcal{P}\).
\end{proof}

Below we apply the result above to the verification of functional queues.
\begin{example}
  \label{ex:queue}
  Let us consider the following OCaml program.
  \begin{quote}
\begin{verbatim}
let initq = ([], [])
let rec reva l1 l2 =
  match l1 with
    [] -> l2
  | x::l1' -> reva l1' (x::l2)
let enq x (l1,l2) = (l1, x::l2)
let deq (l1,l2) =
  match l1 with
    [] ->
     (match reva l2 [] with
        [] -> ([], l1, l2)
      | x::l -> ([x], l, [])
     )
  | x::l1' -> ([x], l1', l2)
let rec enqall l (l1,l2) =
  match l with
    [] -> (l1,l2)
  | x::l' -> enqall l' (enq x (l1,l2))
let rec deqall (l1,l2) =               
  match deq(l1,l2) with
    ([], _, _) -> []
  | (l, l1',l2') -> l@(deqall (l1', l2'))
let main l =
  assert(deqall (enqall l initq) = l)
\end{verbatim}
  \end{quote}
Here, a queue is implemented as a pair of lists \((l_1,l_2)\) where \(l_1l_2^R\) is the sequence of
elements in the queue, so that
the amortized cost of each enqueue or dequeue operation is \(O(1)\)~\cite{FuncQueue}.
The dequeue function \verb|deq| either returns \([\,]\) (when the queue is empty, i.e., \(l_1=l_2=[\,]\))
or a singleton set \([x]\) consisting of the first element \(x\) of the queue, along with the updated
queue.
If \(l_1=[\,]\), then the first element of the queue is computed by extracting the first element of
the reverse of \(l_2\).
Given a list \(l\) of elements as input,
the main function enqueues all the elements of \(l\), dequeues all the elements from the queue,
and then asserts that the result equals \(l\).

The correctness of the above program (i.e., the lack of assertion failures) is reduced to
the satisfiability problem for the following CHCs.
\newcommand\Deq{\mathit{Deq}}
\newcommand\Enq{\mathit{Enq}}
\newcommand\Deqa{\mathit{DeqAll}}
\newcommand\Enqa{\mathit{EnqAll}}
\begin{align*}
  & |x|=1\imp \Enq(x,l_1,l_2,xl_1,l_2).\\
&  \Deq(\epsilon,\epsilon,\epsilon,\epsilon,\epsilon).\\
  &  |x|=1\land \Reva(l_2,\epsilon,xl)\imp \Deq(\epsilon,l_2,x,l,\epsilon).\\
  & |x|=1\imp \Deq(xl_1', l_2, x, l_1', l_2).\\
  & \Enqa(\epsilon, l_1,l_2,l_1,l_2).\\
  & l=xl'\land |x|=1\land\Enq(x, l_1,l_2,l_1',l_2')\land  \Enqa(l', l_1',l_2', r_1,r_2) \imp
  \Enqa(l, l_1,l_2,r_1,r_2).\\
  & \Deq(l_1,l_2,\epsilon,l_1',l_2')\imp \Deqa(l_1,l_2,\epsilon).\\
  & \Deq(l_1,l_2,l,l_1',l_2')\land l\ne \epsilon\land \Deqa(l_1',l_2',r)
  \imp \Deqa(l_1,l_2,lr).\\
  & \Enqa(l, \epsilon, \epsilon,l_1,l_2)\land \Deqa(l_1,l_2,r)\imp r=l.
\end{align*}
Here, each predicate represents the relation between inputs and outputs of the corresponding function.
For example, \(\Enq(x,l_1,l_2,r_1,r_2)\) means that \(\mathtt{enq}\;x\;(l_1,l_2)\) may return \((r_1,r_2)\).
The clauses for \(\Reva\) (which are found in Section~\ref{sec:chc-words}) have been omitted.

The system of CHCs above has the following piecewise CSTP-model:
\begin{align*}
  \{\quad&\\  &
  \Reva\mapsto (l_1,l_2,l_1^Rl_2),\\&
  \Enq\mapsto (l, l_1,l_2, l_1,l^Rl_2),\\&
  \Deq\mapsto ((\lambda (l_1,l_2,x,r_1,r_2).l_1=l_2=\epsilon) \imp (\epsilon,\epsilon,\epsilon,\epsilon,\epsilon))\\&
  \qquad\ \land (\lambda (l_1,l_2,x,r_1,r_2).l_1=\epsilon\land l_2\ne \epsilon) \imp
  (\epsilon,lx,x^R,l^R,\epsilon))\\&
  \qquad\ \land (\lambda (l_1,l_2,x,r_1,r_2).l_1\ne\epsilon\land l_2\ne \epsilon) \imp
  (xl_1,l_2,x,l_1,l_2)),\\&
  \Enqa\mapsto (l, l_1,l_2, l_1, l^Rl_2),\\&
  \Deqa\mapsto (l_1,l_2,l_1l_2^R)\\
  \}.\quad
\end{align*}
Thus, if we set \(\mathcal{P}\) to
\(\set{p_1,p_2,p_3,p_4,p_5}\)  where:
\begin{align*}
& p_1 = \lambda (l_1,l_2,l_3,l_4,l_5).l_1=l_2=\epsilon\\
& p_2 = \lambda (l_1,l_2,l_3,l_4,l_5).l_1=\epsilon \land l_2\ne\epsilon\\
  & p_3 = \lambda (l_1,l_2,l_3,l_4,l_5).l_1\ne\epsilon \land l_2\ne\epsilon\\
  & p_4 = \lambda (l_1,l_2,l_3).\true\\
  & p_5 = \lambda (l_1,l_2,l_3,l_4,l_5).\true,
\end{align*}
then we can automatically prove the satisfiability of the CHCs above.
We expect that necessary predicates in \(\mathcal{P}\) can typically be mined from
the source program; in fact, in the above case, the function \texttt{deq} performs case analysis on whether \(l_1\) and \(l_2\) are
empty lists.
A general method to find an appropriate set \(\mathcal{P}\) of predicates is left for future work.\qed
\end{example}
 \end{document}